\ifx\arxiv\undefined

\documentclass[12pt]{article}

    \usepackage[verbose,tmargin=1in,bmargin=1in,lmargin=1in,rmargin=1in]{geometry}

\else 

\documentclass[12pt]{article}

    \usepackage[verbose,tmargin=1in,bmargin=1in,lmargin=1in,rmargin=1in]{geometry}

\fi

\usepackage[utf8]{inputenc}

\usepackage[unicode=true,pdfusetitle,
    bookmarks=true,bookmarksnumbered=false,bookmarksopen=false,
breaklinks=true,pdfborder={0 0 0},backref=false,colorlinks=true]{hyperref}
\hypersetup{linkcolor=blue,citecolor=blue,urlcolor=blue}

\usepackage{amsmath,amsthm,amsfonts}
\usepackage{xrf}
\usepackage{slashed}
\usepackage{enumerate}
\usepackage{xcolor}

\usepackage{mathabx}
\usepackage{bbm}
\usepackage[normalem]{ulem}

\newcommand{\ed}{\color{black}}
\newcommand{\old}[1]{\ifx\showold\undefined\else\color{gray}\ \sout{#1}\ \ed\fi}

\newcommand{\R}{{\mathbb R}}
\newcommand{\N}{{\mathbb N}}
\newcommand{\Z}{{\mathbb Z}}

\renewcommand{\C}{{\mathbb C}}
\newcommand{\cC}{\mathcal{C}}
\newcommand{\cF}{\mathcal{F}}
\newcommand{\cG}{\mathcal{G}}
\newcommand{\cH}{\mathcal{H}}
\newcommand{\cM}{\mathcal{M}}

\renewcommand{\vec}[1]{{\mathbf{#1}}}

\newcommand{\pol}{\mathtt{Pol}}
\newcommand{\class}{{\cC}}

\newcommand{\im}{\operatorname{Im}}
\newcommand{\re}{\operatorname{Re}}
\newcommand{\out}{{\operatorname{out}}}
\newcommand{\inn}{{\operatorname{in}}}

\newcommand{\id}{\operatorname{id}}

\newcommand{\sk}[1]{\left\langle #1\right\rangle}

\newcommand{\causal}{\operatorname{Causal}}
\newcommand{\past}{\operatorname{Past}}
\newcommand{\konv}[1]{\stackrel{#1}{\longrightarrow}}

\newcommand{\supp}{\operatorname{supp}}
\newcommand{\selfmaps}{\righttoleftarrow}

\newcommand{\CSigma}{{\class_\Sigma}}

\newcommand{\CA}{{\cC_A}}

\newcommand{\HSigma}{\cH_\Sigma}

\newcommand{\HA}{\cH_A}
\newcommand{\Vmax}{{V_{\operatorname{max}}}}
\newcommand{\GammaMax}{{\Gamma_{\operatorname{max}}}}
\newcommand{\cD }{C_D}

\newtheorem{fact}{Fact}[section]

\newtheorem{theorem}[fact]{Theorem}

\newtheorem{corollary}[fact]{Corollary}

\newtheorem{definition}[fact]{Definition}

\newtheorem{lemma}[fact]{Lemma}

\newtheorem{remark}[fact]{Remark}

\begin{document}

\title{\textsc{External Field QED on Cauchy Surfaces\\ for varying Electromagnetic Fields}}

\author{D.-A. Deckert and F. Merkl\\
    \small Mathematisches Institut der Ludwig-Maximilians-Universit\"at M\"unchen\\
    \small Theresienstr. 39, 80333 M\"unchen, Germany\\
    \small deckert@math.lmu.de, merkl@math.lmu.de
}

\maketitle

\begin{abstract}
    The Shale-Stinespring Theorem (1965) together with Ruijsenaar's criterion
    (1977) provide a necessary and sufficient condition for the implementability
    of the evolution of external field quantum electrodynamics between
    constant-time hyperplanes on standard Fock space.  The assertion states that
    an implementation is possible if and only if the spatial components of the
    external electromagnetic four-vector potential $A_\mu$ are zero. We
    generalize this result to smooth, space-like Cauchy surfaces and, for
    general $A_\mu$, show how the second-quantized Dirac evolution can always be
    implemented as a map between varying Fock spaces.  Furthermore, we give
    equivalence classes of polarizations, including an explicit representative,
    that give rise to those admissible Fock spaces.  We prove that the
    polarization classes only depend on the tangential components of $A_\mu$
    w.r.t.\ the particular Cauchy surface, and show that they behave naturally
    under Lorentz and gauge transformations.
\end{abstract}

\section{Introduction and Setup}

We consider the external field model of quantum electrodynamics (QED) or
no-photon QED which describes a Dirac sea of electrons evolving subject to a
prescribed external electromagnetic four-vector 
potential $A_\mu$. To infer the evolution operator
of this model one attempts to implement the one-particle Dirac evolution 
\begin{align}
    \label{eq:dirac-equation}
    (i\slashed\partial-\slashed A) \psi=m\psi
\end{align}
in second-quantized form. Here, $m>0$ denotes the mass of the electron; the
elementary charge of the electron $e$ (having a negative sign in the case of an
electron) is already absorbed in $A$; units are chosen such that $\hbar=1$
and $c=1$. The employed
relativistic notation is introduced with all other notations in
Section~\ref{sec:preliminaries}.  For sake of simplicity we will restrict us to
smooth and compactly supported $A_\mu$, i.e.,
\begin{align}
    \label{def-A}
    A=(A_{\mu})_{\mu=0,1,2,3}=(A_0,\vec A)
    \in C^{\infty}_{c}(\mathbb
    R^{4},\mathbb R^{4}),
\end{align} 
although this condition is unnecessarily strong. 

It is well-known \cite{Shale1965,Ruijsenaars1977a} that, on standard Fock space
and for equal-time hyperplanes, a second quantization of the one particle Dirac
evolution \eqref{eq:dirac-equation} is possible if and only if $\vec A=0$, i.e.,
the spatial components of the external field vanish -- a condition that appears
strange in view of gauge invariance. In physics the ill-definedness of the
evolution operator and its generator for general vector potentials $A$ is
usually ignored at first which later manifests itself in the appearance of
infinities in informal perturbation series. Those infinities have to be taken
out by hand or, as for example in the case of the vacuum expectation value of
the charge current, absorbed in the coefficient of the electron charge.
Nevertheless, since the sole interaction arises only from a prescribed
four-vector field one may rather expect that it should be possible to control
the time evolution non-perturbatively.  One way to construct a well-defined
second-quantized time evolution operator, as sketched in \cite{Fierz1979}, is to
implement it between time-varying Fock spaces.  Such constructions have been
carried out, e.g., in \cite{Langmann1996,Mickelsson1998,Deckert2010a}. 
While the
idea of changing Fock spaces might be unfamiliar as seen from the
non-relativistic setting, in a relativistic formulation it is to be expected. A
Lorentz boost for instance may tilt an equal-time hyperplane to a space-like
space-like hyperplane $\Sigma$, which requires a change from standard Hilbert
space $L^2(\R^3,\C^4)$ to one attached to $\Sigma$, and likewise, for the
corresponding Fock spaces. 

In this work we extend the existing constructions in
\cite{Langmann1996,Mickelsson1998,Deckert2010a}, which deal exclusively with
equal-time hyperplanes, by implementing the second-quantized Dirac evolution
from one Cauchy surface to another.  The
resulting formulation of external field QED has several advantages:
1) Its Lorentz and gauge covariance can be made explicit; 2) as it treats
the initial value problem for general Cauchy surfaces it allows to study
the evolution in the form of local deformations of Cauchy surfaces in the spirit
of Tomonaga and Schwinger, e.g., \cite{Tomonaga1946,Schwinger1951a}; 3) it gives
a geometric and more general version of the implementability condition $\vec
A=0$ that was found in the special case of equal-time hyperplanes.

Before presenting our main results in Section~\ref{sec:main} we outline the
construction of the evolution operator for general space-like Cauchy surfaces.
Given a Cauchy surface $\Sigma$ in Minkowski space-time (see
Definition~\ref{def:cauchy-surface} below), the states of the Dirac sea on
$\Sigma$ are represented by vectors in a conveniently chosen Fock space, here,
denoted by the symbol $\cF(V,\HSigma)$. In this notation $\HSigma$ is the
Hilbert space of $\C^4$-valued, square integrable functions on $\Sigma$ (see
Definition~\ref{def:HSigma} below) and $V\in\pol(\HSigma)$ is one of its
polarizations:
\begin{definition}
    Let $\pol(\HSigma)$ denote the set of all closed, linear subspaces
    $V\subset \HSigma$ such that $V$ and $V^\perp$ are both infinite
    dimensional. Any $V\in \pol(\HSigma)$ is called a \emph{polarization} of
    $\HSigma$.  For $V\in \pol(\HSigma)$, let $P_\Sigma^V:\HSigma\to V$
    denote the orthogonal projection of $\HSigma$ onto $V$.
\end{definition}
The Fock space corresponding to polarization $V$ on Cauchy surface $\Sigma$ is
then defined by
\begin{align}
    \label{eq:Fock-space}
    \cF(V,\HSigma) 
    :=
    \bigoplus_{c\in\Z} \cF_c(V,\HSigma),
    \qquad
    \cF_c(V,\HSigma) 
    :=
    \bigoplus_{\substack{n,m\in\N_0\\c=m-n}} (V^\perp)^{\wedge n} \otimes \overline{V}^{\wedge m},
\end{align}
where $\bigoplus$ denotes the Hilbert space direct sum, $\wedge$ the
antisymmetric tensor product of Hilbert spaces, and $\overline{V}$ denotes the
conjugate complex vector space of $V$, which coincides with $V$ as a set and
has the same vector space operations as
$V$ with the exception of the scalar multiplication, which is redefined by
$(z,\psi)\mapsto z^* \psi$ for $z\in\C$, $\psi\in V$.

Each polarization $V$ splits the Hilbert space $\HSigma$ into a direct sum, i.e., 
$\HSigma=V^\perp\oplus V$. The so-called standard polarizations
$\HSigma^+$ and $\HSigma^-$ are
determined by the orthogonal projectors $P_\Sigma^+$ and $P_\Sigma^-$ 
onto the free positive and negative energy Dirac solutions, respectively, restricted to
$\Sigma$:
\begin{align}
    \HSigma^+ := P_\Sigma^+ \HSigma = (1-P_\Sigma^-)\HSigma,
    \qquad
    \HSigma^- := P_\Sigma^- \HSigma.
\end{align}
Loosely speaking, in terms of Dirac's hole theory, the polarization
$V\in\pol(\HSigma)$ indicates the ``sea level'' of the Dirac sea, and
electron wave functions in $V^\perp$ and $V$ are considered to be
``above'' and ``below'' sea level, respectively. However, it has to be stressed
that the mathematical structure of the external field problem in QED does not
seem to discriminate between particular choices of polarizations $V$.  Hence,
unless an additional physical condition is delivered, the $V$-dependent labels
``electron'' and ``positron'' are somewhat arbitrary, and $V$ should rather be
regarded as a choice of coordinate system w.r.t.\ which the states of the Dirac
sea are represented. To describe pair-creation on the other hand it is necessary
to have a distinguished $V$, and the common (and seemingly most natural) ad hoc
choice in situations when the external field vanishes is $V=\HSigma^-$.
Nevertheless, it is conceivable that only a yet to be found full version of QED,
including the interaction with the photon field, may distinguish particular
polarizations $V$ in general situations.

Given two Cauchy surfaces $\Sigma,\Sigma'$ and two polarizations
$V\in\pol(\HSigma)$ and
$W\in\pol(\cH_{\Sigma'})$ a sensible lift of the one-particle Dirac evolution
$U_{\Sigma'\Sigma}^A:\HSigma\to\cH_{\Sigma'}$ (see Definition~\ref{def:UA})
should be given by a unitary operator
$\widetilde U_{\Sigma'\Sigma}^A:\cF(V,\HSigma)\to\cF(W,\cH_{\Sigma'})$ that fulfills
\begin{align} 
    \widetilde U_{\Sigma'\Sigma}^A \, \psi_{V,\Sigma}(f) \, (\widetilde
    U_{\Sigma'\Sigma}^A)^{-1} =
    \psi_{W,\Sigma'}(U_{\Sigma'\Sigma}^Af),
    \qquad
    \forall\,f\in\HSigma.
    \label{eq:lift-condition}
\end{align}
Here, $\psi_{V,\Sigma}$ denotes the Dirac field operator corresponding to Fock space
$\cF(V,\Sigma)$, i.e.,
\begin{align}
    \psi_{V,\Sigma}(f):=b_\Sigma(P_\Sigma^{V^\perp}f) +
    d_\Sigma^*(P_\Sigma^V f),
    \qquad
    \forall \, f\in \HSigma.
\end{align}
Here, $b_\Sigma, d^*_\Sigma$ denote the annihilation and creation operators on
the $V^\perp$ and $\overline V$ sectors of $\cF_c(V,\HSigma)$, respectively.
Note that $P_\Sigma^V:\cH\to\overline V$ is \emph{anti-linear}; thus,
$\psi_{V,\Sigma}(f)$ is anti-linear 
in its argument $f$.  The condition under which
such a lift $\widetilde U_{\Sigma'\Sigma}^A$ exists can be inferred from a straight-forward application of Shale
and Stinespring's well-known theorem \cite{Shale1965}:
\begin{theorem}[Shale-Stinespring]
    \label{thm:Shale-Stinespring}
    The following statements are equivalent:
    \begin{enumerate}
        \item There is a unitary operator $\widetilde
            U_{\Sigma\Sigma'}^A:\cF(V,\HSigma)\to\cF(W,\cH_{\Sigma'})$ which fulfills
            \eqref{eq:lift-condition}.
        \item The off-diagonals $P_{\Sigma'}^{ {W}^\perp}
            U_{\Sigma'\Sigma}^A P_\Sigma^V$ and
            $P_{\Sigma'}^W U_{\Sigma'\Sigma}^A
            P_\Sigma^{V^\perp}$ are Hilbert-Schmidt operators.
    \end{enumerate}
\end{theorem}
Note that the phase of the lift is not fixed by condition
\eqref{eq:lift-condition}. Even worse, as indicated earlier, 
depending on the external field $A$ this condition is not always satisfied; see
\cite{Ruijsenaars1977a}. On the other hand, the choices made for the
polarizations $V$ and $W$ were completely arbitrary. We shall see next that
adapting these choices
carefully will however yield an evolution of the Dirac sea in the corresponding Fock space
representations.

There is a trivial but not so useful choice. Pick a $\Sigma_\inn$ in the
remote past of the support of $A$ fulfilling
\begin{align}
    \label{eq:Sigma_inn}
    \Sigma_\inn \text{ is a Cauchy surface such that }
    \supp A\cap\Sigma_\inn=\emptyset.
\end{align}
Then the choices $V=U_{\Sigma\Sigma_\inn}^A\cH^-_{\Sigma_\inn}$ and
$W=U_{\Sigma'\Sigma_\inn}^A\cH^-_{\Sigma_\inn}$ fulfill (b) of Theorem~\ref{thm:Shale-Stinespring} as
the off-diagonals are zero. The drawback of these choices is that the
resulting lift depends on the whole history of $A$ between $\Sigma_\inn$ and
$\Sigma,\Sigma'$. Moreover, such $V$ and $W$ are rather implicit.
But statement (b) in Theorem~\ref{thm:Shale-Stinespring} also allows to
differ from the projectors $P_\Sigma^V$ and $P_{\Sigma'}^W$ by a Hilbert-Schmidt
operator.  Hence, it lies near to characterize polarizations according to the
following classes:
\begin{definition}[Physical Polarization Classes]
    \label{def:pol-classes}
    For a Cauchy surface $\Sigma$ we define
    \begin{align}
        \label{eq:def-CA}
        \class_\Sigma(A)
        :=
        \left[ U_{\Sigma\Sigma_\inn}^A \cH_{\Sigma_\inn}^-
        \right]_{\approx},
    \end{align}
    where for $V, V'\in \pol(\HSigma)$, $V\approx V'$ means
    that $P_\Sigma^V-P_\Sigma^{V'}\in I_2(\HSigma)$, i.e., is a Hilbert-Schmidt
    operator $\HSigma\selfmaps$.
\end{definition}
The equivalence relation $\approx$ can be refined to give another equivalence
relation $\approx_0$ describing polarization classes of equal charge; c.f. 
\cite{Deckert2010a} and Remark~\ref{rem:approx_0}.  As a simple corollary of Theorem~\ref{thm:Shale-Stinespring}
one gets:
\begin{corollary}[Dirac Sea Evolution]
    \label{thm:evolution-lift}
    Let $\Sigma,\Sigma'$ be Cauchy surfaces. Then any choice
    $V\in \class_\Sigma(A)$ and $W\in \class_{\Sigma'}(A)$ implies
    condition (b) of
    Theorem~\ref{thm:Shale-Stinespring} and therefore
    the existence of a lift
    $\widetilde U_{\Sigma'\Sigma}^A:\cF(V,\HSigma)\to\cF(W,\cH_{\Sigma'})$ 
    obeying \eqref{eq:lift-condition}.
\end{corollary}
Consequently, any choice $V\in \class_\Sigma(A)$ and $W\in \class_{\Sigma'}(A)$
gives rise to a lift of the one-particle Dirac evolution between the
corresponding $\cF(V,\HSigma)$ and $\cF(W,\cH_{\Sigma'})$ that is unique up to a
phase. The crucial questions are: 1) On which 
properties of $A$ and $\Sigma$ do these polarization classes
depend? 2) How do they behave under Lorentz and gauge transforms? 3) Is there an
explicit representative for each class? These question will be answered by our
main results given in the next section. The next important question is about the
unidentified phase. Although transition probabilities are independent
of this phase, dynamic quantities like the charge current will depend directly on
it.  We briefly discuss this in Section~\ref{sec:outlook} and give an
outlook of what needs to be done to derive the vacuum expectation of the polarization current.

\subsection{Main Results}
\label{sec:main}

The definition~\eqref{eq:def-CA} of the physical polarization classes
involves the one-particle Dirac evolution operator
and is therefore not very useful in finding 
an explicit description of admissible Fock spaces for the
implementation of the second-quantized Dirac evolution.  In our main results
Theorems~\ref{thm:ident-pol-class}-\ref{thm:representative} we give a more
direct identification of the polarization classes classes $\class_\Sigma(A)$ and
state some of their fundamental geometric properties. 

The first one ensures that
the classes $\class_\Sigma(A)$ are independent of the history of $A$, instead
they depend on the tangential components of $A$ on $\Sigma$ only.
\begin{theorem}[Identification of Polarization Classes]
    \label{thm:ident-pol-class}
    Let $\Sigma$ be a Cauchy surface and let $A$ and $\widetilde A$ be two smooth
    and compactly supported
    external fields. Then
    \begin{align}
        \class_\Sigma(A)=\class_\Sigma(\widetilde A)
        \qquad
        \Leftrightarrow
        \qquad
        A|_{T\Sigma}=\widetilde A|_{T\Sigma}
    \end{align}
    where $A|_{T\Sigma}=\widetilde A|_{T\Sigma}$ means that for all $x\in\Sigma$
    and  $y\in T_x\Sigma$ we have $A_\mu(x)y^\mu = \widetilde A_\mu(x)y^\mu$.
\end{theorem}
Ruijsenaar's result, see \cite{Ruijsenaars1977a},
may be viewed as the special case of this theorem pertaining to
$\widetilde A=0$ and, for $t$ fixed, $\Sigma=\Sigma_t=\{x\in\R^4|\;x^0=t\}$ 
being an equal-time hyperplane.  

Furthermore, the polarization classes transform naturally under Lorentz and
gauge transformations:
\begin{theorem}[Lorentz and Gauge Transforms]
    \label{thm:lorentz-gauge}
    Let $V\in\pol(\HSigma)$ be a polarization.
    \begin{enumerate}[(i)]
        \item Consider a \emph{Lorentz transformation}
            given by $L^{(S,\Lambda)}_\Sigma:\HSigma\to\cH_{\Lambda\Sigma}$
            for a spinor transformation matrix $S\in \C^{4\times 4}$
            and an associated proper orthochronous Lorentz transformation
            matrix $\Lambda\in \operatorname{SO}^\uparrow(1,3)$, cf.
            \cite[Section 2.3]{Deckert2014}.
            Then:
            \begin{align}
                V\in\class_\Sigma(A)
                \qquad
                \Leftrightarrow
                \qquad
                L^{(S,\Lambda)}_\Sigma V\in
                \class_{\Lambda\Sigma}(\Lambda A(\Lambda^{-1}\cdot)).
            \end{align}
        \item 
            Consider a \emph{gauge transformation}
            $A\mapsto A+\partial \Omega$ for some 
            $\Omega\in C^\infty_c(\R^4,\R)$
            given by the multiplication operator
            $e^{-i\Omega}:\HSigma\to\HSigma$,
            $\psi\mapsto \psi'=e^{-i\Omega}\psi$.
            Then:
            \begin{align}
                V\in\class_\Sigma(A)
                \qquad
                \Leftrightarrow
                \qquad
                e^{-i\Omega}V\in
                \class_\Sigma(A+\partial\Omega).
            \end{align}
    \end{enumerate}
\end{theorem}
As we are mainly interested in a \emph{local} study of the second-quantized
Dirac evolution, we only allow compactly supported vector potentials $A$, and
therefore, have to restrict the gauge transformations $e^{-i \Omega}$ to
compactly supported $\Omega$ as well. Treating more general vector potentials
$A$ and gauge transforms $e^{-i \Omega}$ would require an analysis of decay properties at
infinity which is not our focus here.

Finally, given Cauchy surface $\Sigma$, there is an explicit representative of
the equivalence class of polarizations $\class_\Sigma(A)$ which can be given in
terms of a compact, skew-adjoint linear operator $Q^A_\Sigma:\HSigma\selfmaps$,
as defined in \eqref{eq:V_Sigma} below. With it the polarization class can be
identified as follows:
\begin{theorem}
    \label{thm:representative}
    Given Cauchy surface $\Sigma$, we have
    $\class_\Sigma(A)=\left[e^{Q_\Sigma^A}
    \cH_\Sigma^-\right]_{\approx}$.
\end{theorem}
Other representatives for polarization classes $\class_\Sigma(A)$ beyond the
``interpolating representation'' $U_{\Sigma\Sigma_\inn}^A\cH^-_{\Sigma_\inn}$,
as used in Definition~\ref{def:pol-classes}, can be inferred from the so-called
Furry picture, as worked out for equal-time hyperplanes in \cite{Fierz1979}, and
from the global constructions of the fermionic projector given in
\cite{Finster2013,Finster2015c}.  In contrast to global constructions, the
representation given in Theorem~\ref{thm:representative} uses only \emph{local}
geometric information of the vector potential $A$ at $\Sigma$; cf.
\eqref{eq:V_Sigma}, \eqref{eq:P-special-lambda}, and Lemma~\ref{lem:Pminus-lambda}
below.

The implications on the physical picture can be seen as follows. The Dirac sea
on Cauchy surface $\Sigma$ can be described in any Fock space $\cF(V,\HSigma)$
for any choice of polarization $V\in \class_\Sigma(A)$. The polarization class
$\class_\Sigma(A)$ is uniquely determined by the tangential components of the
external potential $A$ on $\Sigma$. This is an object that transforms
covariantly under Lorentz and gauge transformations.  The choice of the particular
polarization can then be seen as a ``choice of coordinates'' in which the Dirac sea is
described. When regarding the Dirac evolution from one Cauchy surface $\Sigma$
to $\Sigma'$ another ``choice of coordinates'' $W\in \class_{\Sigma'}(A)$ 
has to be
made. Then one yields an evolution operator $\widetilde
U_{\Sigma'\Sigma}^A:\cF(V,\HSigma)\to\cF(W,\cH_{\Sigma'})$ which is unique up to
an arbitrary phase
Corollary~\ref{thm:evolution-lift}. Transition probabilities of the kind
$|\langle \Psi,\widetilde U_{\Sigma'\Sigma}^A\Phi\rangle|^2$ for
$\Psi\in\cF(W,\cH_{\Sigma'})$ and $\Phi\in\cF(V,\HSigma)$ are well-defined and
unique without the need of a renormalization method. 
Finally, for a family of Cauchy surfaces $(\Sigma_t)_{t\in\R}$ that
interpolates smoothly between $\Sigma$
and $\Sigma'$ we also give an infinitesimal version of how the external
potential $A$ changes the polarization in terms of the flow parameter $t$; see
Theorem~\ref{thm:inf} below.

\begin{remark}[Charge Sectors]
    \label{rem:approx_0}
    Given two polarizations $V,W\in\pol(\HSigma)$ such that
    $P_\Sigma^V-P_\Sigma^W$ is a compact operator, e.g., as in the case $V{\approx}W$ as
    defined in \eqref{eq:def-CA}, 
    one can define their relative
    charge, denoted by $\operatorname{charge}(V,W)$, to be the Fredholm index
    of $P^W_\Sigma|_{V\to W}$; cf. \cite{Deckert2010a}.
    The equivalence relation $\approx$ 
    in the claim of
    Theorem~\ref{thm:representative} can then be replaced by the finer
    equivalence relation $\approx_0$, which is defined as follows: $V\approx_0
    W$ if and only if $V\approx W$ and $\operatorname{charge}(V,W)=0$.  This is
    shown as an addendum to the proof of Theorem~\ref{thm:representative}.
\end{remark}

\subsection{Outlook}
\label{sec:outlook}

As indicated at the end of the introduction
the current operator depends directly on the unspecified phase of $\widetilde
U_{\Sigma'\Sigma}^A$. This can be seen from Bogolyubov's formula
\begin{align}
    \label{eq:current}
    j^\mu(x)
    =
    i\widetilde{U}_{\Sigma_\inn\Sigma_\out}^A \frac{\delta
        \widetilde{U}_{\Sigma_\out\Sigma_\inn}^A}{\delta A_\mu(x)}
\end{align}
where $\Sigma_\out$ is a Cauchy surfaces in the remote future of the support of
$A$ such that $\Sigma_\out\cap\supp A=\emptyset$. Hence, without identification
of the derivative of the phase of $\widetilde U_{\Sigma'\Sigma}^A$ the physical
current
is not fully specified. Nevertheless, now the
situation is slightly better than in the standard perturbative approach. As for
each choice of admissible polarizations in $\class_{\Sigma'}(A)$ and
$\class_\Sigma(A)$, identified above, there is a well-defined lift
$\widetilde U_{\Sigma'\Sigma}^A$ of the Dirac evolution operator
$U_{\Sigma'\Sigma}^A$ and therefore also a well-defined current
\eqref{eq:current}. Now it is only the task to select the physical relevant one.
One way of doing so is to impose extra conditions on the \eqref{eq:current}, and hence, the
phase, so that the set of admissible phases shrinks to one that produces the same
currents up to the known freedom of charge renormalization; see
\cite{Epstein1973,Scharf1995,Mickelsson1998,gracia2000}. In the case of equal-time hyperplanes a
choice of the unidentified phase was given by parallel transport in
\cite{mickelsson_phase_2014}. On top of the geometric construction and despite
the fact that there are still degrees of freedom left, Mickelsson's current
is particularly interesting because it agrees with conventional
perturbation theory up to second order. 
Yet the open question remains which additional physical requirements may 
constraint these degree of freedoms up to the one of the numerical
value of the elementary charge $e$ fixed by the experiment.

The issue of the unidentified phase particularly concerns the so-called
phenomenon of ``vacuum polarization'' as well as the dynamical description of
pair creation processes for which only a few rigorous treatments are available;
e.g., see \cite{Gravejat2013} for vacuum polarization in the Hartree-Fock
approximation for static external sources, \cite{Pickl2008} for adiabatic pair
creation, and for a more fundamental approach the so-called ``Theory of Causal
Fermion Systems'' \cite{FinsterBook,Finster2015a,Finster2015b}, which is based
on a reformulation of quantum electrodynamics by means of an action principle.

\subsection{Definitions, Constants, Notation, and previous Results}
\label{sec:preliminaries}

In this section we briefly review the notation and results about the
one-particle Dirac evolution on Cauchy surfaces provided in a previous work
\cite{Deckert2014}. The present article, dealing with the second-quantization
Dirac evolution, is based on this work. 

Space-time $\R^4$ is endowed with
metric tensor
$g=(g_{\mu\nu})_{\mu,\nu=0,1,2,3}=\operatorname{diag}(1,-1,-1,-1)$, and its
elements are denoted by four-vectors $x=(x^0,x^1,x^2,x^3)=(x^0,\vec x)=x^\mu
e_\mu$, for $e_\mu$ being the canonical basis vectors. Raising and lowering of
indices is done w.r.t.\ $g$. Moreover, we use Einstein's summation
convention, the standard
representation of the Dirac matrices $\gamma^\mu\in\C^{4\times 4}$ that fulfill
$\{\gamma^\mu,\gamma^\nu\}=2g^{\mu\nu}$, and Feynman's slash-notation
$\slashed{\partial}=\gamma^\mu \partial_\mu$, $\slashed A = \gamma^\mu A_\mu$.
When considering subsets of space-time
$\mathbb R^4$ we shall use the following notations: $\causal:=\{x\in\R^4|\;x_\mu
x^\mu\ge 0\}$ and $\past:=\{x\in\R^4|\;x_\mu x^\mu> 0, x^0<0\}$.

The central geometric objects for posing the initial value problem for
\eqref{eq:dirac-equation} are Cauchy surfaces defined as follows:
\begin{definition}[Cauchy Surfaces]\label{def:cauchy-surface}
    We define a Cauchy surface $\Sigma$ in $\R^4$ to be a smooth, 3-dimensional
    submanifold of $\R^4$ that fulfills the following three conditions:
    \begin{enumerate}
        \item Every inextensible,  two-sided, time- or light-like, continuous path
            in $\R^4$ intersects $\Sigma$ in a unique point.
        \item For every $x\in\Sigma$, the tangential space $T_x\Sigma$ is
            space-like.
        \item The tangential spaces to $\Sigma$ are bounded away from light-like
            directions in the following sense: The only light-like accumulation point
            of $\bigcup_{x\in\Sigma}T_{x}\Sigma$ is zero.
    \end{enumerate} 
\end{definition}
In coordinates, every Cauchy surface $\Sigma$ can be parametrized as 
\begin{align}
    \label{eq:parametrize Sigma}
    \Sigma=\{\pi_\Sigma(\vec x):=(t_\Sigma(\vec x),\vec x)\;|\;\vec x\in\R^3\}
\end{align}
with a smooth function $t_\Sigma:\R^3\to\R$. For convenience and without
restricting generality of our results we keep a
global constant
\begin{align}
    \label{eq:Vmax}
    0<\Vmax<1
\end{align}
fixed and work only with Cauchy surfaces $\Sigma$ such that
\begin{align}
    \label{eq:bound-grad_tSigma}
    \sup_{\vec x\in\R^3}|\nabla t_\Sigma(\vec x)|\le \Vmax.
\end{align}
The assumption (c) in Definition \ref{def:cauchy-surface} and
\eqref{eq:bound-grad_tSigma} can be relaxed to $|\nabla t_\Sigma(\vec
x)|<1$ for all $\vec x\in\R^3$ due to the causal structure of the solutions to
the Dirac equation, although this is not worked out in this paper.

The standard volume form over $\R^4$ is denoted by
$d^4x=dx^0\,dx^1\,dx^2\,dx^3$; the product of forms is understood as wedge
product. The symbols $d^3x$ and $d^3\vec x$
mean the 3-form $d^3x=dx^1\,dx^2\,dx^3$ on $\R^4$ and on $\R^3$, respectively.
Contraction of a form $\omega$ with a vector $v$ is denoted
by $i_v(\omega)$.
The notation $i_v(\omega)$ is also used for the spinor matrix
valued vector $\gamma=(\gamma^0,\gamma^1,\gamma^2,\gamma^3)=\gamma^\mu e_\mu$:
\begin{align}
    i_\gamma(d^4x)=\gamma^\mu i_{e_\mu}(d^4x).
\end{align}
Furthermore, for a $4$-spinor $\psi\in\C^4$ (viewed as column vector),  
$\overline{\psi}$ stands for the row vector $\psi^*\gamma^0$,
where ${}^*$ denotes hermitian conjugation.

Smooth families $(\Sigma_t)_{t\in T}$ of Cauchy surfaces, indexed
by an interval $T\subseteq\R$ and fulfilling
\eqref{eq:bound-grad_tSigma}, are denoted by
\begin{align}
    \label{eq:boldsigma}
    \boldsymbol{\Sigma}:=\{(x,t)|\;t\in T,\, x\in\Sigma_t\}.
\end{align}
Given the external electromagnetic vector potential 
$A\in C^\infty_c(\R^4,\R^4)$ of interest, 
we assume that the set $\{(x,t)\in\boldsymbol{\Sigma}|\;x\in\operatorname{supp}(A)\}$
is compact. This condition is trivially fulfilled
in the important case of a compact interval $T=[t_0,t_1]$ with $\boldsymbol{\Sigma}$
interpolating between two Cauchy surfaces $\Sigma_{t_0}$ and $\Sigma_{t_1}$.
The compactness condition is also automatically fulfilled in the case that
$T=\R$ with $\boldsymbol{\Sigma}$ being a smooth foliation of 
the Minkowski space-time $\R^4$.

We assume furthermore that
the family $(\Sigma_t)_{t\in T}$ is driven driven by a (Minkowski) normal vector field
$vn:\boldsymbol{\Sigma}\to\R^4$, where $n:\boldsymbol{\Sigma}\to\R^4$,
$(x,t)\mapsto n_t(x)$, denotes the future-directed (Minkowski) normal unit
vector field to the Cauchy surfaces and $v:\boldsymbol{\Sigma}\to\R$,
$(x,t)\mapsto v_t(x)$, denotes the speed at which the Cauchy surfaces move
forward in normal direction. For technical reasons, 
in particular when using the chain rule,
it is convenient to extend
the ``speed'' $v$ and the unit vector field 
$n$ in a smooth way to the domain $\R^4\times T$.
In the case that $\boldsymbol{\Sigma}$ is a foliation of
space-time, we may even drop the $t$--dependence of $v$ and $n$.
In this important case, some of the arguments below become slightly simpler.
\begin{definition}[Spaces of Initial Data]
    \label{def:HSigma}
    For any Cauchy surface $\Sigma$ we define the vector space 
    $\CSigma:=C^\infty_c(\Sigma,\C^4)$.
    For a given Cauchy surface $\Sigma$, 
    let $\HSigma=L^2(\Sigma,\C^4)$ denote the  vector space of
    all 4-spinor valued measurable functions $\phi:\Sigma\to \C^4$ (modulo changes on null sets)
    having a finite norm
    $\|\phi\|=\sqrt{\sk{\phi,\phi}}<\infty$ w.r.t.\ the scalar
    product 
    \begin{align}
        \label{eq:scalar-product}
        \sk{\phi,\psi}
        =\int_\Sigma\overline{\phi(x)}i_\gamma(d^4x)\psi(x).
    \end{align}
\end{definition}
For $x\in\Sigma$, the restriction of the  spinor matrix valued 3-form
$i_\gamma(d^4x)$ to the tangential space $T_x\Sigma$ is given by
\begin{align}
    \label{eq: repr igammad4x}
    i_\gamma(d^4x)=
    \slashed{n}(x)i_n(d^4x)=
    \left(
    \gamma^0-\sum_{\mu=1}^3 \gamma^\mu \frac{\partial t_\Sigma(\vec{x})}{\partial x^\mu}
    \right)d^3x =: \Gamma(\vec x) d^3x
    \text{ on }(T_x\Sigma)^3.
\end{align}
As a consequence of the \eqref{eq:bound-grad_tSigma}, 
there is a positive constant 
$\GammaMax=\GammaMax(\Vmax)$ such that 
\begin{align}
    \label{eq:bound-Gamma}
    \| \Gamma(\vec x) \| 
    & \leq
    \GammaMax,
    \qquad
    \forall \, \vec x\in\R^3.
\end{align}

The class of solutions to the Dirac equation \eqref{eq:dirac-equation}
considered in this work is defined by:
\begin{definition}[Solution Spaces]\label{def:CA}
    \mbox{}
    \begin{enumerate}[(i)]
        \item Let $\CA$ denote the space of all smooth solutions $\psi\in
            C^\infty(\R^4,\C^4)$ of the Dirac equation~\eqref{eq:dirac-equation} which
            have a spatially compact causal support in the following sense: There is a
            compact set $K\subset\R^4$ such that $\operatorname{supp}\psi\subseteq
            K+\causal$.         
        \item We endow $\CA$ with the scalar product given in
            \eqref{eq:scalar-product}; note that due to conservation of the
            4-vector current $\overline{\phi}\gamma^\mu\psi$,
            the scalar product
            $\sk{\cdot,\cdot}:\CA\times\CA\to\C$ is independent
            of the particular choice of $\Sigma$.
        \item Let $\HA$ be the Hilbert space given by the (abstract) completion
            of $\CA$. 
    \end{enumerate}
\end{definition}
Theorem 2.21 in \cite{Deckert2014} ensures:
\begin{theorem}[Initial Value Problem and Support]
    \label{dirac-existence-uniqueness}
    Let $\Sigma$ be a Cauchy surface and
    $\chi_\Sigma\in\cC_{c}^{\infty}(\Sigma,\C^{4})$ be given initial data.
    Then, there is a $\psi\in\CA$ such that $\psi|_\Sigma=\chi_\Sigma$ and
    $\supp \psi \subseteq \supp \chi_\Sigma + \operatorname{Causal}$.
    Moreover, suppose $\widetilde\psi\in C^\infty(\R^4,\C^4)$ solves the Dirac
    equation \eqref{eq:dirac-equation} for initial data
    $\widetilde\psi|_\Sigma=\chi_\Sigma$, then $\widetilde\psi=\psi$.
\end{theorem}
This theorem gives rise to the following definition in which we use the notation
$\psi|_\Sigma\in \CSigma$ to denote the restriction of a
$\psi\in\CA$ to a Cauchy surface $\Sigma$.
\begin{definition}[Evolution Operators]
    \label{def:UA}
    Let $\Sigma,\Sigma'$ be Cauchy surfaces.  In view of
    Theorem~\ref{dirac-existence-uniqueness} we define the isomorphic isometries
    \begin{align}
        \begin{split}
            U_{\Sigma A}&:\CA\to\CSigma, \\
            U_{A \Sigma}&:\CSigma\to\CA, \\
            U_{\Sigma'\Sigma}^A&:\CSigma \to\cC_{\Sigma'},
        \end{split}
        \begin{split}
            &U_{\Sigma A} \, \phi := \phi|_\Sigma, \\
            &U_{\Sigma A}\,\chi_\Sigma:=\psi, \\
            &U_{\Sigma'\Sigma}^A:=U_{\Sigma'A}U_{A \Sigma},
        \end{split} 
    \end{align}
    where $\chi_\Sigma\in\CSigma$, $\phi\in\CA$, and $\psi$ is the solution
    corresponding to initial value $\chi_\Sigma$ as in
    Theorem~\ref{dirac-existence-uniqueness}.
    These maps extend uniquely to unitary maps
    $U_{A\Sigma}:\HSigma\to\HA$, $U_{\Sigma A}:\HA\to\HSigma$ and
    $U_{\Sigma'\Sigma}^A:\HSigma\to\cH_{\Sigma'}$.
\end{definition}
Here we differ from the notation used in Theorem~2.23
in \cite{Deckert2014} where $U^A_{\Sigma'\Sigma}$ was denoted by
$\cF^A_{\Sigma'\Sigma}$. Furthermore, it will be useful to express the
orthogonal projector $P_\Sigma^-$ in an momentum integral representation over
the mass shell 
\begin{align}
    {\mathcal M}=\{p\in\R^4|\;p_\mu p^\mu=m^2\}
    =
    \cM_+\cup\cM_-,
    \qquad
    {\mathcal M}_\pm=\{p\in{\mathcal M}|\;\pm p^0>0\};
\end{align}
cf.\ Lemma~\ref{lemma:kern-p-minus} and the definition of $\cF_{\cM\Sigma}$ in
\cite{Deckert2014}. We endow ${\mathcal M}$ with the orientation that makes the projection
${\mathcal M}\to\R^3$, $(p^0,\vec p)\mapsto \vec p$ positively oriented.
One finds that
\begin{align}\label{eq:ipd4p}
    i_p(d^4 p)=\frac{m^2}{p^0}dp^1\,dp^2\,dp^3=
    \frac{m^2}{p^0}d^3p\text{ on } (T_p{\mathcal M})^3.
\end{align}

\paragraph{General Notation.} Positive constants and remainder terms are denoted by $C_1, C_2,
C_3,\ldots$ and $r_1, r_2, r_3,\ldots$, respectively. 
They keep their meaning throughout the whole article. Any fixed
quantity a constant depends on (except numerical constants like electron mass
$m$
and charge $e$) is displayed at least once when the constant is introduced.
Furthermore, we classify the behavior of functions using the following 
variant of the Landau symbol notation.
\begin{definition}
    For lists of variables 
    $x,y,z$
    we use the notation
    \begin{align}
        \label{eq:landau}
        f(x,y,z)
        =
        O_{y}
        \left(
        g(x)
        \right),
        \qquad
        \text{for all } (x,y,z)\in\mathrm{domain}
    \end{align}
    to mean the following: There exists a constant $C(y)$
    depending only on the parameters $y$, but neither on $x$ nor on $z$,
    such that
    \begin{align}
        | f(x,y,z) |
        \leq
        C(y) | g(x) |,
        \qquad
        \text{ for all } (x,y,z)\in\mathrm{domain},
    \end{align}
    where $|{\cdot}|$ stands for the appropriate norm applicable to $f$.
    Note that the notation \eqref{eq:landau} does not mean that
    $f(x,y,z)=f(x,y)$, i.e., that the value of $f$ is independent of $z$.
    Rather, it
    just means that the bound is uniform in $z$.
\end{definition}

\section{Proofs}
\label{sec:proofs}

The key idea in the proofs of our main results
Theorem~\ref{thm:ident-pol-class}, \ref{thm:lorentz-gauge}, and
\ref{thm:representative} is to guess a simple enough
operator $P_\Sigma^A:\HSigma\selfmaps$ so that
\begin{align}
    \label{eq:key-prop}
    U^A_{\Sigma\Sigma_\inn}P_{\Sigma_\inn}^- U^A_{\Sigma_\inn\Sigma}
    -
    P_\Sigma^A \in I_2(\HSigma).
\end{align}
It turns out that all claims about the properties of the polarization classes
$\class_\Sigma(A)$ above can
then be inferred from the properties of $P_\Sigma^A$. This is due to the fact
that \eqref{eq:key-prop} is compatible with the Hilbert-Schmidt operator freedom
encoded in the ${\approx}$ equivalence relation that was used to define the polarization
classes $\class_\Sigma(A)$; see Definition~\ref{def:pol-classes}.

The intuition behind our guess of $P_\Sigma^A$ comes from gauge transforms.
Imagine the special situation in which an external potential $A$ could be gauged
to zero, i.e., $A=\partial\Omega$ for a given scalar field $\Omega$. In this
case $e^{-i\Omega} P_\Sigma^- e^{i\Omega}$ is a good candidate for $P_\Sigma^A$.
Now in the case of general external potentials $A$ that cannot be attained by a
gauge transformation of the zero potential, the idea is to implement different
gauge transforms locally near to each space-time point.  For example, if
$p^-(y-x)$ denotes the informal integral kernel of the operator $P_\Sigma^-$,
one could try to define $P_\Sigma^A$ as the operator corresponding to the
informal kernel $p^A(x,y)=e^{-i\lambda^A(x,y)}p^-(y-x)$ for the choice
$\lambda^A(x)=\frac{1}{2}(A(x)+A(y))_\mu(y-x)^\mu$.  Due to this choice, the
action of $\lambda^A(x,y)$ can be interpreted as a local gauge transform of
$p^-(y-x)$ from the zero potential to the potential $A_\mu(x)$ at space-time
point $x$. It turns out that these local gauge transforms give rise to
an operator $P_\Sigma^A$ that fulfills \eqref{eq:key-prop}.

\paragraph{Section Overview} In Section~\ref{sec:operator_PA} we define the
operators $P_\Sigma^-$ and $P_\Sigma^A$ and state their main properties.
Assuming these properties  we prove 
our main results in Section~\ref{sec:main-results-proofs}.
The proofs of those employed properties are delivered afterwards in Sections~\ref{sec:pol-classes} and
\ref{sec:evolution-polarization-classes}. 

\subsection{The Operators $P_\Sigma^-$ and $P_\Sigma^A$}
\label{sec:operator_PA}

As described in the previous section, the central objects of our study are the
operators $P_\Sigma^-$ and operators which are derived from them like the
discussed $P_\Sigma^A$.  Lemma~\ref{lemma:kern-p-minus} describes the integral
representation of the orthogonal projector $P_\Sigma^-$.  For this we introduce
the notation
\begin{align}
    \label{eq:def-r}
    r(w):=\sqrt{-w_\mu w^\mu}
    \quad\text{ for }\quad
    w\in\operatorname{domain}(r):=\{w\in\C^4|\;
    -w_\mu w^\mu\in\C\setminus\R^-_0\}.
\end{align}
The square root is interpreted as its principal value
$\sqrt{r^2e^{2i\varphi}}=re^{i\varphi}$ for $r>0$, 
\mbox{$-\frac{\pi}{2}<\varphi<\frac{\pi}{2}$}.
\ifx\arxiv\undefined
We note that for a Cauchy surface $\Sigma$ fulfilling 
\eqref{eq:bound-grad_tSigma} and $0\neq z=y-x$ with
$x,y\in\Sigma$ one has
\begin{align}
    \label{eq:rz-z}
    \sqrt{1-\Vmax^2}|\vec z|
     \leq r(z)
    \leq |\vec z| 
    \leq |z|\leq \sqrt{1+\Vmax^2}|\vec z|.
\end{align}
\fi
To deal with the singularity of the informal integral kernel $p^-(y-x)$ of
the projection operator $P_\Sigma^-$ at the diagonal $x=y$, we use 
a regularization shifting the argument $y-x$
a little in direction of the imaginary past.
\begin{lemma}
    \label{lemma:kern-p-minus}
    For $\phi,\psi\in\CSigma$ 
    and any past-directed time-like vector 
    $u\in\past$
    one has
    \begin{align}
        \label{eq:pminus-claim}
        \sk{\phi,P^-_\Sigma\psi} 
        =
        \lim_{\epsilon\downarrow 0}
        \int_{x\in\Sigma} \overline{\phi}(x) \, i_\gamma(d^4x)
        \int_{y\in\Sigma} p^-(y-x+i\epsilon u) \, i_\gamma(d^4y) \, \psi(y),
    \end{align}  
    where
    \begin{gather}
        p^-: \R^4 + i\operatorname{Past}
        \to \C^{4x4}, 
        \quad
        p^-(w)
        =
        \frac{1}{(2\pi)^3m}
        \int_{\cM_-} \frac{\slashed p+m}{2m} e^{ipw} \, i_p(d^4p)
        =\frac{-i\slashed \partial+m}{2m} D(w),
        \label{eq:pminus}
        \\
        D: \R^4 + i\operatorname{Past} 
        \to \C,\quad 
        D(w) 
        =
        \frac{1}{(2\pi)^3m}
        \int_{\cM_-} e^{ipw} \, i_p(d^4p)
        =
        -\frac{m^3}{2 \pi^2} \frac{K_1(m r(w))}{m r(w)},
    \label{eq:def-D}
    \\
        K_1:\R^+ + i\mathbb R 
         \to 
        \C,
        \qquad
        K_1(\xi)=
        \xi\int_1^\infty e^{-\xi s}\sqrt{s^2-1}\, ds.
        \label{eq:K1}
    \end{gather}
    $K_1$ is the modified Bessel function of second kind of order one.
    The functions $D$ and $p^-$ have
    analytic continuations defined on $\operatorname{domain}(r)$.
    The corresponding continuations are denoted by the same symbols.
\end{lemma}
The proof is given in Section~\ref{sec:pol-classes}.
It is based on the momentum integral representation given in Theorem~2.15 in
\cite{Deckert2014}. In the following we define several candidates for $P_\Sigma^A$
fulfilling the key property \eqref{eq:key-prop} as discussed in the beginning of
Section~\ref{sec:proofs}.  We will denote these operators by
$P_\Sigma^\lambda:\HSigma\selfmaps$ where the superscript $\lambda$ denotes an
element out of the following class of ``local'' gauge functions:
\begin{definition}
    \label{def:PLambda}
    For $A\in C^\infty_c(\R^4,\R^4)$
    let $\cG(A)$ denote the set of all functions $\lambda:\R^4\times\R^4\to\R$
    with the following properties:
    \begin{enumerate}[(i)]
        \item $\lambda\in  C^\infty(\R^4\times\R^4,\R)$.
        \item There is a compact set $K\subset\R^4$ such that $\supp \lambda
            \subseteq K\times\R^4\cup\R^4\times K$.
        \item $\lambda$ vanishes on the diagonal, 
            i.e., $\lambda(x,x)=0$ for $x\in\R^4$.
        \item On the diagonal the first derivatives fulfill
            \begin{align}
                \partial^x\lambda(x,y)=-\partial^y\lambda(x,y)=A(x)
                \qquad\text{for } x=y\in\R^4.
            \end{align} 
    \end{enumerate}
\end{definition}
Given a ``local'' gauge transform $\lambda\in\cG(A)$ we define the corresponding operator
$P_\Sigma^\lambda$ using the heuristic idea behind $P_\Sigma^A$
we discussed in the beginning of Section~\ref{sec:proofs}.
\begin{lemma}
    \label{lem:Pminus-lambda}
    Given $A\in C^\infty_c(\R^4,\R^4)$ and $\lambda\in\cG(A)$ there is a unique bounded operator
    $P_\Sigma^\lambda:\HSigma\selfmaps$ with matrix elements
    \begin{align}
        \label{eq:Pminus-lambda}
        \sk{\phi,P_\Sigma^\lambda\psi} 
        &=
        \lim_{\epsilon\downarrow 0} \sk{\phi,P_\Sigma^{\lambda,\epsilon u}\psi}
        \text{ with}
        \\ 
        \label{eq:Pminus-lambda-epsilon}
        \sk{\phi,P_\Sigma^{\lambda,\epsilon u}\psi}&:= \int_{x\in\Sigma} \overline{\phi}(x) \, i_\gamma(d^4x)
        \int_{y\in\Sigma} e^{-i\lambda(x,y)}p^-(y-x+i\epsilon u) \, i_\gamma(d^4y)
        \, \psi(y).
    \end{align}
    for any given $\phi,\psi\in\CSigma$ and any past-directed time-like vector
    $u\in\past$.  In particular, the limit in \eqref{eq:Pminus-lambda} does not
    depend on the choice of $u\in\operatorname{Past}$. For 
    $\Delta P_\Sigma^\lambda:=P_\Sigma^\lambda-P_\Sigma^-$,  $\psi\in\HSigma$,
    and almost all $x\in\Sigma$ it holds
    \begin{align}
        \label{eq:pdiff}
        \left(\Delta P_\Sigma^\lambda\psi\right)(x)=
        \int_{y\in\Sigma} (e^{-i\lambda(x,y)}-1)
        p^-(y-x) \, i_\gamma(d^4y)
        \, \psi(y),
    \end{align} 
    and furthermore:
    \begin{enumerate}[(i)]
        \item The operator norm of $P_\Sigma^\lambda$ is bounded by a constant
            $\constl{c:Pminus-lambda}(\Vmax,\lambda)$; 
            cf. \eqref{eq:bound-grad_tSigma}; 
        \item 
            $\Delta P_\Sigma^\lambda$ is a
            compact operator;
        \item 
            $|\Delta P_\Sigma^\lambda|^2$ is a Hilbert-Schmidt operator.
        \item If $\lambda(x,y)=-\lambda(y,x)$ for all $x,y\in\Sigma$,
          then $P_\Sigma^\lambda$ is self-adjoint.
    \end{enumerate}
\end{lemma}
This lemma is proven in Section \ref{sec:pol-classes}.
Two important examples of elements in $\cG(A)$ are:
\begin{itemize}
    \item The choice $\lambda(x,y)=\Omega(x)-\Omega(y)$ for
        $\Omega\in\cC_c^\infty(\R^4,\R)$ fulfills 
        $\lambda\in\cG(\partial\Omega)$. 
        Such a $\lambda$
        delivers a good candidate for the operator $P_\Sigma^A$ fulfilling
        \eqref{eq:key-prop} if the external field $A$ can be attained from the
        zero field via a gauge transform $A=0\mapsto A=\partial \Omega$.
        We observe for any path $C_{y,x}$ from $y$ to $x$
        \begin{align}
          \label{eq:lambda-as-integral}
          \lambda(x,y)=\int_{C_{y,x}}A_\mu(u)\,du^\mu=
          \frac{1}{2}(A_\mu(x)+A_\mu(y))(x^\mu-y^\mu)+O_A(|x-y|^3).
        \end{align}
    \item For an arbitrary vector potential $A\in C^\infty_c(\R^4,\R^4)$ also
        \begin{align}
            \label{eq:special-lambda}
            \lambda^A(x,y):=\frac{1}{2}(A_\mu(x)+A_\mu(y))(x^\mu-y^\mu)
        \end{align} 
        fulfills $\lambda^A\in\cG(A)$. This choice is motivated by
        the special case \eqref{eq:lambda-as-integral}. It will be particularly
        convenient for our work. Note that it has the symmetry
        $\lambda^A(x,y)=-\lambda^A(y,x)$; cf. part (iv) in 
        Lemma~\ref{lem:Pminus-lambda}. In particular, the operator
        $P_\Sigma^A$ from the
        discussion will be given by 
        \begin{align}
            \label{eq:P-special-lambda}
            P_\Sigma^A:=P_\Sigma^{\lambda^A}.
        \end{align}
\end{itemize}
We shall show that for $\lambda\in\cG(A)$ the operators $P_\Sigma^\lambda$ obey
the key property \eqref{eq:key-prop}.  Our first result about $P_\Sigma^\lambda$
for a $\lambda\in\cG(A)$ is that, up to a Hilbert-Schmidt operator, it depends
only on the restriction of the 1-form $A$ to the tangent bundle $T\Sigma$ of the
Cauchy surface $\Sigma$.
\begin{theorem}
    \label{thm:equivalence-pol-classes}
    Given $A,\widetilde{A}\in C^\infty_c(\R^4,\R^4)$ and
    $\lambda\in \cG(A)$, $\widetilde{\lambda}\in\cG(\widetilde{A})$, the
    following is true:
    \begin{align}
        P_\Sigma^\lambda-P_\Sigma^{\widetilde{\lambda}}\in I_2(\HSigma)
        \qquad
        \Leftrightarrow
        \qquad
        A|_{T\Sigma}=\widetilde A|_{T\Sigma}.
        \label{eq:equivalence-pol-classes}
    \end{align}
\end{theorem}
This theorem is also proven in Section \ref{sec:pol-classes}.
From our next result we can infer that the operators $P_\Sigma^\lambda$ obey the
key property \eqref{eq:key-prop}.
\begin{theorem}
    \label{thm:evo-polclass}
    Given $A\in C^\infty_c(\R^4,\R^4)$,
    $\lambda\in \cG(A)$, and two Cauchy surfaces $\Sigma,\Sigma'$,
    one has 
    \begin{align}
        U_{A\Sigma'}P_{\Sigma'}^\lambda U_{\Sigma' A}-
        U_{A\Sigma}P_\Sigma^\lambda U_{\Sigma A}
        \in I_2(\HA),
        \label{eq:gen-key-prop}
    \end{align}
    where $U_{A\Sigma}$ and $U_{\Sigma A}$ are the Dirac evolution operators
    given in Definition~\ref{def:UA}.
\end{theorem}
Instead of proving this theorem directly we prove it at the end of
Section~\ref{sec:evolution-polarization-classes} as consequence of
Theorem~\ref{thm:inf} below. The latter can be understood as an infinitesimal version
of Theorem~\ref{thm:evo-polclass}. To state Theorem~\ref{thm:inf} we consider a
family $(\Sigma_t)_{t\in T}$ of Cauchy surfaces encoded by
$\boldsymbol\Sigma$, see \eqref{eq:boldsigma}, such that $\Sigma=\Sigma_{t_0}$
and $\Sigma'=\Sigma_{t_1}$. In addition we
need the
following helper object $s_\Sigma^A$ defined in Definition~\ref{def:S} below as
well as the following notation.  Given an electromagnetic potential $A\in
C^\infty_c(\R^4,\R^4)$ and a Cauchy surface $\Sigma$ with future-directed unit
normal vector field $n$, we define the electromagnetic field tensor
$F_{\mu\nu}=\partial_\mu A_\nu-\partial_\nu A_\mu$ and
\begin{align}
    E_\mu:=F_{\mu\nu}n^\nu
    \label{eq:electric-field}
\end{align}
referred to as the ``electric field'' with respect to the local Cauchy surface
$\Sigma$.  In the special case $n=e_0=(1,0,0,0)$, this encodes just the electric
part of the electromagnetic field tensor.

Recall from the paragraph preceding Definition~\ref{def:HSigma} that we extended
the unit normal field $n$ on the Cauchy surface to a smooth unit normal field
$n:\R^4\times T\to\R^4$ and velocity field $v:\R^4\times T\to\R$, which induces
the ``electric field'' $E$ to be defined on $\R^4\times T$ as well.  In
particular, after this extension, the partial derivative $\partial
E_\mu(x,t)/\partial t=F_{\mu\nu}(x) \, \partial {n_t}^\nu(x)/\partial t$ then makes
sense.
\begin{definition}
    \label{def:S}
    Recall the definitions of $r(w)$ and $D(w)$ given in \eqref{eq:def-r}
    and \eqref{eq:def-D}, respectively. 
    For $\epsilon>0$, $u\in\operatorname{Past}$, and $x,y\in\R^4$,
    we define the integral kernel
    \begin{align}
        s_\Sigma^{A,\epsilon u}(x,y):=\frac{1}{8m}\slashed{n}(x)\slashed{E}(x)
        r(w)^2 \slashed\partial D(w),
        \quad\text{ where }
        w=y-x+i\epsilon u.
    \end{align}
    Furthermore,  for $x-y$ being space-like (in particular $x\neq y$), 
    we also define the integral kernel
    \begin{align}
        s_\Sigma^A(x,y)=s_\Sigma^{A,0}(x,y):=
        \lim_{\epsilon\downarrow 0}s_\Sigma^{A,\epsilon u}(x,y)
        =\frac{1}{8m}\slashed{n}(x)\slashed{E}(x)
        r(y-x)^2 \slashed\partial D(y-x).
    \end{align}
\end{definition}

We remark that restricted to $x$ and $y$ within a single Cauchy surface
$\Sigma$, the value of the kernel $s_\Sigma^{A,\epsilon u}(x,y)$ depends only on
$\Sigma$ through its normal field $n:\Sigma\to\R^4$. In this case
the definition makes sense without specifying neither the velocity field
$v$ nor the extension of $n$ and $v$ to $\R^4\times T$. In particular,
$s_\Sigma^{A,\epsilon u}(x,y)$ depends only on the Cauchy surface $\Sigma$
but not on the choice of a family $(\Sigma_t)_{t\in T}$.
This stands in contrast
to the derivative $\partial s_{\Sigma_t}^{A,\epsilon u}/\partial t$, which makes
sense everywhere only given a family $(\Sigma_t)_{t\in T}$ and
the extended version of $n$.

Exploiting the properties of $D(w)$ given in Lemma \ref{lemma:kern-p-minus}
and in Corollary~\ref{lem:upper-bounds} in the appendix we shall find:
\begin{lemma}
    \label{lem:SA}
    Let $u\in\past$.
    \begin{enumerate}[(i)]
        \item 
            The integral kernels $s_\Sigma^{A,\epsilon u}$, $\epsilon\geq 0$, 
            give rise to Hilbert-Schmidt operators 
            \begin{align}
                S_\Sigma^{A,\epsilon u}:
                \HSigma\selfmaps,
                \qquad
                S_\Sigma^{A,\epsilon u}\psi(x):=
                \int_\Sigma s_\Sigma^{A,\epsilon u}(x,y)\,i_\gamma(d^4y)\,\psi(y)
                \quad\text{ for almost all }x\in\Sigma,
            \end{align}
            $S_\Sigma^A:=S_\Sigma^{A,0}$,
            with the property that
            $\|S_\Sigma^A - S_\Sigma^{A,\epsilon u}\|_{I_2(\HSigma)}\konv{\epsilon\downarrow 0}0$.
        \item
            Similarly, for $t\in T$, the integral kernels  
            $\partial s_{\Sigma_t}^{A,\epsilon u}/\partial t$, $\epsilon\geq 0$,
            give rise to Hilbert-Schmidt operators 
            \begin{align}
                \dot S_{\Sigma_t}^{A,\epsilon u}:
                \HSigma\selfmaps,
                \qquad
                \dot S_{\Sigma_t}^{A,\epsilon u}\psi(x):=
                \int_{\Sigma_t} \frac{\partial s_{\Sigma_t}^{A,\epsilon u}}{\partial t}(x,y)\,
                i_\gamma(d^4y)\,\psi(y)
                \quad\text{ for almost all }x\in\Sigma_t,
            \end{align}
            $\dot S_{\Sigma_t}^A:=\dot S_{\Sigma_t}^{A,0}$,
            with the property that
            $\sup_{t\in T}\|\dot S_{\Sigma_t}^A\|_{I_2(\cH_{\Sigma_t})}<\infty$
            and
            $\|\dot S_{\Sigma_t}^A - \dot S_{\Sigma_t}^{A,\epsilon u}\|_{I_2(\cH_{\Sigma_t})}
            \konv{\epsilon\downarrow 0}0$ for all $t$.
    \end{enumerate}
 \end{lemma}
With this ingredient our infinitesimal version of Theorem~\ref{thm:evo-polclass}
can  be formulated as follows; for technical convenience, we phrase it only for
the special choice $\lambda^A\in\cG(A)$ defined in \eqref{eq:special-lambda}.
\begin{theorem}
    \label{thm:inf}
    Given $A\in C^\infty_c(\R^4,\R^4)$, any
    smooth family of Cauchy surfaces $\boldsymbol\Sigma$, cf.
    \eqref{eq:boldsigma}, and $t_0,t_1\in T$, and  one has
    \begin{align}
        \label{eq:plambda-inf}
        U_{A\Sigma_{t_1}}
        \left(P_{\Sigma_{t_1}}^A+S_{\Sigma_{t_1}}^A\right)
        U_{\Sigma_{t_1} A}-
        U_{A\Sigma_{t_0}}
        \left(P_{\Sigma_{t_0}}^A+S_{\Sigma_{t_0}}^A\right)
        U_{\Sigma_{t_0} A}
        =
        \int_{t_0}^{t_1} 
        U_{A\Sigma_{t}} R(t) U_{\Sigma_{t}A}
        \, dt
    \end{align}
    for a family of Hilbert-Schmidt operators $R(t)$, $t\in T$, with
    $\sup_{t\in T} \| R(t) \|_{I_2(\cH_{\Sigma_t})}<\infty$. The
    integral in \eqref{eq:plambda-inf} is understood in the weak sense.
\end{theorem}
Note that for the choice $\lambda\in\cG(A)$, $\Sigma_{t_1}=\Sigma$,
$\Sigma_{t_0}=\Sigma_\inn$ one has
$P_{\Sigma_\inn}^\lambda=P_{\Sigma_\inn}^-$, and the restriction of
\eqref{eq:gen-key-prop} to Cauchy surface $\Sigma$ yields property
$U_{\Sigma\Sigma_\inn}^AP_{\Sigma_\inn}^-U_{\Sigma_\inn\Sigma}^A -
P_\Sigma^\lambda\in I_2(\HSigma)$, i.e., the key property \eqref{eq:key-prop}.
The proof of Theorem~\ref{thm:inf} given in
Section~\ref{sec:evolution-polarization-classes} is the heart of this work. 

\subsection{Proofs of Main Results}
\label{sec:main-results-proofs}
In this section, we prove the main results under the assumption that the claims
in Section~\ref{sec:operator_PA} are true.  The proofs of these assumed claims
are then provided in
Sections~\ref{sec:pol-classes}-\ref{sec:evolution-polarization-classes}.  The
connection of how to infer the properties of $\class_\Sigma(A)$ from the
properties of the operators $P_\Sigma^\lambda$ is given by the following lemma.
\begin{lemma}
    \label{lem:pol-class-plambda}
    Let $A\in C^\infty_c(\R^4,\R^4)$, $\Sigma$ be a Cauchy surface, and $\lambda\in\cG(A)$. Then for every polarization $V$ in $\HSigma$, we have
    \begin{align}
        V\in \class_\Sigma(A)
        \qquad
        \Leftrightarrow
        \qquad
        P_\Sigma^V-P_\Sigma^\lambda \in I_2(\HSigma).
    \end{align}
\end{lemma}

\begin{proof}
    By Definition~\ref{def:pol-classes}, $V\in \class_\Sigma(A)$ is equivalent to
    \begin{align}
        \label{eq:class-pol-step}
        P_\Sigma^V-U_{\Sigma\Sigma_\inn}^AP_{\Sigma_\inn}^- U_{\Sigma_\inn\Sigma}^A \in
        I_2(\HSigma).
    \end{align}
    On the other hand, Theorem~\ref{thm:evo-polclass}
    implies
    \begin{align}
        P_\Sigma^\lambda-U_{\Sigma\Sigma_\inn}^AP_\Sigma^- U_{\Sigma_\inn\Sigma}^A \in
        I_2(\HSigma).
    \end{align}
    Thus, statement \eqref{eq:class-pol-step} is equivalent to
    $P_\Sigma^V-P_\Sigma^\lambda\in I_2(\HSigma)$.
\end{proof}

\begin{proof}[Proof of Theorem~\ref{thm:ident-pol-class}]
    $\class_\Sigma(A)=\class_\Sigma(\widetilde A)$ holds true if and only
    if there are  $V\in \class_\Sigma(A)$ and $W\in \class_\Sigma(\widetilde A)$ such that
    \begin{align}
        \label{eq:ident-pol-class-step}
        P_\Sigma^V-P_\Sigma^W\in I_2(\HSigma).
    \end{align}
    Let $\lambda \in \cG(A)$ and $\widetilde \lambda\in \cG(\widetilde A)$.  In view
    of Lemma~\ref{lem:pol-class-plambda}, statement
    \eqref{eq:ident-pol-class-step} is equivalent to
    $P_\Sigma^\lambda - P_\Sigma^{\widetilde \lambda} \in I_2(\HSigma)$.  Due to
    Theorem~\ref{thm:equivalence-pol-classes} the latter is equivalent to
    $A|_{T\Sigma}=\widetilde A|_{T\Sigma}$, which proves the claim.
\end{proof}

\begin{proof}[Proof of Thorem~\ref{thm:lorentz-gauge}]
    Claim (i): Is is sufficient to prove that there exist 
    $V\in \class_\Sigma(A)$ and $W\in
    \class_{\Lambda\Sigma}(\Lambda A(\Lambda^{-1}\cdot))$ such that 
    $L^{(S,\Lambda)}\, P_\Sigma^V\,(L^{(S,\Lambda)})^{-1} - P_{\Lambda\Sigma}^W
    \in I_2(\cH_{\Lambda\Sigma})$. We remark that for the {\it linear form} $A$,
    $\Lambda A$ stands for the linear form with coordinates ${\Lambda_\mu}^\nu
    A_\nu$, while for a
{\it vector} $x$, the term $\Lambda x$ stands for the vector with coordinates
${\Lambda^\mu}_\nu x^\nu$.
    We take $\lambda\in\cG(A)$, e.g., $\lambda=\lambda^A$ from \eqref{eq:special-lambda}.
    Thanks to Lemma~\ref{lem:pol-class-plambda}, for all $V\in
    \class_\Sigma(A)$ 
    we have $P_\Sigma^V -
    P_\Sigma^\lambda\in I_2(\HSigma)$. 
    First, let us discuss how such a $P_\Sigma^\lambda$ behaves under the Lorentz
    transforms $L^{(S,\Lambda)}$.  
    For $\epsilon>0$ and 
    $u\in\operatorname{Past}$, the integral kernel 
    $p_\Sigma^{\lambda,\epsilon u}(x,y)=e^{-i\lambda(x,y)}p_-(y-x+i\epsilon u)$ of 
    $P_\Sigma^{\lambda,\epsilon u}$, cf.
    \eqref{eq:Pminus-lambda-epsilon}, transforms as follows:
    The integral kernel of 
    $L^{(S,\Lambda)}_\Sigma P_\Sigma^{\lambda,\epsilon u} (L^{(S,\Lambda)}_\Sigma)^{-1}$ 
    is given by
    \begin{align} 
        S p_\Sigma^{\lambda,\epsilon u}(\Lambda^{-1}x,\Lambda^{-1}y)  S^*
        & =
        e^{-i\lambda(\Lambda^{-1} x,\Lambda^{-1} y)} \, S \,
        p^-(\Lambda^{-1}(y-x)+i\epsilon u)) \, S^*
        \cr
        & =
        e^{-i\lambda(\Lambda^{-1} x,\Lambda^{-1} y)}
        p^-(y-x+i\epsilon \Lambda u)
        =
        p_{\Lambda\Sigma}^{\overline\lambda,\epsilon \Lambda  u}(x,y),
    \end{align}
    where $\overline \lambda(x,y)=\lambda(\Lambda^{-1}x,\Lambda^{-1}y)$.
    We claim $\overline \lambda\in\cG(\Lambda A(\Lambda^{-1}\cdot))$.
    Indeed, $\overline\lambda$
    clearly fulfills conditions (i)-(iii) of the Definition~\ref{def:PLambda} 
    of $\cG(\Lambda A(\Lambda^{-1}\cdot))$. It also fulfills condition (iv)
    since
    \begin{align}
        \frac{\partial}{\partial x^\mu} \overline\lambda(x,y) \big|_{y=x}
        & =
        \frac{\partial}{\partial x^\mu} \lambda(\Lambda^{-1}x,\Lambda^{-1}y)
        \big|_{y=x}
        =
        {(\Lambda^{-1})^\nu}_\mu
        \frac{\partial}{\partial z^\nu} \lambda(z,\Lambda^{-1}y)
        \big|_{z=\Lambda^{-1}x,y=x}
        \cr
        & =
        {\Lambda_\mu}^\nu A_\nu(\Lambda^{-1}x)
        \label{eq:p-lorentz-transformed}
    \end{align}
    and similarly
    $ \partial^y_\mu \overline\lambda(x,y) \big|_{x=y} =
    -{\Lambda_\mu}^\nu A_\nu(\Lambda^{-1}x)$,
    where we have used ${(\Lambda^{-1})^\nu}_\mu={\Lambda_\mu}^\nu$.
    This shows 
    $L^{(S,\Lambda)}_\Sigma P_\Sigma^{\lambda,\epsilon u}
    (L^{(S,\Lambda)}_\Sigma)^{-1} =  P_\Sigma^{\overline{\lambda},\epsilon
    \Lambda u},$
    which implies
    $L^{(S,\Lambda)}_\Sigma P_\Sigma^\lambda (L^{(S,\Lambda)}_\Sigma)^{-1} =  P_\Sigma^{\overline{\lambda}}$
    in the limit as $\epsilon\downarrow 0$; recall from 
    Lemma~\ref{lem:Pminus-lambda} that the limit does not depend on
    the choice of $u,\Lambda u\in\operatorname{Past}$.

    Again by
    Lemma~\ref{lem:pol-class-plambda}, there is a 
    $W\in \class_{\Lambda\Sigma}(\Lambda A(\Lambda^{-1}\cdot))$ such that
    $P_{\Lambda\Sigma}^W-P_{\Lambda\Sigma}^{\overline\lambda}\in
    I_2(\cH_{\Lambda\Sigma})$.
    We conclude
    \begin{align}
        L^{(S,\Lambda)}_\Sigma \, P_\Sigma^V \, \left(L^{(S,\Lambda)}_\Sigma\right)^{-1}
        -
        P_{\Lambda\Sigma}^W
        =
        L^{(S,\Lambda)}_\Sigma \, \left( P_\Sigma^V - P_\Sigma^\lambda \right) \,
        \left(L^{(S,\Lambda)}_\Sigma\right)^{-1}
        -
        \left( P_{\Lambda\Sigma}^W - P_{\Lambda\Sigma}^{\overline \lambda}
        \right)
        \in
        I_2(\cH_{\Lambda \Sigma}).
    \end{align}

    Claim (ii): The integral kernel of 
    $e^{-i\Omega}P_\Sigma^{\lambda,\epsilon u} e^{i\Omega}$ for $\lambda\in\cG(A)$,
    $\epsilon>0$ and $u\in\operatorname{Past}$ equals
    \begin{align}
        e^{-i\Omega(x)}p_\Sigma^{\lambda,\epsilon u}(x,y) e^{i\Omega(y)}
        =
        e^{-i\Omega(x)} e^{-i\lambda(x,y)}
        p^-(y-x+i\epsilon u) e^{i\Omega(y)}
        =
        p_\Sigma^{\overline\lambda,\epsilon u} (x,y),
    \end{align}
    where $\overline \lambda(x,y)=\Omega(x)+\lambda(x,y)-\Omega(y)$, which
    clearly fulfills $\overline\lambda\in \cG(A+\partial \Omega)$; 
    cf.\ Definition~\ref{def:PLambda}.
    Taking the limit as $\epsilon\downarrow 0$, 
    the claim follows from the same kind of reasoning as in part (i).
\end{proof}

Finally, one can also use the self-adjoint operator $P_\Sigma^A$ from
\eqref{eq:P-special-lambda} 
to construct a unitary operator
$e^{Q^A_\Sigma}:\HSigma\selfmaps$ 
which adapts the standard polarization $\HSigma^-$ to one corresponding to
$A|_{T\Sigma}$, more precisely, $e^{Q^A_\Sigma}\HSigma^-\in \class_\Sigma(A)$.
It is defined as follows:
\begin{definition}
We set
\begin{align}\label{eq:V_Sigma}
    Q_\Sigma^A
    &:=
    [P_\Sigma^A,P_\Sigma^-]=
    P_{\Sigma}^+ (P_\Sigma^A - P_\Sigma^-) P_{\Sigma}^-
    -
    P_{\Sigma}^- (P_\Sigma^A - P_\Sigma^-) P_{\Sigma}^+.
\end{align}
\end{definition}

\begin{proof}[Proof of Theorem~\ref{thm:representative}]
  In this proof, we use a $2\times 2$-matrix notation for
linear operators of the type $\HSigma\selfmaps$. This matrix notation
always refers to the splitting $\HSigma=\HSigma^+ \oplus \HSigma^-$.
In particular, we set
\begin{align}
        \begin{pmatrix}
            \Delta_{++} & \Delta_{+-} \\
            \Delta_{-+} & \Delta_{--}
        \end{pmatrix}
        =
        \Delta P_\Sigma^{\lambda^A}=P_\Sigma^A-P_\Sigma^-,
\end{align}
cf. \eqref{eq:pdiff} for $\lambda=\lambda^A$.
Using this matrix notation, we write
\begin{align}
  \label{eq:Q-as-matrix}
        Q_\Sigma^A=
        \begin{pmatrix}
            0 & \Delta_{+-} \\
            -\Delta_{-+} & 0
        \end{pmatrix}
.
\end{align}
In the following we use the notation $X=Y \mod I_2(\HSigma)$ to mean $X-Y\in
I_2(\HSigma)$. By (iii) of Lemma~\ref{lem:Pminus-lambda} we know
that $(\Delta P_\Sigma^{\lambda^A})^2\in I_2(\HSigma)$, and therefore
    \begin{align}
        (P_\Sigma^A)^2 
        =
        (P_\Sigma^-+\Delta P_\Sigma^{\lambda^A})^2 
        = P_\Sigma^A +
        \begin{pmatrix}
            -\Delta_{++} & 0 \\
            0 & \Delta_{--}
        \end{pmatrix}
        \mod I_2(\HSigma).
        \label{eq:PA-square-mod-I2}
    \end{align}
    Furthermore, Lemma~\ref{lem:pol-class-plambda} implies for all $V\in \class_\Sigma(A)$
    that the corresponding orthogonal projector $P_\Sigma^V$
    fulfills $P_\Sigma^A-P_\Sigma^V\in I_2(\HSigma)$.
    However, this means also that $(P_\Sigma^A)^2-P_\Sigma^A\in
    I_2(\HSigma)$, and therefore, $\Delta_{++},\Delta_{--}\in I_2(\HSigma)$;
    see \eqref{eq:PA-square-mod-I2}.
    In conclusion, we obtain 
    \begin{align}
        P_\Sigma^A=P_\Sigma^- + \Delta P_\Sigma^{\lambda^A} = 
        \begin{pmatrix}
            0 & \Delta_{+-} \\
            \Delta_{-+} & \id_{\HSigma^-}
        \end{pmatrix}
        \mod I_2(\HSigma).
    \end{align}
    Since $(\Delta P_\Sigma^{\lambda^A})^2\in I_2(\HSigma)$ we have 
    $\Delta_{-+}\Delta_{+-},\Delta_{-+}\Delta_{+-}\in
    I_2(\HSigma)$ and hence $(Q_\Sigma^A)^2\in I_2(\HSigma)$; cf. 
    \eqref{eq:Q-as-matrix}.
    Defining 
    \begin{align}
        \label{eq:def-Pi}
        \Pi_\Sigma^A:=e^{Q_\Sigma^A} P_\Sigma^- e^{-Q_\Sigma^A},
      \end{align}
    we conclude 
    \begin{align}
        \label{eq:Q-rep-step}
        \Pi_\Sigma^A&=
        (\id_{\HSigma}+Q_\Sigma^A)P_\Sigma^- (\id_{\HSigma}-Q_\Sigma^A)
        =        
        \begin{pmatrix}
            0 & \Delta_{+-} \\
            \Delta_{-+} & \id_{\HSigma^-}
        \end{pmatrix}
        = P_\Sigma^A=P_\Sigma^V 
        \mod I_2(\HSigma).
     \end{align}
    Furthermore, we observe that $e^{Q_\Sigma^A}$ is unitary because $Q_\Sigma^A$ 
    is skew-adjoint, so that $\Pi_\Sigma^A$ is an
    orthogonal projector.
    Summarizing, we have shown
    $e^{Q_\Sigma^A}\HSigma^-=\Pi_\Sigma^A\HSigma\in\class_\Sigma(A)$, 
    which proves the claim of
    Theorem~\ref{thm:representative}.\\

    As an addendum we prove the refinement of Theorem~\ref{thm:representative}
    described in Remark~\ref{rem:approx_0}. For this it is left to show
    that $\operatorname{charge}(U_{\Sigma\Sigma_\inn}^A
    \cH_{\Sigma_\inn}^-,\Pi_\Sigma^A\HSigma)=0$. We choose a future oriented foliation 
    $(\Sigma_t)_{t\in \R}$ 
    of space-time
    such that $\Sigma_0=\Sigma_\inn$ and $\Sigma_1=\Sigma$. Recall the choice of $\Sigma_\inn$ described in
    \eqref{eq:Sigma_inn}.
    The operators $Q_{\Sigma_t}^A$ are compact because they are
    skew-adjoint and $(Q_{\Sigma_t}^A)^2\in I_2(\cH_{\Sigma_t})$. 
    Hence, the operators $e^{-Q_{\Sigma_t}^A}$
    are compact perturbations of the identity operators $\id_{\cH_{\Sigma_t}}$.
    Translating this fact to an interaction picture,
    the operators 
    \begin{align}
        Q_t
        :=
        U_{\Sigma_\inn\Sigma_t}^0
        e^{-Q_{\Sigma_t}^A}
        U_{\Sigma_t\Sigma_\inn}^0
    \end{align}
    are as well compact perturbations of the identity operator
    $\id_{\cH_{\Sigma_\inn}}$. We define the evolution operators in the
    interaction picture 
    \begin{align}
        U_t 
        := 
        U_{\Sigma_\inn\Sigma_t}^0
            U_{\Sigma_t\Sigma_\inn}^A,
    \end{align}
    which are continuous in $t\in \R$ w.r.t.\ the operator norm; this follows
    from Lemma 3.9 in \cite{Deckert2014}. Moreover, using $V\approx W
    \, \Leftrightarrow \, P_\Sigma^V P_\Sigma^{W^\perp},P_\Sigma^{V^\perp}
    P_\Sigma^{W} \in I_2(\HSigma)$,  the just proven
    Theorem~\ref{thm:representative} implies
    \begin{align}
        e^{Q_\Sigma^A} \cH_\Sigma^- \approx U_{\Sigma\Sigma_\inn}^A \cH_{\Sigma_\inn}^-
        \Rightarrow \quad & 
        P_\Sigma^\pm e^{-Q_\Sigma^A} U_{\Sigma\Sigma_\inn}^A
        P_{\Sigma_\inn}^\mp \in I_2(\HSigma)
        \\
        \Rightarrow \quad & 
        U_{\Sigma_\inn\Sigma}^0 
        P_\Sigma^\pm e^{-Q_\Sigma^A} U_{\Sigma\Sigma_\inn}^A
        P_{\Sigma_\inn}^\mp \in I_2(\HSigma)
        \\
        \Rightarrow \quad & 
        P_{\Sigma_\inn}^\pm Q_t U_t P_{\Sigma_\inn}^\mp 
        =
        P_\Sigma^\pm 
        U_{\Sigma_\inn\Sigma}^0 
        e^{-Q_\Sigma^A} U_{\Sigma\Sigma_\inn}^A
        P_{\Sigma_\inn}^\mp
        \in
        I_2(\cH_{\Sigma_\inn}).
        \label{eq:comp_op}
    \end{align}
    Since $Q_t-\id_{\cH_{\Sigma_\inn}}$ is compact, the operator
    $P_{\Sigma_\inn}^\pm (Q_t - \id_{\cH_{\Sigma_\inn}}) U_t P_{\Sigma_\inn}^\mp$ 
    is compact as well. Taking the difference with the compact operator in
    \eqref{eq:comp_op} yields that
    $P_{\Sigma_\inn}^\pm U_t P_{\Sigma_\inn}^\mp$ is compact so that
    \begin{align}
        \begin{pmatrix}
            P_{\Sigma_\inn}^+ U_t P_{\Sigma_\inn}^+ & 0 \\
            0 & P_{\Sigma_\inn}^- U_t P_{\Sigma_\inn}^- 
        \end{pmatrix}
        =
        U_t 
        -
        \begin{pmatrix}
            0 & P_{\Sigma_\inn}^+ U_t P_{\Sigma_\inn}^- \\
            P_{\Sigma_\inn}^- U_t P_{\Sigma_\inn}^+ & 0
        \end{pmatrix}
    \end{align}
    deviates from the unitary operator $U_t$ by a compact perturbation, and
    hence, is  a Fredholm operator. This implies that $P_{\Sigma_\inn}^- U_t
    P_{\Sigma_\inn}^-\big|_{\cH_{\Sigma_\inn}^-\selfmaps}$ is a Fredholm
    operator.  We note that the Fredholm index of $P_{\Sigma_\inn}^- U_{t=0}
    P_{\Sigma_\inn}^-\big|_{\cH_{\Sigma_\inn}^-\selfmaps}=\id_{\cH_{\Sigma_\inn}^-}$ 
    equals zero. The map 
    $t\mapsto P_{\Sigma_\inn}^- U_t P_{\Sigma_\inn}^-$ is continuous in the
    operator norm which implies that the Fredholm index is constant, and hence,
    \begin{multline}
        0 
        = 
        \operatorname{index}P_{\Sigma_\inn}^- U_{t=1}
        \big|_{\cH_{\Sigma_\inn}^-\selfmaps}
        =
        \operatorname{index}
        P_{\Sigma_\inn}^- U_{\Sigma_\inn\Sigma}^0 U^A_{\Sigma\Sigma_\inn}
        \big|_{\cH_{\Sigma_\inn}^-\selfmaps}
        \\
        =
        \operatorname{index}
        P_\Sigma^- U^A_{\Sigma\Sigma_\inn}
        \big|_{\cH_{\Sigma_\inn}^-\to
            \cH_\Sigma^-}
        =
        \operatorname{index}
        P_\Sigma^- 
        \big|_{U_{\Sigma\Sigma_\inn}^A \cH_{\Sigma_\inn}^-\to
            \cH_\Sigma^-}
        =
        \operatorname{index}
        P_\Sigma^- e^{-Q_\Sigma^A}
        \big|_{U_{\Sigma\Sigma_\inn}^A \cH_{\Sigma_\inn}^-\to
            \cH_\Sigma^-}
        \\
        =
        \operatorname{index}
        e^{Q_\Sigma^A} P_\Sigma^- e^{-Q_\Sigma^A}
        \big|_{U_{\Sigma\Sigma_\inn}^A \cH_{\Sigma_\inn}^-\to
            \Pi_\Sigma^A\HSigma}
        =
        \operatorname{charge}
        (U_{\Sigma\Sigma_\inn}^A \cH_{\Sigma_\inn}^-,\Pi_\Sigma^A\HSigma),
    \end{multline}
    where in the fifth equality we have used that
    $e^{-Q_\Sigma^A}$ is a compact perturbation of the identity.
\end{proof}

This concludes the proofs of the main results under the condition that the
claims in Section~\ref{sec:operator_PA} are true. The proofs of 
these claims will be provided in the next two sections.

\subsection{Proof of Lemma~\ref{lemma:kern-p-minus}, Lemma~\ref{lem:Pminus-lambda},
and Theorem~\ref{thm:equivalence-pol-classes}}
\label{sec:pol-classes}

\begin{proof}[Proof of Lemma \ref{lemma:kern-p-minus}.]
    Given $\phi,\psi\in\CSigma$, we set $\widehat{\phi}=\cF_{\cM \Sigma}\phi$
    and $\widehat{\psi}=\cF_{\cM \Sigma}\psi$ where $\cF_{\cM\Sigma}$ is the
    generalized Fourier transform
    \begin{align}
        \label{def FMSigma}
        (\cF_{\cM\Sigma}\psi)(p)&=\frac{\slashed p+m}{2m}(2\pi)^{-3/2}
        \int_\Sigma e^{ipx}\,i_\gamma(d^4x)\,\psi(x)
        &&\text{for } \psi\in\CSigma, p\in\cM,
    \end{align}
    introduced in 
    Theorem~2.15 of
    \cite{Deckert2014}. This theorem ensures
    that $\overline{\widehat{\phi}(p)}\widehat{\psi}(p)\,
    i_p(d^4p)$ is integrable
    on $\cM_-$.
    Let $u\in\operatorname{Past}$. 
    With justifications given below, we compute the following.
    \begin{align}
        &\sk{\phi,P_\Sigma^-\psi}=
        \lim_{\epsilon\downarrow 0}\int\limits_{p\in\cM_-} e^{-\epsilon pu}
        \overline{\widehat{\phi}(p)}\widehat{\psi}(p)\frac{i_p(d^4p)}{m}
        \label{eq:pminus1}
        \\&=
        \frac{1}{(2\pi)^3m}\lim_{\epsilon\downarrow 0}\int\limits_{p\in\cM_-} e^{-\epsilon pu}
        \int\limits_{x\in\Sigma} \overline{\phi(x)} \,i_\gamma(d^4x)\, e^{-ipx}
        \left(\frac{\slashed p+m}{2m}\right)^2
        \int\limits_{y\in\Sigma}e^{ipy}\,i_\gamma(d^4y)\, \psi(y)\,i_p(d^4p)
        \label{eq:pminus2}
        \\&=
        \frac{1}{(2\pi)^3m}\lim_{\epsilon\downarrow 0}
        \int\limits_{p\in\cM_-}
        \int\limits_{x\in\Sigma} \overline{\phi(x)} \,i_\gamma(d^4x)\, 
        \frac{\slashed p+m}{2m}
        \int\limits_{y\in\Sigma}e^{ip(y-x+i\epsilon u)}\,i_\gamma(d^4y)\, \psi(y)\,i_p(d^4p)
        \label{eq:pminus3}
        \\&=
        \frac{1}{(2\pi)^3m}\lim_{\epsilon\downarrow 0}
        \int\limits_{x\in\Sigma} \overline{\phi(x)} \,i_\gamma(d^4x)\, 
        \int\limits_{y\in\Sigma}\int\limits_{p\in\cM_-}
        \frac{\slashed p+m}{2m}
        e^{ip(y-x+i\epsilon u)}\,i_p(d^4p)\,i_\gamma(d^4y)\, \psi(y)
        \label{eq:pminus4}
        \\&
        =\lim_{\epsilon\downarrow 0}
        \int\limits_{x\in\Sigma} \overline{\phi}(x) \, i_\gamma(d^4x)
        \int\limits_{y\in\Sigma} p^-(y-x+i\epsilon u) \, i_\gamma(d^4y) \, \psi(y).
        \label{eq:pminus5}
    \end{align}
    The interchange of the $p$-integral and the limit $\epsilon\downarrow 0$
    in \eqref{eq:pminus1} is justified by dominated convergence
    since
    $\overline{\widehat{\phi}(p)}\widehat{\psi}(p) i_p(d^4p)$ is integrable
    on $\cM_-$ and by $|e^{-\epsilon pu}|\le 1$ for $\epsilon>0$, $p\in\cM_-$.
    In the step from \eqref{eq:pminus1} to \eqref{eq:pminus2} we have used
    \eqref{def FMSigma} and that $\gamma^0(\gamma^\mu)^*\gamma^0=\gamma^\mu$,
    from \eqref{eq:pminus2} to \eqref{eq:pminus3} that $\slashed p^2=p^2$ and
    that $p^2=m^2$ for $p\in\cM_-$.
    In the step from \eqref{eq:pminus3} to \eqref{eq:pminus4} we
    have used Fubini's theorem to interchange the integrals.
    This is justified because $\phi$ and $\psi$ are bounded and compactly
    supported, and because for any given $\epsilon>0$,
    $|e^{ip(y-x+i\epsilon u)}|=e^{-\epsilon pu}$ tends exponentially fast to $0$ 
    as $|p|\to\infty$, $p\in\cM_-$. This proves the claim
    \eqref{eq:pminus-claim}. \\
   
    Now we prove the claimed properties of $D$ and $p^-$.  For any
    $w\in\R^4+i\operatorname{Past}$, the modulus $|e^{ipw}|=e^{-p \im w}$ tends
    exponentially fast to $0$ as $|p|\to\infty$, $p\in\cM_-$. Consequently,
    exchanging differentiation and integration in the following calculation is
    justified:
    \begin{align}
        &p^-(w)=         
        \frac{1}{(2\pi)^3m}
        \int_{\cM_-} \frac{-i\slashed\partial^w+m}{2m} e^{ipw} \, i_p(d^4p)
        \cr&
        =         
        \frac{1}{(2\pi)^3m}
        \frac{-i\slashed\partial^w+m}{2m} \int_{\cM_-} e^{ipw} \, i_p(d^4p)
        =\frac{-i\slashed\partial+m}{2m} D(w).
    \end{align} 
    To show the second equality in \eqref{eq:def-D}, we proceed as follows:
    First, we show that $w\in \R^4+i\operatorname{Past}$
    implies $-w_\mu w^\mu\in\C\setminus\R_0^-=\operatorname{domain}(\sqrt{\cdot})$.
    We take $w=z+iu$ with 
    $z\in\R^4$ and $u\in \operatorname{Past}$, 
    and assume $-w_\mu w^\mu\in\R$.
    Then
    $0=\im(w_\mu w^\mu)=2z_\mu u^\mu$, i.e., $z$ is orthogonal to $u$ in 
    the Minkowski sense. Because $u$ is time-like, we conclude that
    $z$ is space-like or zero.
    We obtain $w_\mu w^\mu= \re(w_\mu w^\mu)=z_\mu z^\mu-u_\mu u^\mu<0$,
    i.e., $-w_\mu w^\mu\in\operatorname{domain}(\sqrt{\cdot})$.
    It follows that $\sqrt{-w_\mu w^\mu}\in \R^++i\R=\operatorname{domain}(K_1)$.
    In particular,
    \begin{align}
        \label{eq:D-version2}
        \widetilde{D}:\R^4+i\operatorname{Past}\ni w\mapsto
        -\frac{m^3}{2 \pi^2} \frac{K_1(m\sqrt{-w_\mu w^\mu})}{m\sqrt{-w_\mu w^\mu}}
    \end{align}
    is a well-defined holomorphic function.
    Because $|e^{ipw}|$ decays fast as $|p|\to\infty$, $p\in\cM_-$,
    uniformly for $w$ in any compact subset
    of $\R^4+i\operatorname{Past}$,
    \begin{align}
        \label{eq:D-version1}
        D:\R^4+i\operatorname{Past}\ni w\mapsto
        \frac{1}{(2\pi)^3m}
        \int_{\cM_-} e^{ipw} \, i_p(d^4p)
    \end{align}
    is also a holomorphic function. 
    We need to show $D=\widetilde{D}$.
    By the identity theorem for holomorphic functions, 
    it suffices to show that
    the restrictions of $D$ and $\widetilde{D}$ to $i\operatorname{Past}$
    coincide.
    Given $w=iu\in i\operatorname{Past}$, 
    we choose a proper, orthochronous Lorentz transform 
    $\Lambda\in\operatorname{SO}^\uparrow(1,3)\subseteq\R^{4\times 4}$ that maps
    $u$ to the negative time axis:
    \begin{align}
        \Lambda u=-te_0=(-t,0,0,0)\text{ with } t=\sqrt{u_\mu u^\mu}=\sqrt{-w_\mu w^\mu}>0.
    \end{align}
    By Lorentz invariance of the volume-form $i_p(d^4p)$
    on $\cM_-$,
    we know
    \begin{align}
        \int_{\cM_-}e^{ipw}\,i_p(d^4p)=\int_{\cM_-}e^{ip\Lambda w}\,i_p(d^4p)
    \end{align}
    and
    $\sqrt{-w_\mu w^\mu}=\sqrt{-(\Lambda w)_\mu(\Lambda w)^\mu}$.
    Summarizing, we have reduced the claim $D=\widetilde D$
    to its special case $D(w)=\widetilde{D}(w)$ for
    $w=-ite_0$, $t=\sqrt{-w_\mu w^\mu}>0$.
    This special case is proven as follows.
    Using 
    \begin{align}
        i_p(d^4p)=\frac{m^2}{p^0}d^3p 
        \text{ on } (T_p \cM)^3,
    \end{align}
    rotational symmetry, and the substitution 
    \begin{align}
        s=\frac{\sqrt{k^2+m^2}}{m},
        \quad 
        k=m\sqrt{s^2-1},
        \quad
        m^2 s\,ds=k\,dk,
    \end{align}
    we obtain with the abbreviation $E(\vec p)=\sqrt{\vec p^2+m^2}$:
    \begin{align}
        &\int_{\cM_-}e^{ipw}\,i_p(d^4p)
        =
        -m^2 \int_{\R^3} e^{-E({\bf p}) t}\frac{d^3{\bf p}}{E({\bf p})}
        \cr&=
        -4\pi m^2 \int_0^\infty \exp\left(-t\sqrt{k^2+m^2}\right) 
        \frac{k^2\,dk}{\sqrt{k^2+m^2}}
        \cr&=
        -4\pi m^4 \int_1^\infty e^{-mts}\sqrt{s^2-1}\,ds 
        =
        -4\pi m^4 \frac{K_1(mt)}{mt},
    \end{align}
    using the definition of $K_1$ in \eqref{eq:K1},
    and hence, the claim $D(-ite_0)=\widetilde{D}(-ite_0)$.

    The representation \eqref{eq:D-version2} of $D$ shows also that
    $D$ can be analytically extended to all arguments $w\in\C^4$ with 
    $-w_\mu w^\mu\in\operatorname{domain}(\sqrt{\cdot})=\C\setminus\R_0^-$.
    The same holds true for $p^-=(2m)^{-1}(-i\slashed\partial+m)D$.
    To sum up, $p^-$ has an analytic continuation
    $p^-:\operatorname{domain}(r)\to\C^{4\times 4}$,
    which also concludes the proof of Lemma~\ref{lemma:kern-p-minus}.
\end{proof}

\ifx\arxiv\undefined
\else
We remark that the modified Bessel function $K_1$ has an analytic continuation
to the Riemann surface of the logarithm, although this is not used in this
paper. However, the analytic continuation of $K_1$ to $\C\setminus\R^+_0$ is
defined in \eqref{eq:K1-1} in the appendix.
\fi

\begin{proof}[Proof of Lemma~\ref{lem:Pminus-lambda}]
We remark that most of the arguments in this proof are valid 
without regularization, i.e., also in the case $\epsilon=0$. This is
in contrast to Section \ref{sec:evolution-polarization-classes} below,
where the regularization with $\epsilon>0$ turns out to be very useful.

    Let $A\in\cC^\infty_c(\R^4,\R^4)$, $\lambda\in\cG(A)$, and $\Sigma$ be a Cauchy
    surface.
    Before proving the claim
    \eqref{eq:Pminus-lambda}-\eqref{eq:Pminus-lambda-epsilon}
    it will be convenient to introduce the operators $\Delta
    P_\Sigma^{\lambda,\epsilon u}$, $\epsilon\geq 0$, which shall 
    act on any $\psi\in\HSigma$ as
    \begin{align}
        \label{eq:deltap-epsilon}
        \left(\Delta P_\Sigma^{\lambda,\epsilon u}\psi\right)(x)=
        \int_{y\in\Sigma} (e^{-i\lambda(x,y)}-1)
        p^-(y-x+i\epsilon u) \, i_\gamma(d^4y)
        \, \psi(y),
    \end{align}
    where the fixed vector $u\in\R^4$ is past-directed time-like.
    We remark that the special case $\epsilon=0$ is included in the form
    $\Delta P_\Sigma^{\lambda,0}=\Delta P_\Sigma^{\lambda}$; cf.~\eqref{eq:pdiff}.
   
    We show now that 
    $\Delta P_\Sigma^{\lambda,\epsilon u}:\HSigma\selfmaps$ is well-defined.
    Recall the parametrization $\pi_\Sigma(\vec x)$ of $\Sigma$ as stated in
    \eqref{eq:parametrize Sigma} and the identity $i_\gamma(d^4x)=\Gamma(\vec
    x)\,d^3x$ on $(T_x\Sigma)^3$ given in \eqref{eq: repr igammad4x}.  
    We use the
    abbreviation $x=\pi_\Sigma(\vec x)$, $y=\pi_\Sigma(\vec y)$ in the following.
    Line
    \eqref{eq:deltap-epsilon} can be recast into
    \begin{align} 
        \left(\Delta P_\Sigma^{\lambda,\epsilon u}\psi\right)(x)
        & =
        \int_{\R^3}
        \Delta p_\Sigma^{\lambda,\epsilon u}(\vec x,\vec y)
        \, \Gamma(\vec y) 
        \, \psi(y)
        \, d^3\vec y 
        \quad
        \text{for}
        \label{eq:deltap-epsilon-2}
        \\
        \Delta p_\Sigma^{\lambda,\epsilon u}(\vec x,\vec y)
        &:=
        \left(e^{-i\lambda(x,y)}-1\right)
        p^-(y-x+i\epsilon u) 
        .
        \label{eq:Pminus-deltaP-kernel}
    \end{align}
    To show at the same time that
    the right-hand side of \eqref{eq:deltap-epsilon-2}, i.e.,
    \eqref{eq:deltap-epsilon}, is well-defined for $\psi\in\HSigma$ 
    and almost every 
    $x\in\Sigma$, and that $\Delta P_\Sigma^{\lambda,\epsilon}\psi\in\HSigma$,
    it suffices to prove that for every $\phi\in\HSigma$, we have
    \begin{align}
        \int_{\vec x\in\R^3}\int_{\vec y\in\R^3}\left|\overline{\phi(x)}\Gamma(\vec x) 
        \Delta p_\Sigma^{\lambda,\epsilon u}(\vec x,\vec y)
        \Gamma(\vec y) 
        \psi(y)\right|
        \, d^3\vec y \le\constl{c:Pminus-lambda-3}
        \| \phi \| \| \psi \|
        \label{eq:claim-DeltaP-bounded}
    \end{align}
 with some constant $\constr{c:Pminus-lambda-3}(u,\Vmax)$.
 We collect the necessary ingredients:
    \begin{itemize}
        \item 
            As $\lambda$ is smooth and vanishes on the diagonal,
            there is a positive constant $\constl{c:lambda}(\lambda)$ such that
            \begin{align}
                \label{eq:first-order-exp-lambda-bound}
                |e^{-i\lambda(x,y)}-1|\le\constr{c:lambda}|x-y|[1_K(x)\vee 1_K(y)]
                \text{ for } x,y\in\R^4. 
            \end{align}
            Note that this bound holds globally, not only locally
            close to the diagonal, because $e^{-i\lambda}-1$ is bounded
            and vanishes outside $K\times \R^4\cup\R^4\times K$ for some compact
            set $K$.
        \item 
            The bounds \eqref{eq:rz-z} from the appendix, 
            cf.~\eqref{eq:bound-grad_tSigma}, show 
            that for all $x,y\in\Sigma$
            and $(z^0,\vec z)=z=y-x$ we find
            $|\vec z|
            \leq
            |z|
            \leq \sqrt{1+\Vmax^2} |\vec z|$.
        \item
            Formula \eqref{eq:bound-pminus-R3} in 
            Corollary~\ref{lem:upper-bounds} of the Appendix
            ensures for all $\epsilon\geq 0$ that 
            for all $z=(z^0,\vec z)$ such that $z = y-x$
            for $x,y\in\Sigma$ and $\vec z\neq 0$ that
            \begin{align}
                \label{eq:majorante}
                \| p^-( z+i\epsilon u ) \| 
                \leq 
                O_{u,\Vmax}\left(
                \frac{e^{- \cD |\vec z|}}
                {|\vec z |^3}\right).
            \end{align}
    \end{itemize}
    Thanks to these ingredients we find the estimate 
    \begin{align}
        \| \Delta p_\Sigma^{\lambda,\epsilon u}(\vec x,\vec y) \|
        \leq
        \constl{c:Pminus-lambda-0}
        \frac{e^{-\cD |\vec y-\vec x|}}{|\vec y-\vec x|^2}[1_K(x)\vee 1_K(y)]
        \label{eq:Pminus-deltaP-kernel-bound}
    \end{align}
    for all $x,y\in\Sigma$ such that $\vec y-\vec x\neq 0$ and 
    $\epsilon\geq 0$ with some constant 
    $\constr{c:Pminus-lambda-0}(u,\Vmax,\lambda)$.
    Consequently, 
    using the bound for $\Gamma$ from
    \eqref{eq:bound-Gamma},
    we have the dominating function
    \begin{align}
      \sup_{\epsilon\ge 0}\left|\overline{\phi(x)}\Gamma(\vec x) 
      \Delta p_\Sigma^{\lambda,\epsilon u}(\vec x,\vec y)
      \Gamma(\vec y) 
      \psi(y)\right|
      \le
      \constr{c:Pminus-lambda-0}\GammaMax^2|\phi(x)|
      \frac{e^{-\cD |\vec y-\vec x|}}{|\vec y-\vec x|^2}
        |\psi(y)|,
        \label{eq:dominating-function}
      \end{align}
      which is integrable, as the following calculation shows:
    \begin{align}
        & 
        \constr{c:Pminus-lambda-0}\GammaMax^2\int_{\vec x\in\R^3} \int_{\vec y\in\R^3}
        |\phi(x)|
        \frac{e^{-\cD |\vec y-\vec x|}}{|\vec y-\vec x|^2}
        |\psi(y)|
        \, d^3\vec y \, d^3\vec x
        \label{eq:Pminus-lambda-4-0}
        \\
        &=
        \constr{c:Pminus-lambda-0}\GammaMax^2
        \int_{\vec z\in\R^3} 
        \frac{e^{-\cD |\vec z|}}{|\vec z|^2}
        \int_{\vec x\in\R^3}
        |\phi(\pi_\Sigma(\vec x))|
        |\psi(\pi_\Sigma(\vec x+\vec z))|
        \, d^3\vec x \, d^3\vec z
        \label{eq:Pminus-lambda-4}
        \\
        & \leq
        4\pi \constr{c:Pminus-lambda-0}\GammaMax^2
        \int_0^\infty
        e^{-\cD s}
        \, 
        ds
        \,
        \| \phi \circ \pi_\Sigma \|_2 
        \| \psi \circ \pi_\Sigma \|_2
        \label{eq:Pminus-lambda-5}
        \\
        & \leq
        \constr{c:Pminus-lambda-3}
        \| \phi \| \, \| \psi \|,
        \label{eq:Pminus-lambda-6}
    \end{align}
    for a constant
    $\constr{c:Pminus-lambda-3}(u,\Vmax,\lambda)$. In the step from 
    \eqref{eq:Pminus-lambda-4} to
    \eqref{eq:Pminus-lambda-5}
    we use the
    Cauchy-Schwarz inequality, and in the step from \eqref{eq:Pminus-lambda-5}
    to \eqref{eq:Pminus-lambda-6}, we use
    that the norms $\|{\cdot}\circ\pi_\Sigma \|_2$ and $\|{\cdot}\|$ are
    equivalent.
    On the one hand, this proves claim~\eqref{eq:claim-DeltaP-bounded},
    which implies that the operators
    $\Delta P_\Sigma^{\lambda,\epsilon u}:\HSigma\selfmaps$ 
    described in \eqref{eq:deltap-epsilon-2} and \eqref{eq:Pminus-deltaP-kernel}
    are well-defined for all $\epsilon\ge 0$ 
    and bounded by
    \begin{align}
      \sup_{\epsilon\ge 0}\|\Delta P_\Sigma^{\lambda,\epsilon u}\|_{\HSigma\selfmaps}
        \le \constr{c:Pminus-lambda-3}.
        \label{eq:bound-Delta-P}
    \end{align}
    On the other hand, we use again the integrable domination from 
    \eqref{eq:dominating-function}
    together with the point-wise convergence
    \begin{align}
      \lim_{\epsilon\downarrow 0} p^-(y-x+i\epsilon u)=p^-(y-x)
    \end{align}
    for $x,y\in\Sigma$ with $x\neq y$;
    cf.~the analytic continuation of $p^-$ 
    described in Lemma~\ref{lemma:kern-p-minus}. 
    Using these ingredients, the dominated convergence theorem
    yields the following convergence in the weak operator topology:
    \begin{align}
      \sk{\phi,\Delta P_\Sigma^{\lambda,\epsilon u}\psi}
      \konv{\epsilon\downarrow 0} \sk{\phi,\Delta P_\Sigma^\lambda\psi}
\text{ for }\phi,\psi\in\HSigma.
\end{align}
The next argument needs this fact only restricted 
to $\phi,\psi\in\CSigma$.
Using the notation~\eqref{eq:Pminus-lambda-epsilon}
and Lemma~\ref{lemma:kern-p-minus}, we get for  $\phi,\psi\in\CSigma$
\begin{align}
\sk{\phi,P_\Sigma^{\lambda,\epsilon u}\psi}
=\sk{\phi,P_\Sigma^{0,\epsilon u}\psi}
+\sk{\phi,\Delta P_\Sigma^{\lambda,\epsilon u}\psi}
  \konv{\epsilon\downarrow 0}
\sk{\phi,P_\Sigma^-\psi}
+\sk{\phi,\Delta P_\Sigma^\lambda\psi}
.
\end{align}
Because  $P_\Sigma^-,\Delta P_\Sigma^\lambda:\HSigma\selfmaps$
are bounded operators and $\CSigma$ is dense in $\HSigma$,
this implies that
\begin{align}
P_\Sigma^\lambda:=P_\Sigma^-+\Delta P_\Sigma^\lambda:\HSigma\selfmaps
\end{align}
is the unique bounded operator that satisfies~\eqref{eq:Pminus-lambda},
together with the bound
\begin{align}
\|P_\Sigma^\lambda\|_{\HSigma\selfmaps}\le
\|P_\Sigma^-\|_{\HSigma\selfmaps}+
\|\Delta P_\Sigma^\lambda\|_{\HSigma\selfmaps}\le 
1+\constr{c:Pminus-lambda-3}(u,\Vmax,\lambda)
\end{align}
coming from \eqref{eq:bound-Delta-P}.
Note that we may take any fixed $u\in\operatorname{Past}$,
e.g., $u=(-1,0,0,0)$, in this bound and in the bounds below.

    Next, we show that 
    $K^{\lambda}:=|\Delta P_\Sigma^\lambda|^2$ is a
    Hilbert-Schmidt operator. 
    It is the integral operator (here written in 3-vector notation)
    \begin{align}
      K^\lambda\psi(x)=
      \int_{\R^3}k^\lambda(\vec x,\vec y)\Gamma(\vec y)\psi(y)d^3\vec y
    \end{align}
    for $\psi\in\HSigma$ and almost all $x\in\Sigma$
    with the integral kernel
    \begin{align}
      k^\lambda(\vec x,\vec y)
      =\int_{\R^3}
      \gamma^0 \Delta p_\Sigma^{\lambda,0}(\vec x,\vec z)^*\gamma^0
      \Gamma(\vec z)
      \Delta p_\Sigma^{\lambda,0}(\vec z,\vec y)
      \, d^3\vec z.
    \end{align}
We remark that under the symmetry assumption 
$\lambda(x,y)=-\lambda(y,x)$, we have
\begin{align}
\gamma^0 \Delta p_\Sigma^{\lambda,0}(\vec x,\vec z)^*\gamma^0
=\Delta p_\Sigma^{\lambda,0}(\vec z,\vec x);
\end{align}
cf.~formula \eqref{eq:hermitean} below.
    Thanks to the estimate \eqref{eq:Pminus-deltaP-kernel-bound} we find
    \begin{align} 
        \left\| k^\lambda(\vec x,\vec y)
        \right\|
        \leq
        \GammaMax \constr{c:Pminus-lambda-0}^2
        \int_{\R^3}
        \frac{e^{ - \cD |\vec x - \vec z| }}
        {|\vec x - \vec z |^2}
        \frac{e^{ - \cD |\vec z - \vec y| }}
        {|\vec z - \vec y |^2}
        (1_K(x)\vee 1_K(z))
        (1_K(z)\vee 1_K(y))
        \, d^3\vec z.
        \label{eq:bound-klambda}
    \end{align}
    Next, we use the bound
    \begin{align}
      e^{ - \cD |\vec x - \vec z| }e^{ - \cD |\vec z - \vec y| }
      (1_K(x)\vee 1_K(z))
      (1_K(z)\vee 1_K(y))
      \le 
      \constl{c:geometrie-vorfaktor}
      e^{ - \cD(|\vec y - \vec x| +|\vec x|)/2}
      \label{eq:geometric-bound}
    \end{align}
    with the constant 
    $\constr{c:geometrie-vorfaktor}(\lambda,\Vmax)=\sup_{z\in K} 
    e^{\cD|\vec z|/2}$.
    Substituting this bound in \eqref{eq:bound-klambda}
    and carrying out the integration yields
    \begin{align}
      &\| k^\lambda(\vec x,\vec y)\|
      \leq
      \GammaMax \constr{c:Pminus-lambda-0}^2
      \constr{c:geometrie-vorfaktor}
      e^{ - \cD(|\vec y - \vec x| +|\vec x|)/2}
      \int_{\R^3}
      \frac{d^3\vec z}
      {|\vec x - \vec z |^2|\vec z - \vec y |^2}
      =
      \constl{c:plambda-foliation1}
      \frac{e^{ - \cD(|\vec y - \vec x| +|\vec x|)/2}}
      {|\vec y - \vec x |}
    \end{align}
    for a finite constant
    $\constr{c:plambda-foliation1}(\lambda,\Vmax)$. We can 
    therefore bound the Hilbert-Schmidt norm of $K^\lambda$ as follows:
   \begin{align}
        &\| K^\lambda\|_{I_2(\HSigma)}^2
        =
        \int_{\R^3}\int_{\R^3}\operatorname{trace}[\gamma^0
        k^\lambda(\vec x,\vec y)^*\gamma^0\Gamma(\vec x)
        k^\lambda(\vec x,\vec y)\Gamma(\vec y)]\,d^3\vec x\,d^3\vec y
        \cr&
        \leq
        4\GammaMax^2
        \int_{\R^3}
        \int_{\R^3}
        \| k^\lambda(\vec x,\vec y)\|^2
        \, d^3\vec x
        \, d^3\vec y
        \cr&
        \leq
        4\GammaMax^2\constr{c:plambda-foliation1}^2
        \int_{\R^3}
        \int_{\R^3}
            \frac{e^{ - \cD(|\vec y - \vec x| +|\vec x|)}}
            {|\vec y - \vec x |^2}
        \, d^3\vec x
        \, d^3\vec y
        <
        \infty.
    \end{align}
    This proves that $K^\lambda=|\Delta P_\Sigma^\lambda |^2$ 
    is a Hilbert-Schmidt operator, and
    therefore, $\Delta P_\Sigma^\lambda$ is compact. 

    To prove part (iv) of Lemma~\ref{lem:Pminus-lambda}, we assume
    $\lambda(x,y)=-\lambda(y,x)$ for all $x,y\in\Sigma$.
    From the symmetries $D(w^*)=D(w)^*$ 
    and $D(-w)=D(w)$ for all 
    $w\in
    \operatorname{domain}(r)$ and $(\gamma^\mu)^*=\gamma^0\gamma^\mu\gamma^0$,
    we conclude
    \begin{align}
        p^-(-w^*)=\gamma^0p^-(w)\gamma^0,
    \end{align}
    and hence, using the assumed symmetry of $\lambda$,
    \begin{align} 
        \gamma^0\left(e^{-i\lambda(y,x)}p_-(y-x+i\epsilon u)\right)^*\gamma^0
        =e^{-i\lambda(x,y)}p_-(x-y+i\epsilon u)
        \label{eq:hermitean}
    \end{align}
    for $x,y\in\Sigma$, $\epsilon>0$ and $u\in\operatorname{Past}$.  Substituting
    this in the
    specification~\eqref{eq:Pminus-lambda}-\eqref{eq:Pminus-lambda-epsilon} of
    $P_\Sigma^\lambda$, it follows that $P_\Sigma^\lambda$ is self-adjoint and
    concludes the proof.  
\end{proof}

\begin{proof}[Proof of Theorem \ref{thm:equivalence-pol-classes}.]
    To show the equivalence we need to control of the kernel of
    $P_\Sigma^\lambda-P_\Sigma^{\widetilde\lambda}$ from above and from below.
    Let $\Delta \vec A:\R^3\to\R^3$ be the vector field on $\R^3$ with
    \begin{align}
        \Delta \vec A(\vec x)\cdot \vec{z}= (A_\mu(x)-\widetilde{A}_\mu(x))z^\mu
    \end{align}
    for any $x=(x^0,\vec x)\in\Sigma$ and $z=(z^0,\vec{z})\in T_x\Sigma$.
    Then for any $x=(x^0,\vec x)\in\Sigma$, 
    $A(x)|_{T_x\Sigma}=\widetilde A(x)|_{T_x\Sigma}$
    holds if and only if $\Delta \vec A(\vec x)=0$.
    From $\lambda\in\mathcal{G}(A)$ and  $\lambda\in\mathcal{G}(\widetilde A)$,
    see Definition \ref{def:PLambda},
    we get the Taylor expansions
    \begin{align}
        e^{-i\lambda(x,y)}&=
        1+iA_\mu(x)(y^\mu-x^\mu)+O_\lambda(|x-y|^2)(1_K(x)\vee 1_K(y)),
        \\
        e^{-i\widetilde\lambda(x,y)}&=
        1+i\widetilde{A}_\mu(x)(y^\mu-x^\mu)+O_{\widetilde\lambda}(|x-y|^2)(1_K(x)\vee
        1_K(y)),
        \\
        y^0-x^0 &= \nabla t_\Sigma(\vec x)\cdot(\vec y - \vec x) +
        O_\Sigma(|\vec x-\vec y|^2)
        &
    \end{align}
    for $y,x\in\Sigma$ from which we
    conclude
    \begin{align}
        e^{-i\lambda(x,y)}-e^{-i\widetilde\lambda(x,y)}=
        i\Delta \vec A(\vec x)\cdot (\vec{y}-\vec{x})+
        \remainl{r}{r:Delta-expilambda}(\vec x,\vec y)
    \end{align}
    with an error term $\remainr{r}{r:Delta-expilambda}$
    that fulfills for any $x,y\in\Sigma$ 
    \begin{align}
        \label{eq:bound-error-R}
        |\remainr{r}{r:Delta-expilambda}(\vec x,\vec y)|\le
        O_{\lambda,\widetilde \lambda, \Vmax}(|\vec x-\vec y|^2)
        \left( 1_{K}(x) \vee 1_{K}(y) \right),
    \end{align}
    where we used $|x-y| = O_{\Vmax}(|\vec x-\vec y|)$
    due to \eqref{eq:bound-grad_tSigma}.
    Note that the bound \eqref{eq:bound-error-R}
    holds not only locally near the diagonal but also {\it
    globally} for $x,y\in\Sigma$ because
    $e^{-i\lambda}-e^{-i\widetilde\lambda}$ is bounded and $\lambda$ and
    $\widetilde\lambda$ vanish outside $K\times\R^4\cup\R^4\times K$ for some
    compact set $K\subset\R^4$.  For $\phi,\psi\in \HSigma$ 
    formula \eqref{eq:pdiff} from
    Lemma~\ref{lem:Pminus-lambda} implies
    \begin{align}
        &\sk{\phi,(P_\Sigma^\lambda-P_\Sigma^{\widetilde{\lambda}})\psi} 
        \cr
        & =
        \int_{x\in\Sigma} \overline{\phi}(x) \, i_\gamma(d^4x)
        \int_{y\in\Sigma} (e^{-i\lambda(x,y)}-e^{-i\widetilde\lambda(x,y)})
        p^-(y-x) \, i_\gamma(d^4y)
        \, \psi(y)
        \cr
        &=
        \int_{\vec x\in\R^3} \int_{\vec y\in\R^3} \phi(x)^*\gamma^0\Gamma(\vec x)
        [t_1(x,y)+t_2(x,y)]\gamma^0\Gamma(\vec y)
        \psi(y)\,d^3\vec y\,d^3\vec x  
        \label{eq:Delta-Pminus-diff}
      \end{align}
    with
    \begin{align}
        t_1(x,y)&=
        i\Delta \vec A(\vec x)\cdot (\vec{y}-\vec{x})
        p^-(y-x)\gamma^0,\\
        t_2(x,y)&=\
        \remainr{r}{r:Delta-expilambda}(\vec x,\vec y)
        p^-(y-x)\gamma^0,
    \end{align}
    where we use the abbreviations $x=\pi_\Sigma(\vec x)$,
    $y=\pi_\Sigma(\vec y)$ again, and $\Gamma$ is defined in
    \eqref{eq: repr igammad4x}. We have introduced
    two extra factors
    $\gamma^0$
    in \eqref{eq:Delta-Pminus-diff} in order to
    have a positive-definite weight $\gamma^0\Gamma$.

    We claim that the kernel $t_2(x,y)\gamma^0$ gives rise to a
    Hilbert-Schmidt-operator $T_2$.  Indeed, using the bound
    \eqref{eq:bound-Gamma} for $\Gamma$, the bound \eqref{eq:bound-pminus-R3}
    from Corollary~\ref{lem:upper-bounds} in the appendix for $p^-$,  and the
    bound \eqref{eq:bound-error-R} for $ \remainr{r}{r:Delta-expilambda}$, 
    we have
    \begin{align}
        \|T_2\|^2_{I_2(\HSigma)}&=\int_{\vec x\in\R^3}\int_{\vec y\in\R^3}
        \operatorname{trace}[t_2(x,y)^*
        \gamma^0\Gamma(\vec x)t_2(x,y)\gamma^0\Gamma(\vec y)]
        \,d^3\vec y\,d^3\vec x  
        \nonumber\\
        &\le
        \constl{c:error-lambda-2}
         \int_{\vec x\in\R^3}\int_{\vec y\in\R^3}
        \left|
            \frac{e^{-\cD |\vec y-\vec x|^2}}{|\vec y-\vec x|}
        \right|^2
        \left( 1_{K}(x)+1_{K}(y) \right)\,d^3\vec y\,d^3\vec x  
        \le \constl{c:t2-I2}
        <\infty
    \end{align}
    for some constants $\constr{c:error-lambda-2}$ and
    $\constr{c:t2-I2}$ that depend on $\Sigma,\lambda,\tilde{\lambda}$.\\
    
    If $A|_{T\Sigma}=\widetilde A|_{T\Sigma}$ 
    then $\Delta\vec A=0$. This implies $t_1=0$ and
    therefore $P_\Sigma^\lambda-P_\Sigma^{\widetilde{\lambda}}=T_2$
    is a Hilbert-Schmidt operator. This proves the ``$\Leftarrow$'' part
    of the claim~\eqref{eq:equivalence-pol-classes}.\\

    Conversely, let us assume that  
    $A|_{T\Sigma}=\widetilde A|_{T\Sigma}$  does not hold.
    Then we can take some $x_0\in\R^3$ with $\Delta\vec A(\vec x_0)\neq 0$.
    By continuity of $\Delta\vec{A}$, we have $\inf_{\vec x\in U}|\Delta\vec A(x)|>0$
    for some neighborhood $U$ of $\vec x$. Furthermore there is a constant
    $\constl{c:Gamma-lower}(\Vmax)$ such that 
    $\gamma^0\Gamma(\vec x)-\constr{c:Gamma-lower}$ is positive-semidefinite
    for all $x=(x^0,\vec x)\in\Sigma$.
    Consequently, we get the following bound for all $\vec x\in U$ 
    and $\vec y\in\R^3$:
    \begin{align}
        &
        \operatorname{trace}\Big[t_1(\vec x,\vec y)^*
        \gamma^0\Gamma(\vec x)t_1(\vec x,\vec y)\gamma^0\Gamma(\vec y)\Big]
        \ge \constr{c:Gamma-lower}^2
        \operatorname{trace}\Big[t_1(\vec x,\vec y)^*t_1(\vec x,\vec y)\Big]
        \nonumber\\&
        \label{eq:lower-bound-trace-t1}
        \ge
        \constl{c:trace-t1}
        |\Delta \vec A(\vec x)\cdot (\vec{y}-\vec{x})|^2
        \|p^-(y-x)\|^2
        \ge
        \constl{c:trace-t1-2}
        |\Delta \vec A(\vec x)\cdot (\vec{y}-\vec{x})|^2
        \left(
            \frac{e^{-m|\vec y-\vec x|}}{|\vec y-\vec x|^3}
        \right)^2.
    \end{align}
    with two positive constants $\constr{c:trace-t1}$ and 
    $\constr{c:trace-t1-2}$
    depending on $\Vmax$.
    In the last step,
    we have used the lower bound \eqref{eq:lower-bound-pminus-R3}
    for $\|p^-\|$ from Corollary \ref{lem:lower-bounds}
    in the appendix.
    Because the lower bound given in 
    \eqref{eq:lower-bound-trace-t1} is not integrable
    over $(\vec x,\vec y)\in U\times\R^4$, we conclude that $T_1$ is
    not a Hilbert-Schmidt operator.
    Because $T_2$ is a Hilbert-Schmidt operator,
    this implies that $P_\Sigma^\lambda-P_\Sigma^{\widetilde{\lambda}}$ cannot
    be a Hilbert-Schmidt operator. Thus, we have proven part ``$\Rightarrow$''
    of the Theorem.
\end{proof}

\subsection{Proof of Theorem~\ref{thm:inf}}
\label{sec:evolution-polarization-classes}
This section contains the centerpiece of this work.
The proof of Theorem~\ref{thm:inf} will be given
at the end of this section. To show that the claimed equality
\eqref{eq:plambda-inf} holds, we analyze the difference of matrix elements
\begin{align}
        \sk{\phi,(P_{\Sigma_{t_1}}^A+S_{\Sigma_{t_1}}^A)\psi} 
        -
        \sk{\phi,(P_{\Sigma_{t_0}}^A+S_{\Sigma_{t_0}}^A)\psi}
\end{align}
for $\psi,\phi\in\CA$. This is done in two
steps. First, using Stokes' theorem, we provide a formula for the
derivative w.r.t.\ the flow parameter of the family of Cauchy surfaces
$(\Sigma_t)_{t\in T}$ in Lemma~\ref{lem:time-derivative} and
Corollary~\ref{cor:time-derivative}. Second, we give the relevant
estimates on this derivative in Lemmas~\ref{le:Dtp}-\ref{le:Dt-P+S} which are
summarized in Corollary~\ref{cor:r7r8}, and conclude with the proof of
Theorem~\ref{thm:inf}.

For the first step, the following notations for the Dirac operators acting from
the left and from the right, respectively, are convenient:
\begin{align}
    D^A\psi(x)=D^A_x\psi(x)
    &:=(i\slashed\partial^x-\slashed A(x)-m)\psi(x),
    \label{eq:DA}
    \\
    \overline{\phi(y)}\,\overleftarrow{D^A}
    =\overline{\phi(y)}\,\overleftarrow{D^A_y}
    &:=\overline{\phi(y)}(-i\overleftarrow{\slashed\partial^y}-\slashed A(y)-m)
    =\overline{D^A_y\phi(y)},
    \label{eq:DA-left}
\end{align}
where
$f(y)\overleftarrow{\slashed\partial^y}
    =f(y)\overleftarrow{\slashed\partial}:=\partial_\mu f(y)\gamma^\mu$.

\begin{lemma}\label{lem:time-derivative}
    Let $k:\R^4\times\R^4\to\C^{4\times 4}$ be a smooth
    function. Let $\phi,\psi\in\CA$. 
    Then for any $t\in T$ we have
    \begin{align}
        &\frac{d}{dt}
        \int_{x\in\Sigma_t}\int_{y\in\Sigma_t}
        \overline{\phi(x)}\,i_\gamma(d^4x)\,
        k(x,y)
        i_\gamma(d^4y)\,\psi(y)
        \nonumber\\&
        =
        -i\int_{x\in\Sigma_t}\int_{y\in\Sigma_t}
        \overline{\phi(x)}\,i_\gamma(d^4x)\,
        \mathcal{D}_t^Ak(x,y)
        i_\gamma(d^4y)\,\psi(y)
        \label{eq:claim-Dk}
    \end{align}
    with
    \begin{align}
        \label{eq:def-Dk}
        \mathcal{D}_t^Ak(x,y)
        :=
        v_t(x)\slashed{n}_t(x)D^A_xk(x,y)-k(x,y)\overleftarrow{D^A_y}v_t(y)\slashed{n}_t(y).
    \end{align}
\end{lemma}

\begin{proof}
    Assume that $\phi',\psi':\R^4\to\C^4$ are smooth functions
    with
    $\supp\phi'\cap\supp\psi'\subseteq K+\causal$ for some compact
    set $K\subset\R^4$.
    
    We set
    \begin{align}
      \boldsymbol{\Sigma}_{t_0t_1}:=\{(x,t)\in\boldsymbol{\Sigma}|\;t_0\le t\le t_1\}
    \end{align}
    for any real numbers $t_0\le t_1$.
    By Stokes' theorem, we have:
    \begin{align}
        \left(\int_{\Sigma_{t_1}}-\int_{\Sigma_{t_0}}\right)
        \overline{\phi'(x)}\,i_\gamma(d^4x)\,\psi'(x)
        =
        \int_{\boldsymbol{\Sigma}_{t_0t_1}}
        d[\overline{\phi'(x)}\,i_\gamma(d^4x)\,\psi'(x)].
    \end{align}
    We calculate:
    \begin{align}
        &d[\overline{\phi'(x)}i_\gamma(d^4x)\psi'(x)]=
        \partial_\mu(\overline{\phi'(x)}\gamma^\mu\psi'(x))\,d^4x
        \nonumber
        \\&=
        (\partial_\mu\overline{\phi'(x)})\gamma^\mu\psi'(x)\,d^4x
        +\overline{\phi'(x)}\gamma^\mu\partial_\mu\psi'(x)\,d^4x
        \nonumber
        \\&=
        \overline{\slashed{\partial}\phi'(x)}\psi'(x)\,d^4x
        +\overline{\phi'(x)}\slashed{\partial}\psi'(x)\,d^4x
        \nonumber
        \\&=
        i\overline{D^A\phi'(x)}\psi'(x)\,d^4x
        -i\overline{\phi'(x)}D^A\psi'(x)\,d^4x
        ,
        \label{eq:stokes-argument}
    \end{align}
    see also the calculation from (17) to (20) in \cite{Deckert2014}.
    Integration yields
    \begin{align}
        &\left(\int_{\Sigma_{t_1}}-\int_{\Sigma_{t_0}}\right)
        \overline{\phi'(x)}\,i_\gamma(d^4x)\,\psi'(x)
        \nonumber\\&=
        i\int_{\boldsymbol{\Sigma}_{t_0t_1}}
        [\overline{D^A\phi'(x)}\psi'(x)
        -\overline{\phi'(x)}D^A\psi'(x)]\,d^4x
        \nonumber\\&=
        i\int_{t_0}^{t_1}\int_{\Sigma_t} 
        [\overline{D^A\phi'(x)}\psi'(x)
        -\overline{\phi'(x)}D^A\psi'(x)]\,i_{v_tn_t}(d^4x)\,
        dt.
    \end{align}
    Differentiating this with respect to the upper boundary $t_1$,
    we conclude
    \begin{align}
        &\frac{d}{dt}\int_{\Sigma_t}
        \overline{\phi'(x)}\,i_\gamma(d^4x)\,\psi'(x)
        \nonumber\\&=
        i\int_{\Sigma_t} 
        [\overline{D^A\phi'(x)}\psi'(x)
        -\overline{\phi'(x)}D^A\psi'(x)]\,i_{v_tn_t}(d^4x)
        \nonumber\\&=
        i\int_{\Sigma_t} 
        [\overline{\phi'(x)}\,\overleftarrow{D^A}
            v_t(x)\slashed{n}_t(x)i_\gamma(d^4x)\psi'(x)
            -\overline{\phi'(x)}i_\gamma(d^4x)v_t(x)\slashed{n}_t(x)
        D^A\psi'(x)],
    \end{align}
    using
    \eqref{eq: repr igammad4x}.
    In the special case $\phi'\in\CA$
    this boils down to
    \begin{align}
        \label{eq:dirac-right}
        &\frac{d}{dt}\int_{\Sigma_t}
        \overline{\phi'(x)}\,i_\gamma(d^4x)\,\psi'(x)
        =
        -i\int_{\Sigma_t} 
        \overline{\phi'(x)}
        \,i_\gamma(d^4x)\,v_t(x)\slashed{n}_t(x)D^A\psi'(x),
    \end{align}
    while in the special case $\psi'\in\CA$ it boils down to 
    \begin{align}
        &\frac{d}{dt}\int_{\Sigma_t}
        \overline{\phi'(x)}\,i_\gamma(d^4x)\,\psi'(x)
        =
        i\int_{\Sigma_t} 
        \overline{\phi'(x)}\,\overleftarrow{D^A}
        v_t(x)\slashed{n}_t(x)\,i_\gamma(d^4x)\,\psi'(x).
        \label{eq:dirac-left}
    \end{align}
    We consider the function 
    \begin{align}
        F:T\times T\to\C,\quad
        F(s,t):=
        \int_{x\in\Sigma_s}\int_{y\in\Sigma_t}
        \overline{\phi(x)}\,i_\gamma(d^4x)\,
        k(x,y)\,
        i_\gamma(d^4y)\,\psi(y).
    \end{align}
    We apply \eqref{eq:dirac-right} to $\phi'=\phi$ and
        $\psi'(x)=
        \int_{y\in\Sigma_t}
        k(x,y)
        i_\gamma(d^4y)\,\psi(y)$
    to get
    \begin{align}
        \frac{\partial}{\partial s} F(s,t)
        =
        -i\int_{x\in \Sigma_s} \int_{y\in \Sigma_t} 
        \overline{\phi(x)}
        \,i_\gamma(d^4x)\,v_s(x)\slashed{n}_s(x)D^A_xk(x,y)
        \,i_\gamma(d^4y)\,\psi(y).
    \end{align}
    Similarly, we apply 
    \eqref{eq:dirac-left} to 
        $\overline{\phi'(y)}=
        \int_{y\in\Sigma_t}
        \overline{\phi(x)}\,i_\gamma(d^4x)\,k(x,y)$
    and $\psi'=\psi$ to get
    \begin{align}
        \frac{\partial}{\partial t} F(s,t)
        =
        i\int_{x\in \Sigma_s} \int_{y\in \Sigma_t} 
        \overline{\phi(x)}
        \,i_\gamma(d^4x)\,k(x,y)\overleftarrow{D^A_y}v_t(y)\slashed{n}_t(y)
        \,i_\gamma(d^4y)\,\psi(y).
    \end{align}
    From the chain rule, claim \eqref{eq:claim-Dk} follows:
    \begin{align}
        \nonumber&
        \frac{d}{dt} F(t,t)
        \\&=
        -i\int_{x\in \Sigma_s} \int_{y\in \Sigma_t} 
        \overline{\phi(x)}
        \,i_\gamma(d^4x)
        [v_t(x)\slashed{n}_t(x)D^A_xk(x,y)-k(x,y)\overleftarrow{D^A_y}v_t(y)\slashed{n}_t(y)]
        \,i_\gamma(d^4y)\,\psi(y).
    \end{align}
\end{proof}
From formula \eqref{eq:claim-Dk} and
the chain rule, we immediately get the following corollary.
\begin{corollary}
    \label{cor:time-derivative}
    For any smooth function  $k:\R^4\times\R^4\times T\to\C^{4\times 4}$, 
    $(x,y,t)\mapsto k_t(x,y)$,  any $\phi,\psi\in\CA$,
    and any $t\in T$ we have
    \begin{align}
        &\frac{d}{dt}
        \int_{x\in\Sigma_t}\int_{y\in\Sigma_t}
        \overline{\phi(x)}\,i_\gamma(d^4x)\,
        k_t(x,y)
        i_\gamma(d^4y)\,\psi(y)
        \nonumber\\&
        =
        \int_{x\in\Sigma_t}\int_{y\in\Sigma_t}
        \overline{\phi(x)}\,i_\gamma(d^4x)\,
        \left[-i\mathcal{D}_t^Ak(x,y)+\frac{\partial k_t}{\partial t}(x,y)\right]
        i_\gamma(d^4y)\,\psi(y).
        \label{eq:claim-Dk-t}
    \end{align}
\end{corollary}

This completes step one, and next, we turn to the relevant estimates.
In the following calculations for fixed $t\in T$, we drop the index $t$
in $v=v_t$ and $n=n_t$. Also, the $t$--dependence of the remainder terms
$r_{\ldots}$ is suppressed in the notation below, as we have
uniformity in $t$ of the error bounds.
Recall from equation \eqref{eq:electric-field} that $E_\mu=F_{\mu\nu}n^\nu$
denotes the ``electric field'' 
of the electromagnetic field $F_{\mu\nu}=\partial_\mu A_\nu-\partial_\nu A_\mu$
with respect 
to the local Cauchy surface $\Sigma$.
\begin{lemma}
    \label{le:Dtp}
    For $u\in\operatorname{Past}$, $\epsilon>0$,  and $x,y\in\R^4$, let 
    \begin{align}
      \label{eq:def-pA}
       p^{A,\epsilon u}(x,y)&:=e^{-i\lambda^A(x,y)}p^-(y-x+i\epsilon u) 
    \end{align} with $\lambda^A$ defined in \eqref{eq:special-lambda}.
    Then for $t\in T$, $x,y\in\Sigma_t$, $z=(z^0,\vec z)=y-x$, 
    and $w=z+i\epsilon u$ we have
    \begin{align}
        &\mathcal{D}_t^Ap^{A,\epsilon u}(x,y)
        \cr
        =&\frac12 v(x)\slashed{n}(x)
        \gamma^\nu F_{\mu\nu}(x)z^\mu p^{A,\epsilon u}(x,y)
        +\frac12 p^{A,\epsilon u}(x,y)\gamma^\nu 
        F_{\mu\nu}(y)z^\mu v(y)\slashed{n}(y)
        +\remainl{r}{r:R3}(x,y,\epsilon u)
        \label{eq:claimA}
        \\=&
        -\frac{i}{2m}
        v(x) z^\mu E_{\mu}(x)\slashed \partial D(w)
        +\remainl{r}{r:R3A}(x,y,\epsilon u)+\remainl{r}{r:R3B}(x,y,\epsilon u)
        \label{eq:claimB}
    \end{align}
    with error terms 
    \begin{align}
        \remainr{r}{r:R3}&=
        O_{A,u,\boldsymbol{\Sigma}}
        \left(\frac{e^{-\cD |\vec z|}}{|\vec z|}\right)
        [1_K(x)\vee 1_K(y)],
        \label{eq:r1-estimate-in-lemma}
        \\
        \remainr{r}{r:R3A}&=
        O_{A,u,\boldsymbol{\Sigma}}
        \left(\frac{e^{-\cD |\vec z|}}{|\vec z|}\right)[1_K(x)\vee 1_K(y)],
        \label{eq:claimB-error-bound1}
        \\
        \remainr{r}{r:R3B}&=
        O_{A,u,\boldsymbol{\Sigma}}\left(
          \sqrt{\epsilon}
          \frac{e^{-\cD |\vec z|}}{|\vec z|^{5/2}}\right)
        [1_K(x)\vee 1_K(y)]
        \label{eq:claimB-error-bound2}
      \end{align}
    for any compact set $K$ containing the support of $A$.
      For any two different points $x\neq y$ in $\Sigma_t$, the limit 
      $\remainr{r}{r:R3A}(x,y,0):=
      \lim_{\epsilon\downarrow 0}\remainr{r}{r:R3A}(x,y,\epsilon u)$
      exists.
  \end{lemma}

\begin{proof}
    We calculate for $x,y\in\Sigma_t$, $u\in\past$, and $\epsilon>0$:
    \begin{align}
        &D^A_x[e^{-i\lambda^A(x,y)}p^-(y-x+i\epsilon u)]
        \nonumber\\&
        =[\slashed\partial^x \lambda^A(x,y)-\slashed A(x)] 
        e^{-i\lambda^A(x,y)}p^-(y-x+i\epsilon u)
        +e^{-i\lambda^A(x,y)}(i\slashed\partial^x-m)p^-(y-x+i\epsilon u)
        \nonumber\\&
        =[\slashed\partial^x \lambda^A(x,y)-\slashed A(x)] p^{A,\epsilon
        u}(x,y),
        \quad
        \text{because}
        \quad
        (i\slashed\partial^x-m)p^-(y-x+i\epsilon u)=0.
    \end{align}
    Using the definition \eqref{eq:special-lambda} of $\lambda^A$, we get
    \begin{align}
        &\slashed\partial^x \lambda^A(x,y)-\slashed A(x)
        =\frac12 \gamma^\nu[A_\nu(y)-A_\nu(x)+(x^\mu-y^\mu)\partial_\nu^xA_\mu(x)]
        \nonumber\\&=
        \frac12 [\gamma^\nu F_{\mu\nu}(x)(y^\mu-x^\mu)+\remainl{r}{r:R1}(x,y)]
      \label{eq:slashed-partial-lambda-A}
    \end{align}
    with the Taylor rest term 
    \begin{align}
        &\remainr{r}{r:R1}(x,y)=\gamma^\nu[A_\nu(y)-A_\nu(x)-
        (y^\mu-x^\mu)\partial_\mu^xA_\nu(x)]=O_A(|x-y|^2)
        [1_K(x)\vee 1_K(y)]
        \cr&
        =
        O_A(|\vec z|^2)
        [1_K(x)\vee 1_K(y)]
        \text{ with }
        \vec z=\vec y-\vec x;
      \label{eq:taylor-rest-term}
    \end{align}
    cf.~formula \eqref{eq:rz-z} in the appendix, which compares $|z|$ with 
    $|\vec z|$.
    Recall that $K$ denotes a compact set containing the support of $A$. 
    Similarly, we find
    \begin{align}
        &[e^{-i\lambda^A(x,y)}p^-(y-x+i\epsilon u)]\overleftarrow{D^A_y}
        \nonumber\\&
        =
        e^{-i\lambda^A(x,y)}p^-(y-x+i\epsilon u)
        [-\slashed\partial^y \lambda^A(x,y)-\slashed A(y)] 
        +p^-(y-x+i\epsilon u)(-i\overleftarrow{\slashed\partial^y}-m)e^{-i\lambda^A(x,y)}
        \nonumber\\&
        =p^{A,\epsilon u}(x,y)[-\slashed\partial^y \lambda^A(x,y)-\slashed A(y)] .
    \end{align}
    Using the symmetry $\lambda^A(x,y)=-\lambda^A(y,x)$ and interchanging
    $x$ and $y$,
    equation \eqref{eq:slashed-partial-lambda-A} can be rewritten in the form
    \begin{align}
        &-\slashed\partial^y \lambda^A(x,y)-\slashed A(y)
        =
        \frac12 
        \left[
          -\gamma^\nu F_{\mu\nu}(y)(y^\mu-x^\mu)+\remainr{r}{r:R1}(y,x)
        \right].
    \end{align}
    Combining this with the definition \eqref{eq:def-Dk} of $\mathcal{D}_t^A$, 
    we find for $x,y\in\Sigma_t$, $z=y-x$
    \begin{align}
        &\mathcal{D}_t^Ap^{A,\epsilon u}(x,y)
        \nonumber\\&=
        \frac12 v(x)\slashed{n}(x)
        [\gamma^\nu F_{\mu\nu}(x)z^\mu+\remainr{r}{r:R1}(x,y)]p^{A,\epsilon u}(x,y)
        \nonumber\\&\quad+\frac12
        p^{A,\epsilon u}(x,y)
        [\gamma^\nu F_{\mu\nu}(y)z^\mu-\remainr{r}{r:R1}(y,x)]v(y)\slashed{n}(y)
        \nonumber\\&
        =\frac12 v(x)\slashed{n}(x)
        \gamma^\nu F_{\mu\nu}(x)z^\mu p^{A,\epsilon u}(x,y)
        +\frac12 p^{A,\epsilon u}(x,y)\gamma^\nu F_{\mu\nu}(y)z^\mu v(y)\slashed{n}(y)
        +\remainr{r}{r:R3}(x,y,\epsilon u)
        \label{eq:Dtp-form1}
    \end{align}
    with the error term 
    \begin{align}
        \remainr{r}{r:R3}(x,y,\epsilon u)&
        = \frac12 v(x)\slashed{n}(x)\remainr{r}{r:R1}(x,y)p^{A,\epsilon u}(x,y)
        -\frac12 p^{A,\epsilon u}(x,y)
        \remainr{r}{r:R1}(y,x)v(y)\slashed{n}(y)
        \cr&=
        O_{A,u,\boldsymbol{\Sigma}}
        \left(\frac{e^{-\cD |\vec z|}}{|\vec z|}\right)
        [1_K(x)\vee 1_K(y)],
        \label{eq:r1-estimate}
    \end{align}
    for $t\in T$, $x,y\in\Sigma_t$, $\epsilon>0$, $u\in\operatorname{Past}$.
    Here we used the bound \eqref{eq:bound-pminus-R3}
    in Lemma \ref{lem:upper-bounds} in the appendix for $p^-$, 
    the quadratic bound~\eqref{eq:taylor-rest-term} for 
    $\remainr{r}{r:R1}(x,y)$,
    and the fact that $|vn|$, being continuous,
    is bounded on compact sets.
    This proves 
    the claim given in \eqref{eq:claimA} with the error bound
    \eqref{eq:r1-estimate-in-lemma}. 
    
    It remains to
    prove the claim given in \eqref{eq:claimB} with the bounds
    \eqref{eq:claimB-error-bound1} and \eqref{eq:claimB-error-bound2}.
    Recall the definitions of $p^{A,\epsilon u}$ and $p^-$ given in
    \eqref{eq:def-pA} and \eqref{eq:pminus}, respectively.
    We have
    \begin{align}
        p^{A,\epsilon u}(x,y)=-\frac{i}{2m}\slashed\partial D(w)
        +\remainl{r}{r:rpA}(x,y,\epsilon u)
    \end{align}
    with the error term
    \begin{align}
        &\remainr{r}{r:rpA}(x,y,\epsilon u)=
        \frac{1}{2}e^{-i\lambda^A(x,y)}D(w)+
        (e^{-i\lambda^A(x,y)}-1)p^-(z+i\epsilon u)
        =
        O_{A,u,\boldsymbol{\Sigma}}
        \left(\frac{e^{-\cD |\vec z|}}{|\vec z|^2}\right)
    \end{align}
    using the bounds \eqref{eq:bound-D}, \eqref{eq:bound-pminus-R3}
    from the appendix and
    the Taylor bound
    \begin{align}
      |e^{-i\lambda^A(x,y)}-1|=O_A(|z|)\le O_{A,\boldsymbol{\Sigma}}(|\vec z|),
    \end{align}
    which follows 
    from $\lambda^A\in\mathcal{G}(A)$, cf.~Definition \ref{def:PLambda}
    and, once more, from the estimate \eqref{eq:rz-z} in the appendix.
    Hence we get from \eqref{eq:Dtp-form1}
    \begin{align}
        & \mathcal{D}_t^Ap^{A,\epsilon u}(x,y)
        -
        \remainr{r}{r:R3}(x,y,\epsilon u)
        \nonumber
        \\
        &=
        \frac12 v(x)\slashed{n}(x)
        \gamma^\nu F_{\mu\nu}(x)z^\mu p^{A,\epsilon u}(x,y)
        +\frac12 p^{A,\epsilon u}(x,y)\gamma^\nu F_{\mu\nu}(y)z^\mu v(y)\slashed{n}(y)
        \nonumber\\&=
        -\frac{i}{4m}
        v(x)\slashed{n}(x)
        \gamma^\nu F_{\mu\nu}(x)z^\mu \slashed\partial D(w)
        -\frac{i}{4m}\slashed\partial D(w)
        \gamma^\nu F_{\mu\nu}(y)z^\mu v(y)\slashed{n}(y)
        +\remainl{r}{r:error-Dp}(x,y,\epsilon u)
        \label{eq:drop-lambda-and-m}
      \end{align}
      with the error term
    \begin{align}
      \remainr{r}{r:error-Dp}
      =
      \frac12 v(x)\slashed{n}(x)
        \gamma^\nu F_{\mu\nu}(x)z^\mu \remainr{r}{r:rpA}
        +\frac12 \remainr{r}{r:rpA}\gamma^\nu F_{\mu\nu}(y)z^\mu v(y)\slashed{n}(y)
        =
      O_{A,u,\boldsymbol{\Sigma}}
        \left(\frac{e^{-\cD |\vec z|}}{|\vec z|}\right)
        [1_K(x)\vee 1_K(y)]
        .
        \label{eq:drop-lambda-and-m-error}
    \end{align}
    We employ estimate 
    \eqref{eq:bound-uD} for $\partial D$ from the appendix 
    and the fact $\operatorname{supp}F_{\mu\nu}\subseteq K$ to find
     \begin{align}
        &v(x)\slashed{n}(x) F_{\mu\nu}(x)\gamma^\nu z^\mu 
        \slashed\partial D(w)
        =v(x)\slashed{n}(x) F_{\mu\nu}(x)\gamma^\nu w^\mu 
        \slashed\partial D(w)
        +\remainl{r}{r:upartialD}(x,y,\epsilon u)
        \\
        &\slashed\partial D(w)
        \gamma^\nu F_{\mu\nu}(y)z^\mu v(y)\slashed{n}(y)
        =\slashed\partial D(w)
        \gamma^\nu F_{\mu\nu}(y)w^\mu v(y)\slashed{n}(y)
        +\remainl{r}{r:upartialD2}(x,y,\epsilon u)
     \end{align}
      with the error terms
      \begin{align}
        \remainr{r}{r:upartialD}&=
        -v(x)\slashed{n}(x) F_{\mu\nu}(x)\gamma^\nu i\epsilon u^\mu 
        \slashed\partial D(w)
        =
        O_{A,u,\boldsymbol{\Sigma}}
        \left(\sqrt{\epsilon}
          \frac{e^{-\cD |\vec z|}}{|\vec z|^{5/2}}\right)
        1_K(x)
        ,
        \label{eq:wtoz-order}
        \\
        \remainr{r}{r:upartialD2}
        &=-\slashed\partial D(w)
        \gamma^\nu F_{\mu\nu}(y)i\epsilon u^\mu v(y)\slashed{n}(y)
        =
        O_{A,u,\boldsymbol{\Sigma}}
        \left(\sqrt{\epsilon}
          \frac{e^{-\cD |\vec z|}}{|\vec z|^{5/2}}\right)
        1_K(y)
        .
        \label{eq:wtoz-order2}
       \end{align}
    Substituting this in \eqref{eq:drop-lambda-and-m}, we conclude
    \begin{align}
       \mathcal{D}_t^Ap^{A,\epsilon u}(x,y)
      =
      &
      -\frac{i}{4m}
      v(x)\slashed{n}(x)
      \gamma^\nu F_{\mu\nu}(x)w^\mu \slashed\partial D(w)
      -\frac{i}{4m}\slashed\partial D(w)
      \gamma^\nu F_{\mu\nu}(y)w^\mu v(y)\slashed{n}(y)
      \cr&
      +(\remainr{r}{r:R3}+\remainr{r}{r:error-Dp}+
      \remainl{r}{r:error-Dp-summary})(x,y,\epsilon u) 
      \label{eq:eq:Dp-summary}
    \end{align}
    with the additional error term
\begin{align}
  \remainr{r}{r:error-Dp-summary}=
  -\frac{i}{4m}(\remainr{r}{r:upartialD}+
  \remainr{r}{r:upartialD2})=        
  O_{A,u,\boldsymbol{\Sigma}}
  \left(\sqrt{\epsilon}
    \frac{e^{-\cD |\vec z|}}{|\vec z|^{5/2}}\right)
  [1_K(x)\vee 1_K(y)].
  \label{eq:bound-error-Dp-summary}
\end{align}
    The following ``Lorentz symmetry relation'' will be used
    several times in the calculations below. 
    \begin{align}
        \label{eq:Lorentz-symm-D}
        w_\nu\partial_\mu D(w)=w_\mu\partial_\nu D(w) \quad\text{ for }\quad
        w\in\operatorname{domain}(r).
    \end{align}
    Equation~\eqref{eq:Lorentz-symm-D} can be seen as follows. 
    Using $D=f\circ r$ with $f(\xi)=-m^3(2\pi^2)^{-1}K_1(m\xi)/(m\xi)$
    from \eqref{eq:def-D} and 
        $\partial_\mu r(w)=-\frac{w_\mu}{r(w)}$,
    we obtain
        $w_\nu\partial_\mu D(w)=-\frac{w_\nu w_\mu}{r(w)}f'(r(w))
        =w_\mu\partial_\nu D(w)$.

    Using the anticommutator relation $\{\gamma^\mu,\gamma^\nu\}=2g^{\mu\nu}$
    for the Dirac-matrices three times
    and the Lorentz symmetry relation \eqref{eq:Lorentz-symm-D},
    we calculate
     \begin{align}
        v(x) & \slashed{n}(x) F_{\mu\nu}(x)\gamma^\nu w^\mu 
        \slashed\partial D(w)
        = 
        [\slashed{n}(x)\gamma^\nu \slashed w]  
        v(x) F_{\mu\nu}(x)\partial^\mu D(w)
        \nonumber
        \\
        =&
        [2n^\nu(x) \slashed w
            -2\gamma^\nu n_\sigma(x) w^\sigma
            +2w^\nu \slashed{n}(x)
        -\slashed w\gamma^\nu \slashed{n}(x)]
        v(x) F_{\mu\nu}(x)\partial^\mu D(w)
        \nonumber
        \\
        =&
        2n^\nu(x) \slashed w
        v(x) F_{\mu\nu}(x)\partial^\mu D(w)
        \label{eq:calc1-1}
        \\
        &  -2\gamma^\nu n_\sigma(x) w^\sigma
        v(x) F_{\mu\nu}(x)\partial^\mu D(w)
        \label{eq:calc1-2}
        \\
        &+2w^\nu \slashed{n}(x)
        v(x) F_{\mu\nu}(x)\partial^\mu D(w)
        \label{eq:calc1-3}
        \\
        &-\slashed w\gamma^\nu \slashed{n}(x)
        v(x) F_{\mu\nu}(x)\partial^\mu D(w).
        \label{eq:calc1-4}
    \end{align}
    For the first term \eqref{eq:calc1-1}, using the Lorentz symmetry
    \eqref{eq:Lorentz-symm-D} again,
    we get
    \begin{align}
        \eqref{eq:calc1-1} & = 
        2n^\nu(x) \slashed w
        v(x) F_{\mu\nu}(x)\partial^\mu D(w)
        =
        2v(x) w^\mu E_{\mu}(x)\slashed \partial D(w)
        \nonumber
        \\
        & =
        2v(x) z^\mu E_{\mu}(x)\slashed \partial D(w)
        +
        \remainl{r}{r:vuEpartialD}(x,y,\epsilon u)
        \label{eq:term1}
      \end{align}
      with the error term
     \begin{align}
       \remainr{r}{r:vuEpartialD}=
       2v(x) i\epsilon u^\mu E_{\mu}(x)\slashed \partial D(w)
       =
       O_{A,u,\boldsymbol{\Sigma}}
        \left(
        \sqrt{\epsilon}
        \frac{e^{-\cD |\vec z|}}{|\vec z|^{5/2}}\right)
        1_K(x),
        \label{eq:term1-error}
    \end{align}
    where in the last step we have used estimate \eqref{eq:bound-uD} 
    once more. 
    For the second term \eqref{eq:calc1-2},
    we use $n_\sigma(x) z^\sigma=O_{\boldsymbol{\Sigma}}(|\vec z|^2)$,
    which holds because of $x,y\in\Sigma_t$ and $n(x)\perp T_x\Sigma_t$,
    to get
    \begin{align}
        \label{eq:term2}
        \eqref{eq:calc1-2} = 
        -2\gamma^\nu n_\sigma(x) w^\sigma
        v(x) F_{\mu\nu}(x)\partial^\mu D(w)
        =\remainl{r}{r:term2-boundA}(x,y,\epsilon u)+
        \remainl{r}{r:term2-boundB}(x,y,\epsilon u)
      \end{align}
      with the error terms
      \begin{align}
        \remainr{r}{r:term2-boundA}&=
        -2\gamma^\nu n_\sigma(x) z^\sigma
        v(x) F_{\mu\nu}(x)\partial^\mu D(w)=
        O_{A,u,\boldsymbol{\Sigma}}
        \left(\frac{e^{-\cD |\vec z|}}{|\vec z|}\right)1_K(x),
        \label{eq:term2-boundA}
        \\
        \remainr{r}{r:term2-boundB} &=
        -2\gamma^\nu n_\sigma(x) i\epsilon u^\sigma
        v(x) F_{\mu\nu}(x)\partial^\mu D(w)
        =O_{A,u,\boldsymbol{\Sigma}}
        \left(\sqrt{\epsilon}
        \frac{e^{-\cD |\vec z|}}{|\vec z|^{5/2}}\right)1_K(x).
        \label{eq:term2-boundB}
    \end{align}
    We have used the
    estimates \eqref{eq:bound-partialD} and, once more, \eqref{eq:bound-uD}.
    The contribution of the third term \eqref{eq:calc1-3} is zero, i.e.
    \begin{align}
        \label{eq:term3}
        \eqref{eq:calc1-3}
        =
        2w^\nu \slashed{n}(x)
        v(x) F_{\mu\nu}(x)\partial^\mu D(w)=0,
    \end{align}
    because of symmetry
    $w^\nu \partial^\mu D(w)= w^\mu \partial^\nu D(w)$, 
    cf.~\eqref{eq:Lorentz-symm-D},
    and antisymmetry
    $F_{\mu\nu}=-F_{\nu\mu}$.
    To express the fourth term  \eqref{eq:calc1-4}, we 
    use the Lorentz symmetry relation \eqref{eq:Lorentz-symm-D} again
    and replace $x$ by $y$
    up to the following error term:
    \begin{align}
      \remainl{r}{r:Lipschitz}(x,y)=
        F_{\mu\nu}(x) v(x)\slashed{n}(x)-F_{\mu\nu}(y) v(y)\slashed{n}(y)=
        O_{A,\boldsymbol{\Sigma}}(|\vec z|)[1_K(x)\vee 1_K(y)].
        \label{eq:Lipschitz}
    \end{align}
    We obtain for the fourth term  \eqref{eq:calc1-4}:
    \begin{align}
        \label{eq:term4}
        \eqref{eq:calc1-4}
        &=
        -\slashed w\partial^\mu D(w)\gamma^\nu \slashed{n}(x)
        v(x) F_{\mu\nu}(x)
        =
        -w^\mu \slashed \partial D(w)\gamma^\nu  F_{\mu\nu}(x)
        v(x)\slashed{n}(x)
        \nonumber
        \\
        &=
        -\slashed \partial D(w)\gamma^\nu  F_{\mu\nu}(y)w^\mu 
        v(y)\slashed{n}(y)+ 
        \remainl{r}{r:Lipschitz2}(x,y,\epsilon u)
      \end{align}
    with the error term
    \begin{align}
      \remainr{r}{r:Lipschitz2}&=
      w^\mu \slashed \partial D(w)\gamma^\nu  \remainr{r}{r:Lipschitz} =
      O_{A,u,\boldsymbol{\Sigma}}
      \left(\frac{e^{-\cD |\vec z|}}{|\vec z|}\right)[1_K(x)\vee 1_K(y)].
      \label{eq:Lipschitz2}
    \end{align}
    We have used estimate
    \eqref{eq:bound-wpartialDw} from the appendix and the 
    bound~\eqref{eq:Lipschitz}.
    The expressions 
    \eqref{eq:term1}, \eqref{eq:term2}, \eqref{eq:term3}
    and \eqref{eq:term4} of the four terms
    \eqref{eq:calc1-1}-\eqref{eq:calc1-4}
    give 
    \begin{align}
        & v(x) \slashed{n}(x) F_{\mu\nu}(x)\gamma^\nu w^\mu 
        \slashed\partial D(w)
        = 
        \eqref{eq:calc1-1}+
        \eqref{eq:calc1-2}+
        \eqref{eq:calc1-3}+
        \eqref{eq:calc1-4}
        \\
        =&
        [2v(x) z^\mu E_{\mu}(x)\slashed \partial D(w)
        + 
        \remainr{r}{r:vuEpartialD}]
        +
        [\remainr{r}{r:term2-boundA}+
        \remainr{r}{r:term2-boundB}]
        +0
        +[-\slashed \partial D(w)\gamma^\nu  F_{\mu\nu}(y)w^\mu 
        v(y)\slashed{n}(y)+ \remainr{r}{r:Lipschitz2}],
        \nonumber
     \end{align}
    which can be rewritten in the form
    \begin{align}
        &v(x) \slashed{n}(x) F_{\mu\nu}(x)\gamma^\nu w^\mu 
        \slashed\partial D(w)
        +\slashed \partial D(w)\gamma^\nu  F_{\mu\nu}(y)w^\mu 
        v(y)\slashed{n}(y)
        \nonumber
        \\
        &= 2v(x) z^\mu E_{\mu}(x)\slashed \partial D(w)
        + \remainl{r}{r:zinv-term}(x,y,\epsilon u)
        + \remainl{r}{r:eps-term}(x,y,\epsilon u)
    \end{align}
    with the error terms
    \begin{align}
      \remainr{r}{r:zinv-term}&=\remainr{r}{r:term2-boundA}
      +\remainr{r}{r:Lipschitz2}
      =O_{A,u,\boldsymbol{\Sigma}}
      \left(\frac{e^{-\cD |\vec z|}}{|\vec z|}\right)[1_K(x)\vee 1_K(y)],
      \\
      \remainr{r}{r:eps-term}&= 
      \remainr{r}{r:vuEpartialD}
      +\remainr{r}{r:term2-boundB}
      =
      O_{A,u,\boldsymbol{\Sigma}}
      \left(
        \sqrt{\epsilon}
        \frac{e^{-\cD |\vec z|}}{|\vec z|^{5/2}}\right)
      [1_K(x)\vee 1_K(y)].
    \end{align}
    We have used the estimates
    \eqref{eq:term2-boundA} and \eqref{eq:Lipschitz2}
    to bound $\remainr{r}{r:zinv-term}$ and the estimates
    \eqref{eq:term1-error} and  \eqref{eq:term2-boundB}
    to bound $\remainr{r}{r:eps-term}$.
    Substituting this result in 
    equation~\eqref{eq:eq:Dp-summary}
    together with the error bounds~\eqref{eq:r1-estimate}, 
    \eqref{eq:drop-lambda-and-m-error} and \eqref{eq:bound-error-Dp-summary},
    we infer
    \begin{align}
      & \mathcal{D}_t^Ap^{A,\epsilon u}(x,y)
      \nonumber
      \\
      =&  -\frac{i}{4m}
      v(x) \slashed{n}(x) F_{\mu\nu}(x)\gamma^\nu w^\mu\slashed\partial D(w)
      -\frac{i}{4m}\slashed\partial D(w)
      \gamma^\nu F_{\mu\nu}(y)w^\mu v(y)\slashed{n}(y)
      +\remainr{r}{r:R3}
      +\remainr{r}{r:error-Dp}
      +\remainr{r}{r:error-Dp-summary}
      \nonumber
      \\
      =& 
      -\frac{i}{2m}v(x) z^\mu E_{\mu}(x)\slashed \partial D(w) 
      +\remainr{r}{r:R3A}
      +\remainr{r}{r:R3B}
    \end{align}
    with the error terms
    \begin{align}
      \remainr{r}{r:R3A}(x,y,\epsilon u)&=\remainr{r}{r:R3}
      +\remainr{r}{r:error-Dp}
      -\frac{i}{4m}\remainr{r}{r:zinv-term}
      =O_{A,u,\boldsymbol{\Sigma}}
      \left(\frac{e^{-\cD |\vec z|}}{|\vec z|}\right)[1_K(x)\vee 1_K(y)],
      \label{eq:claimB-error-bound1-v2}
      \\
      \remainr{r}{r:R3B}(x,y,\epsilon u)&=
      \remainr{r}{r:error-Dp-summary}
      - \frac{i}{4m}\remainr{r}{r:eps-term}
      =O_{A,u,\boldsymbol{\Sigma}}\left(
        \sqrt{\epsilon}
        \frac{e^{-\cD |\vec z|}}{|\vec z|^{5/2}}\right)
        [1_K(x)\vee 1_K(y)].
        \label{eq:claimB-error-bound2-v2}
      \end{align}
    This proves the claim given in \eqref{eq:claimB}
    with the bounds \eqref{eq:claimB-error-bound1},
    \eqref{eq:claimB-error-bound2}.
    Recall that despite the uniformity in $\epsilon$ of the bound given in
    \eqref{eq:claimB-error-bound1-v2},
    $\remainr{r}{r:R3A}=\remainr{r}{r:R3A}(x,y,\epsilon u)$ depends on
    $\epsilon$. To ensure existence of the limit 
    $\lim_{\epsilon\downarrow 0}\remainr{r}{r:R3A}(x,y,\epsilon u)$
    for two different points $x,y\in\Sigma_t$ from the explicit form
    of $\remainr{r}{r:R3A}$,
    we observe that $z=y-x$ is space-like, and hence 
    $z\in\operatorname{domain}{r}$. As a consequence, the functions $D$ and
    $\partial_\mu D$ are continuous at $z$, cf.~Lemma~\ref{lemma:kern-p-minus},
    which implies the claim.
\end{proof}

In the following, we abbreviate $\partial_\mu=\partial/\partial w^\mu$.
Recall the notation $r(w)=\sqrt{-w_\mu w^\mu}$ from 
\eqref{eq:def-r}.

\begin{lemma}
    \label{le:dr2dD}
    For $w\in \operatorname{domain}(r)$ and $\mu=0,1,2,3$, one has
    \begin{align}
        \label{eq:dr2dD}
        \partial_\mu[r(w)^2\slashed\partial D(w)]
        =
        2w_\mu\slashed{\partial} D(w)
        -\gamma_\mu w^\nu\partial_\nu D(w)
        +\slashed{w}w_\mu m^2 D(w).
    \end{align}
\end{lemma}

\begin{proof}
    The function $D$ fulfills the Klein-Gordon equation
    \begin{align}
        \label{eq:Klein-Gordon-D}
        (\Box+m^2)D(w)=0,\quad w\in\operatorname{domain}(r).
    \end{align}
    Indeed, for $w\in \R^4+i\operatorname{Past}$, this can be seen 
    from the definition \eqref{eq:def-D} of $D$ as follows:
    Because of the fast convergence of $e^{ipw}$ to $0$ as $|p|\to\infty$,
    $p\in\cM_-$, we can interchange the Klein-Gordon-operator with the integral
    in the following calculation:
    \begin{align}
        (\Box +m^2)D(w)
        &=(2\pi)^{-3}m^{-1}\int_{\cM_-}(\Box^w+m^2)e^{ipw}\,i_p(d^4p)
        \nonumber\\&=
        (2\pi)^{-3}m^{-1}\int_{\cM_-}(-p^2+m^2)e^{ipw}\,i_p(d^4p)=0.
    \end{align}
    By analytic continuation, the Klein-Gordon equation
    \eqref{eq:Klein-Gordon-D} follows for all $w\in \operatorname{domain}(r)$.
    Equation~\eqref{eq:dr2dD} is proven by the following calculation:
    \begin{eqnarray}
        \partial_\mu[r(w)^2\slashed\partial D(w)]
        &
        = 
        &-\partial_\mu[w^\nu w_\nu\slashed\partial D(w)]
        \nonumber
        \\
        &
        \overset{\eqref{eq:Lorentz-symm-D}}{=}
        &-\partial_\mu[w^\nu\slashed w\partial_\nu D(w)] 
        \nonumber
        \\
        &
        = 
        &-\slashed w\partial_\mu D(w)
        -w^\nu\gamma_\mu\partial_\nu D(w)
        -w^\nu\slashed{w}\partial_\mu\partial_\nu D(w)
        \nonumber
        \\
        &
        =
        &-\slashed w\partial_\mu D(w)
        -w^\nu\gamma_\mu\partial_\nu D(w)
        -\slashed{w}\partial_\nu(w^\nu\partial_\mu D(w))+
        \slashed{w}(\partial_\nu w^\nu) \partial_\mu D(w)
        \nonumber
        \\
        &
        =
        &3\slashed w\partial_\mu D(w)
        -w^\nu\gamma_\mu\partial_\nu D(w)
        -\slashed{w}\partial_\nu(w^\nu\partial_\mu D(w))
        \nonumber
        \\
        &
        \overset{\eqref{eq:Lorentz-symm-D}}{=}
        &3\slashed w\partial_\mu D(w)
        -w^\nu\gamma_\mu\partial_\nu D(w)
        -\slashed{w}\partial_\nu(w_\mu\partial^\nu D(w))
        \nonumber
        \\
        &
        =
        &2\slashed w\partial_\mu D(w)
        -w^\nu\gamma_\mu\partial_\nu D(w)
        -\slashed{w}w_\mu\Box D(w) 
        \nonumber
        \\
        &
        \overset{\eqref{eq:Lorentz-symm-D}, 
        \eqref{eq:Klein-Gordon-D}}{=}
        & 2w_\mu\slashed{\partial} D(w)
        -\gamma_\mu w^\nu\partial_\nu D(w)
        +\slashed{w}w_\mu m^2 D(w).
        \label{eq:rechnung}
    \end{eqnarray}
\end{proof}

Recall the definition of the helper object $s_\Sigma^{A,\epsilon u}(x,y)=
[(\slashed{n}\slashed{E})(x)][(r^2\slashed\partial D)(w)]/(8m)$ introduced in
Definition~\ref{def:S}. The properties of $s_\Sigma^{A,\epsilon u}(x,y)$
claimed in Lemma~\ref{lem:SA} follow analogously to the arguments
used in \eqref{eq:Pminus-lambda-4-0}--\eqref{eq:Pminus-lambda-6}, i.e.,
from the bound
\eqref{eq:bound-r2partialD} given in Corollary~\ref{lem:upper-bounds} in the
appendix, the compact support of $E$, boundedness of
$\partial(\slashed{n}_t\slashed{E}_t)/\partial t$, and the dominated convergence
theorem.
\begin{lemma}
    \label{le:Dt-P+S}
    For $t\in \R$, $x,y\in\Sigma_t$, $z=y-x$, $u\in\operatorname{Past}$, and 
    $\epsilon>0$ we have
    \begin{align}
        \label{eq:DtS}
        \mathcal{D}_t^A s_\Sigma^{A,\epsilon u}(x,y)
        &= \frac{i}{2m} v_t(x)z^\mu E_\mu(x)\slashed\partial D(w)
        +\remainl{r}{r:DS1}(x,y,\epsilon u)
        +\remainl{r}{r:DS2}(x,y,\epsilon u),
        \\
        \label{eq:DtPplusS}
        \mathcal{D}_t^A(p^{A,\epsilon u}_\Sigma+s_\Sigma^{A,\epsilon u})(x,y)
        &=
        \remainl{r}{r:DPplusS1}(x,y,\epsilon u)
        +\remainl{r}{r:DPplusS2}(x,y,\epsilon u)
    \end{align}
    with error terms that fulfill the bounds
    \begin{align}
        \remainr{r}{r:DS1}
        =O_{A,u,\boldsymbol{\Sigma}}\left(
        \frac{e^{-\constl{c:Dprime2}|\vec z|}}{|\vec z|}\right)
        1_K(x)
        ,
        \qquad 
        &\remainr{r}{r:DS2}
        =O_{A,u,\boldsymbol{\Sigma}}\left(\sqrt{\epsilon}
        \frac{e^{-\constr{c:Dprime2}|\vec z|}}{|\vec z|^{5/2}}\right)
        1_K(x)
        ,
        \label{eq:DtS-bounds}
        \\
        \remainr{r}{r:DPplusS1}
        =O_{A,u,\boldsymbol{\Sigma}}\left(
        \frac{e^{-\constr{c:Dprime2}|\vec z|}}{|\vec z|}\right)
        [1_K(x)\vee 1_K(y)],\qquad
        &\remainr{r}{r:DPplusS2}
        =O_{A,u,\boldsymbol{\Sigma}}\left(\sqrt{\epsilon}
        \frac{e^{-\constr{c:Dprime2}|\vec z|}}{|\vec z|^{5/2}}\right)
        [1_K(x)\vee 1_K(y)]
        \label{eq:DtPplusS-bounds}
    \end{align}
    with some positive constant 
    $\constr{c:Dprime2}(\boldsymbol{\Sigma})$.    
    Furthermore, for $x\neq y$ the following limit exists:
    \begin{align}
      \label{eq:pointwise-limit}
      \remainr{r}{r:DPplusS1}(x,y,0)
    :=
    \lim_{\epsilon\downarrow 0} \remainr{r}{r:DPplusS1}(x,y,\epsilon u)
    \end{align}
\end{lemma}

\begin{proof}
    In this proof, we abbreviate $w=y-x+i\epsilon u=z+iu\epsilon$. Moreover, we
    suppress the $w$ dependence of $r(w)$, $D(w)$, $\partial^w$ and again
    also the $t$-dependence of $v$, $n$, and of the remainder terms $r_{\ldots}$ 
    in the notation. 
    Using the definition of $\mathcal D_t^A$
    given in \eqref{eq:def-Dk} of Lemma~\ref{lem:time-derivative}, we get
    \begin{align}
        &
        8m\mathcal D_t^A s_{\Sigma_t}^{A,\epsilon u}(x,y)
        \nonumber
        \\
        &=
        v(x)\slashed n(x)D^A_x
        [\slashed n(x)\slashed E(x)
        r^2\slashed\partial D]
        -[\slashed n(x)\slashed E(x)r^2
        \slashed\partial D]
        \overleftarrow{D^A_y}\slashed n(y)v(y)
        \nonumber
        \\
        &=
        v(x)\slashed n(x)i\slashed\partial^x
        [\slashed n(x)\slashed E(x)
        r^2\slashed\partial D]
        -[\slashed n(x)\slashed E(x)r^2
        \slashed\partial D]
        \overleftarrow{\slashed\partial^y}(-i)\slashed n(y)v(y)
        +
        \remainl{r}{r:D1}(x,y,\epsilon u)
        \nonumber
        \\
        &=
        iv(x)\slashed n(x)\gamma^\mu
        \slashed n(x)\slashed E(x)
        \partial^x_\mu[r^2\slashed\partial D]
        +i\slashed n(x)\slashed E(x)
        \partial^y_\mu[r^2
        \slashed\partial D]\gamma^\mu
        \slashed n(y)v(y)
        +
        \remainl{r}{r:D2}(x,y,\epsilon u)
        \nonumber
        \\
        &=
        -i v(x)\slashed n(x)\gamma^\mu\slashed n(x)\slashed E(x)
        \partial_\mu[r^2\slashed\partial D]
        +i\slashed n(x)\slashed E(x)
        \partial_\mu[r^2\slashed\partial D]\gamma^\mu \slashed n(x)v(x)
        +\remainl{r}{r:s1}(x,y,\epsilon u),
    \end{align}
    where the remainder terms are defined and estimated as follows:
    \begin{enumerate}[(i)]
        \item Recalling the definitions~\eqref{eq:DA} and \eqref{eq:DA-left}
          of the Dirac operators $D_A$ and $\overleftarrow{D^A}$
            and the fact that $A$ is compactly supported, 
            the estimate \eqref{eq:bound-r2partialD} of
            Corollary~\ref{lem:upper-bounds} in the appendix ensures
            \begin{align}
                \remainr{r}{r:D1}
                & =
                v(x)\slashed n(x)(-m - \slashed A(x))
                [\slashed n(x)\slashed E(x)
                r^2\slashed\partial D]
                -
                [\slashed n(x)\slashed E(x)r^2
                \slashed\partial D]
                (-m-\slashed A(y))\slashed n(y)v(y)
                \nonumber
                \\
                &=
                O_{A,u,\boldsymbol{\Sigma}}\left(
                \frac{e^{-\cD |\vec z|}}{|\vec z|}
                \right)
                1_{K}(x)
            \end{align}
            for some compact set $K$ containing the support of $E$.
        \item Using once more that $E$ has compact support 
          and using the bound \eqref{eq:bound-r2partialD} again we have the
            analogous estimate
            \begin{align}
                \remainr{r}{r:D2}
                &=
                \remainr{r}{r:D1}
                +
                iv(x)\slashed n(x)\gamma^\mu
                \left(\partial^x_\mu[\slashed n(x)\slashed E(x)]\right)
                r^2\slashed\partial D
                +i\left(\partial^y_\mu[\slashed n(x)\slashed E(x)]\right)
                r^2\slashed\partial D\gamma^\mu
                \slashed n(y)v(y)
                \nonumber
                \\
                &=
                O_{A,u,\boldsymbol{\Sigma}}\left(
                \frac{e^{-\cD |\vec z|}}{|\vec z|}
                \right)
                1_{K}(x)
            \end{align}
        \item Using the signs coming from inner derivatives: $-\partial^x
            D(w)=\partial^y D(w)=\partial D(w)$ and the Taylor expansion
            \begin{align}
                \slashed n(y)v(y) = \slashed n(x)v(x)+
                \remainl{r}{r:Lipschitz-nv}(x,y) 
                \quad\text{ with }\quad
                \remainr{r}{r:Lipschitz-nv}=O_{\boldsymbol{\Sigma}}(|\vec z|)
            \end{align}
            for $x,y\in\Sigma_t$ with $x\in K$ we find
            with the help of bound \eqref{eq:bound-dr2D} in the appendix:
            \begin{align}
                \remainr{r}{r:s1}
                &=
                \remainr{r}{r:D2}
                +
                i\slashed n(x)\slashed E(x)
                \partial_\mu[r^2\slashed\partial D]\gamma^\mu 
                \remainr{r}{r:Lipschitz-nv}
                =
                O_{A,u,\boldsymbol{\Sigma}}\left(
                \frac{e^{-\cD |\vec z|}}{|\vec z|}
                \right)
                1_{K}(x).
                \label{eq:bound-r10}
            \end{align}

    \end{enumerate}
    In the following calculations, we drop the argument $x$;
    thus, $v$, $n$, and $E$ stand for $v(x)$, $n(x)$, and $E(x)$, respectively,
    but $r=r(w)$ and $D=D(w)$.
    Using Lemma \ref{le:dr2dD},
    we get
    \begin{align}
        &-i\left(8m\mathcal D_t^A s^{A,\epsilon u}_\Sigma(x,y)
        -\remainr{r}{r:s1}
        \right)
        =
        -v\slashed n\gamma^\mu\slashed n\slashed E
        \partial_\mu[r^2\slashed\partial D]
        +\slashed n\slashed E
        \partial_\mu[r^2\slashed\partial D]\gamma^\mu \slashed nv
        \nonumber
        \\\nonumber&
        =
        -v\slashed n\gamma^\mu\slashed n\slashed E
        [2w_\mu\slashed{\partial} D
        -\gamma_\mu w^\nu\partial_\nu D]
        +v \slashed n\slashed E
        [2w_\mu\slashed{\partial} D
        -\gamma_\mu w^\nu\partial_\nu D]
        \gamma^\mu \slashed n 
        +\remainl{r}{r:s2}(x,y,\epsilon u)
        \\&
        =
        T_1+T_2+T_3+T_4+\remainr{r}{r:s2}
        \label{eq:Dts}
    \end{align}
    with the four terms
    \begin{align}
        \begin{split}
            T_1&=-2v\slashed n\slashed w\slashed n\slashed E
            \slashed\partial D,\\
            T_3&=2v\slashed n\slashed E
            \slashed\partial D\slashed w\slashed n,
        \end{split}
        \qquad
        \begin{split}
            &T_2=v\slashed n\gamma^\mu\slashed n\slashed E
            \gamma_\mu w^\nu\partial_\nu D,\\
            &T_4=-v\slashed n\slashed E\gamma_\mu 
            \gamma^\mu \slashed n w^\nu\partial_\nu D,
        \end{split}
    \end{align}
    and the remainder term
    \begin{align}
        \remainr{r}{r:s2}
        =
        -v\slashed n\gamma^\mu\slashed n\slashed E
        \slashed{w}w_\mu m^2 D
        +
        v\slashed n\slashed E\slashed{w}w_\mu m^2 D
        \gamma^\mu \slashed n
        =
        O_{A,u,\boldsymbol{\Sigma}}\left(
        e^{-\cD |\vec z|}
        \right)
        1_{K}(x),
        \label{eq:bound-r11}
    \end{align}
    where the bound comes from \eqref{eq:bound-wwD} of
    Corollary~\ref{lem:upper-bounds} in the appendix 
    and from $\operatorname{supp} E\subseteq K$.
    We evaluate the four terms $T_j$ separately.
    Using the anticommutation rules $\{\gamma^\mu,\gamma^\nu\}=2g^{\mu\nu}$
    for the Dirac matrices and
    $\slashed n^2=1$, we get
    \begin{align}
        T_1&=-2v\slashed n[2w^\nu n_\nu-\slashed n\slashed w]\slashed E
        \slashed\partial D
        \nonumber\\&=-4v\slashed n w^\nu n_\nu\slashed E\slashed\partial D
        +2v[2 w^\mu E_\mu -\slashed E\slashed w]
        \slashed\partial D
        \nonumber\\&=
        -4v\slashed n w^\nu n_\nu\slashed E\slashed\partial D
        +4vw^\mu E_\mu\slashed\partial D 
        -2v\slashed E w^\mu\partial_\mu D,
        \label{eq:T1}
    \end{align}
    where in the last step we used the Lorentz symmetry \eqref{eq:Lorentz-symm-D}
    to compute
    \begin{align}
        \slashed w \slashed \partial D
        &=
        \gamma^\mu \gamma^\nu w_\mu \partial_\nu D
        =
        \frac{1}{2}\left(
        \gamma^\mu \gamma^\nu w_\mu \partial_\nu D
        +
        \gamma^\mu \gamma^\nu w_\nu \partial_\mu D
        \right)
        \nonumber
        \\
        &=
        \frac{1}{2}\left(
        \gamma^\mu \gamma^\nu 
        +
        \gamma^\nu \gamma^\mu 
        \right)
        w_\mu \partial_\nu D
        =
        w^\mu\partial_\mu D.
        \label{eq:anti-lorentz}
    \end{align}
    Using the anticommutation rules again, the fact 
    $\gamma^\mu\gamma_\mu=4$, the definition $E_\mu=F_{\mu\nu}n^\nu$
    given in \eqref{eq:electric-field}, and the antisymmetry
    $F_{\mu\nu}=-F_{\nu\mu}$, we get
    \begin{align}
        \gamma^\mu\slashed n\slashed E\gamma_\mu
        &=
        (2n^\mu-\slashed n\gamma^\mu)(2E_\mu-\gamma_\mu\slashed E)
        =
        4n^\mu E_\mu 
        -4\slashed n\slashed E
        +\slashed n \gamma^\mu \gamma_\mu \slashed E
        \nonumber
        \\
        &=4 n^\mu E_\mu=4 n^\mu F_{\mu\nu} n^\nu=0
    \end{align}
    and therefore $T_2=0$.
    Using the same argument that was used to derive \eqref{eq:anti-lorentz} we also find 
    $\slashed\partial D(w) \slashed w=w^\mu\partial_\mu D$, and hence,
    \begin{align}
        T_3=2v\slashed n\slashed E
        w^\mu\partial_\mu D\slashed n.
    \end{align}
    Finally, we have
    \begin{align}
        T_4&=-4v\slashed n\slashed E\slashed n w^\nu\partial_\nu D,
    \end{align}
    which yields
    \begin{align}
        T_3+T_4=-2v\slashed n\slashed E
        \slashed n w^\mu\partial_\mu D
        =2v\slashed n^2
        \slashed E w^\mu\partial_\mu D
        =2v\slashed E w^\mu\partial_\mu D
        .
    \end{align}
    We have used that $\slashed n$ and $\slashed E$ anticommute
    because of $n^\mu E_\mu=n^\mu F_{\mu\nu} n^\nu=0$.
    Together with the expression~\eqref{eq:T1} for $T_1$ and $T_2=0$, we
    conclude
    \begin{align}
      \label{eq:sum-of-T}
        T_1+T_2+T_3+T_4=
        4 vw^\mu E_\mu\slashed\partial D+\remainl{r}{r:s3}(x,y,\epsilon u). 
    \end{align}
    with the error terms
    \begin{align}
        \remainr{r}{r:s3}
        =
        -4v\slashed n w^\nu n_\nu\slashed E\slashed\partial D
        =
        \remainl{r}{r:iu-term}(x,y,\epsilon u)
        +
        \remainl{r}{r:z-term}(x,y,\epsilon u),
        \label{eq:two-error-bounds}
    \end{align}
    where using $w=z+i\epsilon u$
    \begin{align}
        \remainr{r}{r:iu-term}
        =
        -4v\slashed n i\epsilon u^\nu n_\nu\slashed E\slashed\partial D,
        \qquad
        \remainr{r}{r:z-term}
        =
        -4v\slashed n z^\nu n_\nu\slashed E\slashed\partial D.
    \end{align}
    Inequality \eqref{eq:bound-uD} from the appendix 
    and the fact $\operatorname{supp} E\subseteq K$ provide the bound
    \begin{align}
        &\remainr{r}{r:iu-term}
        =O_{A,u,\boldsymbol{\Sigma}}\left(\sqrt{\epsilon}
        \frac{e^{-\cD |\vec z|}}{|\vec z|^{5/2}}\right)
        1_K(x)
        .
        \label{eq:bound-r13}
    \end{align}
    For the next estimate, we observe
    $(\nabla t_{\Sigma_t}(\vec x)\cdot\vec z,\vec z)
    \in T_x\Sigma_t\perp n(x)$;
    recall the parametrization \eqref{eq:parametrize Sigma} 
    of $\Sigma_t$. We obtain the Taylor expansion
    \begin{align}
        z^\nu n_\nu=
        n_0(x)[t_{\Sigma_t}(\vec y)-t_{\Sigma_t}(\vec x)]-\vec n(x)\cdot\vec z
        =
        n_0(x)\nabla t_{\Sigma_t}(\vec x)\cdot\vec z -\vec n(x)\cdot\vec z
        +O_{\boldsymbol{\Sigma}}(|\vec z|^2)
        =O_{\boldsymbol{\Sigma}}(|\vec z|^2)
    \end{align}
    uniformly for $x$ in the compact set $K$.
    Using \eqref{eq:bound-partialD} from the appendix and 
    the support property of $E$ again, this implies
    \begin{align}
        \remainr{r}{r:z-term}=
        O_{A,u,\boldsymbol{\Sigma}}\left(
        \frac{e^{-\cD |\vec z|}}{|\vec z|}\right)
        1_K(x)
        .
        \label{eq:bound-r14}
    \end{align}
    Finally, we have from equation~\eqref{eq:sum-of-T}
    \begin{align}
        T_1+T_2+T_3+T_4-\remainr{r}{r:s3}=
        4 v w^\mu E_\mu\slashed\partial D
        =  4 v z^\mu E_\mu\slashed\partial D+\remainl{r}{r:iu-term2}(x,y,\epsilon u)
        \label{eq:bound-sum-T}
    \end{align}
    with the error term
    \begin{align}
        \remainr{r}{r:iu-term2}=4 v i\epsilon u^\mu E_\mu\slashed\partial D
        =
        O_{A,u,\boldsymbol{\Sigma}}\left(\sqrt{\epsilon}
        \frac{e^{-\cD |\vec z|}}{|\vec z|^{5/2}}\right)
        1_K(x),
        \label{eq:bound-r15}
    \end{align}
    where once again we have used the bound 
    \eqref{eq:bound-uD} from the appendix  
    and the fact $\operatorname{supp} E\subseteq K$.
    Let us summarize:
    We use the equations \eqref{eq:Dts}, 
    \eqref{eq:bound-sum-T}, and \eqref{eq:two-error-bounds}
    to get the claimed formula
    \begin{align}
     &\mathcal D_t^A s_{\Sigma_t}^{A,\epsilon u}(x,y)
        =\frac{i}{2m} v z^\mu E_\mu\slashed\partial D
        +\remainr{r}{r:DS1}
        +\remainr{r}{r:DS2}
       \tag{\ref{eq:DtS}}
     \end{align} 
     with the remainder terms 
     \begin{align}
       \remainr{r}{r:DS1}
       &:=\frac{\remainr{r}{r:s1}}{8m}
       +
       \frac{i}{8m}
       \left(
         \remainr{r}{r:s2}
         +\remainr{r}{r:z-term}
       \right)
       =
        O_{A,u,\boldsymbol{\Sigma}}
        \left(
          \frac{e^{-\cD |\vec z|}}{|\vec z|}
          +e^{-\cD |\vec z|}
        \right)
        1_K(x)
        =
        O_{A,u,\boldsymbol{\Sigma}}
        \left(
          \frac{e^{-\constr{c:Dprime2}|\vec z|}}{|\vec z|}
        \right)
        1_K(x)
       \\   
       \remainr{r}{r:DS2}
       &:=\frac{i}{8m}
       \left(
         \remainr{r}{r:iu-term}
         +\remainr{r}{r:iu-term2}
       \right)
        =O_{A,u,\boldsymbol{\Sigma}}
        \left(
          \sqrt{\epsilon}\frac{e^{-\cD |\vec z|}}{|\vec z|^{5/2}}
        \right)
        1_K(x)
        =
        O_{A,u,\boldsymbol{\Sigma}}
        \left(
          \sqrt{\epsilon}\frac{e^{-\constr{c:Dprime2}|\vec z|}}{|\vec z|^{5/2}}
        \right)
        1_K(x)
      \end{align}
     with any positive constant 
     $\constr{c:Dprime2}(\boldsymbol{\Sigma})<\cD (\boldsymbol{\Sigma})$.
     We have applied the error bounds
     \eqref{eq:bound-r10}, \eqref{eq:bound-r11}, and 
     \eqref{eq:bound-r14}
     for the first remainder term $\remainr{r}{r:DS1}$, and the bounds
     \eqref{eq:bound-r13} and  
     \eqref{eq:bound-r15} for the second remainder term
     $\remainr{r}{r:DS2}$.
     Finally, we have weakened the bounds slightly 
     to get a simpler notation.
     This shows the claimed error bounds 
     in \eqref{eq:DtS-bounds}.

    Combining this with Lemma \ref{le:Dtp}
    and setting 
    $\remainr{r}{r:DPplusS1}
    =\remainr{r}{r:R3A}+\remainr{r}{r:DS1}$,
    $\remainr{r}{r:DPplusS2}
    =\remainr{r}{r:R3B}+\remainr{r}{r:DS2}$,
    equation \eqref{eq:DtPplusS} together with the corresponding error bounds
    \eqref{eq:DtPplusS-bounds} 
    are immediate consequences.
    
    To ensure existence of the limit of 
    $\remainr{r}{r:DPplusS1}(x,y,\epsilon u)$
    as $\epsilon\downarrow 0$ for $x,y\in\Sigma_t$ with $x\neq y$,
    we use the existence of the limits
    $\lim_{\epsilon\downarrow 0}\remainr{r}{r:R3A}(x,y,\epsilon t)$ and 
    $\lim_{\epsilon\downarrow 0}\remainr{r}{r:DS1}(x,y,\epsilon u)$.
    The existence of the former limit was proven in Lemma~\ref{le:Dtp}, 
    and existence of the latter limit follows by the same argument, i.e., from
    the fact that
    the functions $D$ and
    $\partial_\mu D$ are continuous at $z$, and that 
    $\remainr{r}{r:DS1}$ is explicitly given in terms of $D$ and its derivative.
    This yields the claim.
\end{proof}

\begin{corollary}\label{cor:r7r8}
    The error terms $\remainr{r}{r:DPplusS1}(\cdot,\cdot,\epsilon u)$ and
    $\remainr{r}{r:DPplusS2}(\cdot,\cdot,\epsilon u)$
    in \eqref{eq:DtPplusS}
    give rise to bounded linear operators 
    $\remainr{R}{r:DPplusS1}^{\epsilon u}(t), \remainr{R}{r:DPplusS2}^{\epsilon u}(t):
    \cH_{\Sigma_t}\selfmaps$
    with matrix elements
    \begin{align}
        \sk{\phi,\remainr{R}{r:DPplusS1}^{\epsilon u}(t)\psi}
        =
        \int_{x\in\Sigma_t}
        \int_{y\in\Sigma_t}
        \overline{\phi(x)} \, i_\gamma(d^4x) \,
        \remainr{r}{r:DPplusS1}(x,y,\epsilon u)
        \, i_\gamma(d^4y) \, 
        \psi(y),
        \qquad
        \psi,\phi\in \cH_{\Sigma_t}
    \end{align}
    and similarly for $\remainr{r}{r:DPplusS2}(x,y,\epsilon u)$,
    $\remainr{R}{r:DPplusS2}^{\epsilon u}(t)$.
    They fulfill: 
    \begin{enumerate}[(i)]
    \item The operators $\remainr{R}{r:DPplusS1}^{\epsilon u}(t)$,
      $\epsilon\geq 0$, are
      Hilbert-Schmidt operators.
      There is a constant $\constl{c:DPplusS-1}(A,u,\mathbf{\Sigma})$ such
      that
      $\sup_{t\in T,\epsilon>0}
      \|\remainr{R}{r:DPplusS1}^{\epsilon u}(t)\|_{I_2(\cH_{\Sigma_t})}
      \leq
      \constr{c:DPplusS-1}$. Furthermore,
      \begin{align}            
        \lim_{\epsilon\downarrow 0} \|
        \remainr{R}{r:DPplusS1}^{\epsilon u}(t) -
        \remainr{R}{r:DPplusS1}^0(t) \|_{I_2(\cH_{\Sigma_t})}=0.
      \end{align}
    \item 
      $\sup_{t\in T}\|\remainr{R}{r:DPplusS2}^{\epsilon u}(t)\|_{\cH_{\Sigma_t}\selfmaps}
      \leq
      O_{A,u,\boldsymbol{\Sigma}}(\sqrt{\epsilon})$.
    \end{enumerate}
\end{corollary}
\begin{proof}
    (i) 
    For $\psi,\phi\in \cH_{\Sigma_t}$, using
    the bound \eqref{eq:DtPplusS-bounds} for $\remainr{r}{r:DPplusS1}$,
    we find uniformly for $\epsilon>0$ and
    $t\in T$ that
    \begin{align}
        \| \remainr{R}{r:DPplusS1}^{\epsilon u}(t) \|^2_{I_2(\cH_{\Sigma_t})}
        & =
        \int_{\vec x\in\R^3}
        \int_{\vec y\in\R^3}
        \operatorname{trace}
        \left[
          \gamma^0
          \remainr{r}{r:DPplusS1}(x,y,\epsilon u)^*
          \gamma^0 \Gamma(\vec x)
          \remainr{r}{r:DPplusS1}(x,y,\epsilon u)
          \Gamma(\vec y)
        \right]
        \, d^3\vec y d^3\vec x
        \\
        & \leq
        \constl{c:DPplusS}
        \int_{\vec x\in\R^3}
        \int_{\vec y\in\R^3}
        \left[
            \frac{e^{-\cD |\vec y-\vec x|}}{|\vec y-\vec x|}
        \right]^2
        (1_K(x) \vee 1_K(y))
        \, d^3\vec y d^3\vec x < \infty.
    \end{align}
    for some constant $\constr{c:DPplusS}(A,u,\boldsymbol{\Sigma})$.
    The limit $\remainr{R}{r:DPplusS1}^{\epsilon u}(t)\konv{\epsilon\downarrow
    0} \remainr{R}{r:DPplusS1}^0(t)$ in the $I_2(\cH_{\Sigma_t})$ norm is implied
    by the point-wise convergence \eqref{eq:pointwise-limit} 
    stated in Lemma~\ref{le:Dt-P+S} and the point-wise
    bound \eqref{eq:DtPplusS-bounds}, using dominated convergence.
    \newline

   (ii)
    For $\psi,\phi\in \cH_{\Sigma_t}$, using the bound
    in \eqref{eq:DtPplusS-bounds} for $\remainr{r}{r:DPplusS2}$ and
    the
    Cauchy-Schwarz inequality, we find analogously to the calculation
    \eqref{eq:Pminus-lambda-4-0}--\eqref{eq:Pminus-lambda-6}:
    \begin{align}
        |\sk{\phi,\remainr{R}{r:DPplusS2}^{\epsilon u}(t)\psi}|
        & \leq O_{A,u,\boldsymbol{\Sigma}}(\sqrt{\epsilon})
        \int_{\vec x\in\R^3}
        \int_{\vec y\in\R^3}
        |\phi(x)| \, |\Gamma(\vec x)| \, d^3\vec x \,
        \frac{e^{-\cD |\vec y-\vec x|}}{|\vec y-\vec x|^{5/2}}
        \, |\Gamma(\vec y)| \, d^3\vec y \, 
        |\psi(y)|
        \\
        & \leq 
        O_{A,u,\boldsymbol{\Sigma}}(\sqrt{\epsilon})
        \int_{\vec z\in\R^3}
        \frac{e^{-\cD |\vec z|}}{|\vec z|^{5/2}} \, d^3\vec z \,
        \|\phi\| \, \|\psi\|,
    \end{align}
    which is finite and uniform in $t$.

    The existence of the bounded linear operators 
    $\remainr{R}{r:DPplusS1}^{\epsilon u}(t), \remainr{R}{r:DPplusS2}^{\epsilon u}(t):
    \cH_{\Sigma_t}\selfmaps$ follows.
\end{proof}

Finally, we prove the Theorem~\ref{thm:inf} with the collected ingredients.
\begin{proof}[Proof of Theorem~\ref{thm:inf}]
    With justifications given below, we find that for $\phi,\psi\in\CA$
    \begin{align}
        &
        \sk{\phi|_{\Sigma_{t_1}},(P_{\Sigma_{t_1}}^A+S_{\Sigma_{t_1}}^A)\psi|_{\Sigma_{t_1}}} 
        -
        \sk{\phi|_{\Sigma_{t_0}},(P_{\Sigma_{t_0}}^A+S_{\Sigma_{t_0}}^A)\psi|_{\Sigma_{t_0}}}
        \label{eq:claim-Dps-ini}
        \\=&
        \lim_{\epsilon\downarrow 0}\left(\int_{x\in\Sigma_{t_1}}\int_{y\in\Sigma_{t_1}}-
        \int_{x\in\Sigma_{t_0}}\int_{y\in\Sigma_{t_0}}\right)
        \overline{\phi(x)}\,i_\gamma(d^4x)\,
        (p^{A,\epsilon u}+s_{\Sigma_t}^{A,\epsilon u})(x,y)
        i_\gamma(d^4y)\,\psi(y)
        \label{eq:claim-Dps-minus1}
        \\=&
        \lim_{\epsilon\downarrow 0}\int_{t_0}^{t_1}\int_{x\in\Sigma_t}\int_{y\in\Sigma_t}
        \overline{\phi(x)}\,i_\gamma(d^4x)\,
        \left[-i\mathcal{D}_t^A(p^{A,\epsilon u}+s_{\Sigma_t}^{A,\epsilon u})+
        \frac{\partial s_{\Sigma_t}^{A,\epsilon u}}{\partial t}\right](x,y)
        i_\gamma(d^4y)\,\psi(y)\,dt
        \label{eq:claim-Dps-0}\\=&
        \lim_{\epsilon\downarrow 0}\int_{t_0}^{t_1}\int_{x\in\Sigma_t}\int_{y\in\Sigma_t}
        \overline{\phi(x)}\,i_\gamma(d^4x)\,
        \left[
            -i\remainr{r}{r:DPplusS1}(x,y,\epsilon u)
            -i\remainr{r}{r:DPplusS2}(x,y,\epsilon u)+
        \frac{\partial s_{\Sigma_t}^{A,\epsilon u}}{\partial t}(x,y)
        \right]
        \cr&\cdot
        i_\gamma(d^4y)\,\psi(y)\,dt
        \label{eq:claim-Dps-1}\\
        =&\lim_{\epsilon\downarrow 0}
        \int_{t_0}^{t_1}
        \left[-i\sk{\phi|_{\Sigma_t},
            \remainr{R}{r:DPplusS1}^{\epsilon u}(t)
            \psi|_{\Sigma_t}}
          -i\sk{\phi|_{\Sigma_t},
            \remainr{R}{r:DPplusS2}^{\epsilon u}(t)
            \psi|_{\Sigma_t}} 
          +\sk{\phi|_{\Sigma_t},
            \dot S_{\Sigma_t}^{A,\epsilon u},\psi|_{\Sigma_t}}\right]\,dt
        \label{eq:claim-Dps-2}
    \end{align}
    In the first step from \eqref{eq:claim-Dps-ini} to
    \eqref{eq:claim-Dps-minus1} we expressed the matrix elements of the
    operators $P_\Sigma^{\lambda_A}$ and $S_\Sigma^A$ in terms of the respective
    integral kernels $p^{\epsilon u,\lambda_A}$ and $s^{\epsilon u,A}_\Sigma$
    given in
    Lemma~\ref{lem:Pminus-lambda} and part (i) of 
    Lemma~\ref{lem:SA}.  The step from
    \eqref{eq:claim-Dps-minus1} to \eqref{eq:claim-Dps-0} follows from
    Corollary~\ref{cor:time-derivative}.  The step from \eqref{eq:claim-Dps-0}
    to  \eqref{eq:claim-Dps-1} is a consequence of equation~\eqref{eq:DtPplusS}
    in Lemma~\ref{le:Dt-P+S}.  Finally, in the step from
    \eqref{eq:claim-Dps-1} to \eqref{eq:claim-Dps-2} we have used 
    that the integral kernels
    $\remainr{r}{r:DPplusS1}(\cdot,\cdot,\epsilon u)$,
    $\remainr{r}{r:DPplusS2}(\cdot,\cdot,\epsilon u)$,
    and $\partial s_{\Sigma_t}^{A,\epsilon u}/\partial t$
    give rise to bounded operators
    $\remainr{R}{r:DPplusS1}^{\epsilon u}(t)$, 
    $\remainr{R}{r:DPplusS2}^{\epsilon u}(t)$, 
    and $\dot S_{\Sigma_t}^{A,\epsilon u}$
    as ensured by Corollary~\ref{cor:r7r8} and part (ii) of 
    Lemma~\ref{lem:SA}.
    
    Claim (ii) of Corollary~\ref{cor:r7r8} implies that 
    $\remainr{R}{r:DPplusS2}^{\epsilon u}(t)$ converges to zero in operator norm
    as $\epsilon\downarrow 0$, uniformly
    in $t\in T$. Furthermore, claim (i) of Corollary~\ref{cor:r7r8}
    and part (ii) of 
    Lemma~\ref{lem:SA}
    guarantee that 
    $-i\remainr{R}{r:DPplusS1}^{\epsilon u}(t)+\dot S_{\Sigma_t}^{A,\epsilon u}$ 
    converges in the $I_2(\cH_{\Sigma_t})$ norm to a Hilbert-Schmidt operator 
    $R(t):=-i\remainr{R}{r:DPplusS1}^0(t)+\dot S_{\Sigma_t}^{A,0}$ 
    such that
    $\sup_{t\in T}\|R(t)\|_{I_2(\cH_{\Sigma_t})}<\infty$.
      Calculation \eqref{eq:claim-Dps-ini}--\eqref{eq:claim-Dps-2}
      can now be rewritten in the form of claim \eqref{eq:plambda-inf}:
      \begin{align}
        \sk{\phi|_{\Sigma_{t_1}},(P_{\Sigma_{t_1}}^A+S_{\Sigma_{t_1}}^A)\psi|_{\Sigma_{t_1}}} 
        -
        \sk{\phi|_{\Sigma_{t_0}},(P_{\Sigma_{t_0}}^A+S_{\Sigma_{t_0}}^A)\psi|_{\Sigma_{t_0}}}
        =\int_{t_0}^{t_1}\sk{\phi|_{\Sigma_t},R(t)\psi|_{\Sigma_t}}\, dt
      \end{align}
      at first for $\phi,\psi\in\CA$, but then extended by a density argument
      to $\phi,\psi\in\HA$.
    Since the operators $U_{A\Sigma}$ are unitary, we get the estimate
    \begin{align}
        & \left\|
        U_{A\Sigma_{t_1}} 
        ( 
        P_{\Sigma_{t_1}}^A+S_{\Sigma_{t_1}}^A
        )
        U_{\Sigma_{t_1}A} 
        -
        U_{A\Sigma_{t_0}} 
        ( 
        P_{\Sigma_{t_0}}^A+S_{\Sigma_{t_0}}^A
        )
        U_{\Sigma_{t_0}A} 
        \right\|_{I_2(\HA)}
        \\
        & \qquad \leq
        \int_{t_0}^{t_1}
        \|
        R(t)
        \|_{I_2(\cH_{\Sigma_t})}
        \,dt
        < \infty.
    \end{align}
    This proves the claim.
\end{proof}

\begin{proof}[Proof of Theorem~\ref{thm:evo-polclass}]
    As a consequence of Theorem~\ref{thm:inf} and Lemma~\ref{lem:SA} claim
    \eqref{eq:gen-key-prop} holds for the special case $\lambda=\lambda^A$. For
    general $\lambda\in\cG(A)$, Theorem~\ref{thm:equivalence-pol-classes}
    implies $P_\Sigma^A-P_\Sigma^\lambda\in I_2(\HSigma)$ which concludes the
    proof for the general case.
\end{proof}

\begin{appendix}
\section{Appendix}

In this appendix we provide auxiliary estimates for the covariant functions $D$,
its derivatives, and $p^-$ needed in the proofs of the main results.

\begin{lemma}[Upper bounds]
    \label{lem:upper-bounds}
    Let $u$ be a time-like four-vector.
    For all space-like $z\in\R^4$ with 
    $|z^0|\leq \Vmax |\vec z|$
    and $\epsilon\geq 0$ with $w=z+i\epsilon u\neq 0$ 
    we have the following bounds with the constant
    $\cD (\Vmax)=\frac{m}{2}\sqrt{1-\Vmax^2}$,
    reading $1/0$ as $+\infty$:
    \begin{align}
        \left| w^\mu w^\nu D(w) \right| 
        & 
        \le O_{u,\Vmax}\left( e^{-\cD |\vec z|}\right),
        \label{eq:bound-wwD}
        \\
        \left| D(w) \right| 
        & 
        \le O_{\Vmax}\left( \frac{e^{-\cD |\vec z|}}
        {|\vec z|^2}\right)
        ,
        \label{eq:bound-D}
        \\
        \left|r(w)^2 \partial_\mu D (w)\right| 
        & 
        \le  O_{u,\Vmax}\left( \frac{e^{-\cD |\vec z|}}
        {|\vec z|}\right),
        \label{eq:bound-r2partialD}
        \\
        \left| \partial_\mu D (w)\right|
        & \le O_{u,\Vmax}\left( \frac{e^{-\cD |\vec z|}}
        {|\vec z|^3\vee\epsilon^3}\right),
        \label{eq:bound-partialD}
    \end{align}
    \begin{align}
        \left| w^\nu \partial_\mu D (w)\right|
        & \le O_{u,\Vmax}\left( \frac{e^{-\cD |\vec z|}}
        {|\vec z|^2\vee\epsilon^2}\right),
        \label{eq:bound-wpartialDw}
        \\
        \left| \epsilon u^\mu \partial_\nu D (w)\right|
        & \le
        O_{u,\Vmax}\left( \frac{\sqrt{\epsilon} e^{-\cD |\vec z|}}
        {|\vec z|^{5/2}}\right),
        \label{eq:bound-uD}
        \\
        \left|\partial_\nu \left[ r(w)^2 \partial_\mu D (w) \right]\right| 
        & \leq
        O_{u,\Vmax}\left(\frac{e^{-\cD |\vec z|}}{|\vec z|^2}
        \right),
        \label{eq:bound-dr2D}
        \\
        \left\| p^-(w) \right\| 
        & \le 
        O_{u,\Vmax}\left( \frac{e^{-\cD |\vec z|}}
        {|\vec z|^3}\right).
        \label{eq:bound-pminus-R3}
    \end{align}
    For $\epsilon=0$ one may take, e.g., $u=(-1,0,0,0)$. In this case the
    $u$-dependence of the constants in
    \eqref{eq:bound-wwD}-\eqref{eq:bound-pminus-R3} drops out.
\end{lemma}

\begin{lemma}[Lower Bound]
    \label{lem:lower-bounds}
    For all space-like $z\in\R^4\setminus\{0\}$ 
    one has the lower bound
    \begin{align}
        \label{eq:lower-bound-pminus-R3}
        \left\| p^-(z) \right\| 
        &
        \geq
        \constl{c:n-pm-3}
        \frac{e^{-m|\vec z|}}{|\vec z|^3}
    \end{align}
    with  a positive numerical constant $\constr{c:n-pm-3}$.
\end{lemma}

\ifx\arxiv\undefined
The proofs have been carried out in \cite{preprint}. However, they can also be
inferred from the asymptotic behavior of the modified Bessel function $K_1$
and its derivative given in \cite[Chapters 9.6 and 9.7]{abr65}.
\else
The proofs of these two lemmas are based on the control of the
modified Bessel function $K_1$ provided by the following lemma.

\begin{lemma}[Control of $K_1$]
    \label{lem:K1-control}
    The modified Bessel function $K_1:\R^++i\R\to\C$ defined
    in \eqref{eq:K1} has an analytic continuation 
    \begin{align}
        K_1:\C\setminus\R_0^-\to\C,
        \qquad
        K_1(\xi) 
        =
        \frac{e^{-\xi}}{\xi} \int_0^\infty e^{-t}
        \sqrt{t^2+2\xi t} \, dt
        .
        \label{eq:K1-1}
    \end{align}
    The following estimates hold for $\xi\in\C\setminus\R^-_0$: 
    \begin{align}
        \label{eq:bound-K1}
        |K_1(\xi)| 
        & \leq
        e^{-\re \xi}\left( \frac{1}{|\xi|} + \sqrt{\frac{\pi}{2|\xi|}} \right),
        \\
        \label{eq:bound-partial-K1}
        \left|\frac{\partial}{\partial \xi} \frac{K_1(\xi)}{\xi} \right|
        & \leq
        e^{-\re \xi}\left(\frac{2^{3/2}}{|\xi|^3} + 2 \sqrt{\pi}|\xi|^{-3/2}\right).
    \end{align}
    For $\xi>0$, we have the lower bounds
    \begin{align}
        K_1(\xi)\ge \frac{e^{-\xi}}{\xi},
        \qquad
        \label{eq:lower-bound-K1}
        -\frac{\partial}{\partial \xi} \frac{K_1(\xi)}{\xi}\ge 2\frac{e^{-\xi}}{\xi^3}.
    \end{align}
\end{lemma}
 
\begin{proof}[Proof of Lemma~\ref{lem:K1-control}]
    \ifx\arxiv\defined
        The proof is given in \cite{preprint}.
    \else
    In the special case $\xi>0$,
    the representation \eqref{eq:K1-1}
    of $K_1(\xi)$ is found by substituting
    $s=t/\xi+1$ in the integral representation given
    in \eqref{eq:K1}. 
    Again for $\xi>0$, using $\sqrt{t^2+2\xi t}\ge t$,  
    formula \eqref{eq:K1-1} implies
    the first claim in \eqref{eq:lower-bound-K1} as follows:
    \begin{align}
        K_1(\xi) 
        \ge
        \frac{e^{-\xi}}{\xi} \int_0^\infty e^{-t}
        t \, dt=\frac{e^{-\xi}}{\xi} 
        \label{eq:K1-2}
    \end{align}
    Returning to general $\xi\in\C\setminus\R_0^-$, we have for 
    $t^2+2\xi t\in\C\setminus\R^-_0 = \operatorname{domain}(\sqrt\cdot)$
    for $t>0$. As a
    consequence, the dominated convergence theorem allows us to interchange derivatives
    $\partial_\xi$ and $\partial_{\overline \xi}$ with the integral over $t$
    in the right-hand side of \eqref{eq:K1-1}. Hence, the right-hand side in
    \eqref{eq:K1-1}
    is holomorphic in $\xi\in\C\setminus R^-_0$.
    In particular, using the identity theorem  for holomorphic functions,
    formula \eqref{eq:K1-1} yields the analytic continuation
    of $K_1$ from its original domain $\R^++i\R$ to the extended 
    domain $\C\setminus\R_0^-$. We know
    \begin{align}
        \left| \sqrt{t^2 + 2\xi t} \right|
        \leq 
        \sqrt{t^2 + 2|\xi|t}
        \leq t + \sqrt{2|\xi|t}
    \end{align}
    and therefore
    \begin{align}
        \left| \int_0^\infty e^{-t} \sqrt{t^2 + 2\xi t} \,dt \right|
        \leq
        \int_0^\infty e^{-t} \left( t + \sqrt{2|\xi|t} \right) dt
        =
        1
        +
        \sqrt{ \frac{\pi |\xi|}{2} }.
    \end{align}
    Substituting this bound in the representation \eqref{eq:K1-1} of $K_1$, 
    we conclude
    \begin{align}
        |K_1(\xi)| 
        \leq
        e^{-\re \xi}\left( \frac{1}{|\xi|} + \sqrt{\frac{\pi}{2|\xi|}} \right).
    \end{align}
    For the remaining bounds we consider at first again $\xi>0$ and find
    \begin{align}
        \partial_\xi \frac{K_1(\xi)}{\xi}
        =
        -\int_1^\infty e^{-\xi s} s \sqrt{s^2-1} ds
    \end{align}
    where we used the integral representation \eqref{eq:K1}.
    After the same  substitution $s=t/\xi+1$ as above  
    we get
    \begin{align}
        \partial_\xi \frac{K_1(\xi)}{\xi}
        & =
        -\frac{e^{-\xi}}{\xi^3}\int_0^\infty e^{-t} (t+\xi) \sqrt{t^2+2\xi t} \, dt
    \end{align}
    Although this formula was derived for $\xi\in \R^+$, it holds
    for all $\xi\in\C\setminus\R_0^-$ by the identity theorem for holomorphic
    functions.
    In the special case $\xi>0$, using $(t+\xi) \sqrt{t^2+2\xi t}\ge t^2$,
    we obtain the second lower bound in \eqref{eq:lower-bound-K1}
    from
    \begin{align}
        -\partial_\xi \frac{K_1(\xi)}{\xi}
        \ge \frac{e^{-\xi}}{\xi^3}\int_0^\infty e^{-t} t^2 \, dt
        = 2\frac{e^{-\xi}}{\xi^3}.
        \label{eq:partial-K1-divided-xi}
    \end{align}
    Returning to the general the case $\xi\in\C\setminus\R_0^-$, 
    we get, similarly as above, the bound
    \begin{align}
        &\left| \int_0^\infty e^{-t} (t+\xi) \sqrt{t^2+2\xi t} \, dt \right|
        \leq
        \int_0^\infty e^{-t} t^{1/2} (t+2|\xi|)^{3/2} \, dt
        \cr
        & \leq \sqrt 2
        \int_0^\infty e^{-t} t^{1/2} \left(t^{3/2}+(2|\xi|)^{3/2} \right) \, dt
        =
        2^{3/2} + 2 \sqrt{\pi}|\xi|^{3/2},
        \label{eq:jensen-applied}
    \end{align}
    where we used 
    $(a+b)^{3/2}
    \leq
    \sqrt{2} (a^{3/2} + b^{3/2})$ for
    $a,b\ge 0$, which follows from convexity.
    Substituting the bound \eqref{eq:jensen-applied} in 
    \eqref{eq:partial-K1-divided-xi}, claim \eqref{eq:bound-partial-K1}
    follows.
    \fi
\end{proof}

To translate the control of $K_1(\xi)$ into control of $D(w)$, $\partial D(w)$,
and $p^-(w)$ in terms of $z$ with $w=z+i\epsilon u$, we employ the following
auxiliary inequalities about $r(w)=\sqrt{-w_\mu w^\mu}$ defined in \eqref{eq:def-r}.
\begin{lemma}\label{lem:w-and-z}
    For all space-like $z\in\R^4$, time-like $u\in\R^4$, and $\epsilon>0$, one
    has $-w_\mu w^\mu\in\operatorname{domain}(\sqrt{\cdot})=\C\setminus\R_0^-$
    for $w=z+i\epsilon u$, and the following inequalities hold true:
    \begin{align}
        \label{eq:ineq-r1}
        & 
        r(z) \vee \epsilon \sqrt{u_\mu u^\mu} 
        \le 
        \sqrt{-z_\mu z^\mu+\epsilon^2 u_\mu
        u^\mu}\le \re r(w)\le |r(w)|, 
        \\
        \label{eq:ineq-r2}
        &1\leq \frac{|w|}{|r(w)|}\le \frac{|u|}{\sqrt{u_\mu
        u^\mu}}\frac{|z|}{r(z)}.
    \end{align}
    Moreover, for a Cauchy surface $\Sigma$ fulfilling 
    \eqref{eq:bound-grad_tSigma} and $0\neq z=y-x$ with
    $x,y\in\Sigma$ one has
    \begin{align}
        \label{eq:rz-z}
        \sqrt{1-\Vmax^2}|\vec z|
         \leq r(z)
        \leq |\vec z| 
        \leq |z|\leq \sqrt{1+\Vmax^2}|\vec z|.
    \end{align}
\end{lemma}

\begin{proof}
    \ifx\arxiv\defined
        The proof is given in \cite{preprint}.
    \else
    The first claim $-w_\mu w^\mu\notin\R_0^-$ follows because $z_\mu
    z^\mu<0<\epsilon^2 u_\mu u^\mu$ implies 
    $\re(w_\mu w^\mu)=z_\mu z^\mu-\epsilon^2 u_\mu u^\mu<0$.
    The second claim \eqref{eq:ineq-r1} follows from 
    $\epsilon^2 u_\mu u^\mu \leq \epsilon^2 u_\mu
    u^\mu - z_\mu z^\mu$,
    \begin{align}
        0\le r(z)^2=-z_\mu z^\mu \le \epsilon^2 u_\mu u^\mu-z_\mu z^\mu 
        = \re(r(w)^2)\le (\re r(w))^2\le |r(w)|^2,
    \end{align}
    and $\re r(w)\ge 0$.  Using $0<-z_\nu z^\nu\le |z|^2$ and
    $0<u_\mu u^\mu\le|u|^2$ one has
    \begin{align}
         r(z)^2 \, u_\mu u^\mu \, |w|^2
        &=-z_\nu z^\nu \, u_\mu u^\mu \, (\epsilon^2 |u|^2+|z|^2) 
        \nonumber\\
        &\le
        |u|^2|z|^2(\epsilon^2 u_\nu u^\nu-z_\nu z^\nu)
        =
        |u|^2|z|^2\re(r(w)^2)
        \le
        |u|^2|z|^2|r(w)|^2,
    \end{align}
    which proves the upper bound in \eqref{eq:ineq-r2}.
    The lower bound in \eqref{eq:ineq-r2} follows from
    $|r(w)|^2=|r(w)^2|\leq |w|^2$.
    The claim \eqref{eq:rz-z} follows from the uniform Lipschitz
    continuity stated in 
    inequalities \eqref{eq:Vmax} and
    \eqref{eq:bound-grad_tSigma}.
    \fi
\end{proof}

Mainly we will use the content of Lemma~\ref{lem:K1-control} in the form of the
following corollary of it.

\begin{corollary}\label{cor:propergators}
    For all $w\in\C^4$ with $-w_\mu w^\mu \in \C\setminus\R^-_0$ and all
    $\mu=0,1,2,3$ it holds
    \begin{align}
        \label{eq:bound-D-v1}
        \left| D(w) \right| 
        &\leq 
        e^{-m\re r(w)}   
        O\left( \frac{1}{|mr(w)|^2} +\frac{1}{|mr(w)|^{3/2}}
        \right)
        ,
        \\
        \label{eq:bound-partial-D}
        \left| \frac{\partial D}{\partial w^\mu}(w) \right| 
        &\leq 
        e^{-m\re r(w)}
        O \left( \frac{|w^\mu|}{|mr(w)|^4} +\frac{|w^\mu|}{|mr(w)|^{5/2}}\right)
        .
    \end{align}
    Recall that in our units $\hbar=1=c$, $m r(w)$ is dimensionless.
\end{corollary}

\begin{proof}
    Bound \eqref{eq:bound-D-v1} is an immediate consequence of the definition of
    $D$ in \eqref{eq:def-D} and bound \eqref{eq:bound-K1}.
    Similarly, bound \eqref{eq:bound-partial-D} follows from
    the bound \eqref{eq:bound-partial-K1} and the fact
    $\partial r(w) / \partial w^\mu = - w_\mu/r(w)$.
\end{proof}

Finally, we prove the main lemmas of this appendix.

\begin{proof}[Proof of Lemma~\ref{lem:upper-bounds}]
It suffices to consider the case $\vec z\neq 0$; note that most of the
claims are trivial in the case $\vec z=0$. Furthermore, it suffices to regard
only the case $\epsilon>0$. The remaining bounds for $\epsilon=0$ 
or $\vec z=0$ can be
inferred from continuity.
Recall that
    $\cD (\Vmax)=\frac{m}{2}\sqrt{1-\Vmax^2}$.
    The basic estimates needed are
    \begin{align}
        \label{eq:basic-exp}
        m\re r(w)&\ge 2\cD  |\vec z|\vee \constl{c:eps}\epsilon\ge 
        2\cD  |\vec z| 
        \quad\text{with } \constr{c:eps}(u):=m\sqrt{u_\mu u^\mu},
        \\
        \label{eq:basic-z}
        |r(w)|^{-1} &\leq (1-\Vmax^2)^{-1/2} |\vec z|^{-1}=
        O_{\Vmax}(|\vec z|^{-1}),\\
        \label{eq:basic-epsilon}
        |r(w)|^{-1} &\leq \epsilon^{-1}(u_\mu u^\mu)^{-1/2}=O_u(\epsilon^{-1}),\\
        \frac{|w|}{|r(w)|} &\leq \frac{|u|}{\sqrt{u_\mu u^\mu}}
        \sqrt{\frac{1+\Vmax^2}{1-\Vmax^2}} =O_{u,\Vmax}(1),
        \label{eq:basic-frac}\\
        |w|&\le |z|+|u|\epsilon=O_{u,\Vmax}(|\vec z|\vee \epsilon),
        \label{eq:basic-w-upper}
        \\
        |\vec z|^ae^{-2\cD  |\vec z|}&\le O_{a,b}(|\vec z|^be^{-\cD  |\vec z|}),
        \label{eq:basic-comp}
        \\
        (|\vec z|\vee\epsilon)^a e^{-(2\cD  |\vec z|\vee\constr{c:eps}\epsilon)}&\le 
        O_{a,b,u,\Vmax}((|\vec z|\vee\epsilon)^be^{-\cD  |\vec z|})
        \label{eq:basic-comp-eps}
        \quad\mbox{ for fixed } a\ge b.
      \end{align}
    The first three inequalities \eqref{eq:basic-exp}, \eqref{eq:basic-z},
    and \eqref{eq:basic-epsilon} are consequences of 
    \eqref{eq:ineq-r1} and \eqref{eq:rz-z},
    the fourth one \eqref{eq:basic-frac} of
    \eqref{eq:ineq-r2} and \eqref{eq:rz-z}.      
    Inequality \eqref{eq:basic-w-upper}
    is a consequence of $w=z+i\epsilon u$ and \eqref{eq:rz-z}.
     Finally, the last two inequalities
     \eqref{eq:basic-comp} and \eqref{eq:basic-comp-eps} follow
     from $\sup_{x>0}x^{a-b}e^{-x}<\infty$.

    \noindent
    \emph{Bound \eqref{eq:bound-wwD}:} Using \eqref{eq:bound-D-v1} and then
    \eqref{eq:basic-exp}, 
    \eqref{eq:basic-frac}, \eqref{eq:basic-w-upper}, 
    and \eqref{eq:basic-comp-eps}, we find
    \begin{align}
      \label{eq:proof-bound-wwD}
        |w^\mu w^\nu D(w)| &\leq 
        e^{-m\re r(w)}O\left( \frac{|w|^2}{|mr(w)|^2} +\frac{|w|^{3/2}}{|mr(w)|^{3/2}}|w|^{1/2} 
        \right)
        \cr&\le
        e^{-(2\cD  |\vec z|\vee \constr{c:eps}\epsilon)}O_{u,\Vmax}\left(1 +(|\vec z|\vee\epsilon)^{1/2}
        \right)
        \le O_{u,\Vmax}\left( e^{-\cD |\vec z|}\right).
    \end{align}
    \emph{Bound \eqref{eq:bound-D}:}
    Similarly, we infer from 
    \eqref{eq:bound-D-v1}, \eqref{eq:basic-exp}, \eqref{eq:basic-z}, 
    and \eqref{eq:basic-comp} that
    \begin{align}
      \label{eq:proof-bound-D}
        |D(w)| \leq  
        e^{-2\cD |\vec z|}O_{\Vmax}\left( |\vec z|^{-2} +|\vec z|^{-3/2}\right)
        \le O_{\Vmax}\left( \frac{e^{-\cD |\vec z|}}
          {|\vec z|^2}\right).
    \end{align}
    \emph{Bound \eqref{eq:bound-r2partialD}:} Using bound 
    \eqref{eq:bound-partial-D} and then 
    \eqref{eq:basic-exp}, \eqref{eq:basic-z}, \eqref{eq:basic-frac}, 
     \eqref{eq:basic-w-upper}, and \eqref{eq:basic-comp-eps} gives
    \begin{align}
        \left| r(w)^2 \frac{\partial D(w)}{\partial w^\mu}\right|
        &\leq 
        e^{-m\re r(w)}
        O \left( \frac{|w^\mu|}{|mr(w)|}|mr(w)|^{-1} +
          \frac{|w^\mu|^{1/2}}{|mr(w)|^{1/2}}|w^\mu|^{1/2}\right)
        \cr&\le 
        e^{-(2\cD|\vec z|\vee\constr{c:eps}\epsilon)}
        O_{u,\Vmax} \left( |\vec z|^{-1} +(|\vec z|\vee\epsilon)^{1/2}\right)
        \le  O_{u,\Vmax}\left( \frac{e^{-\cD |\vec z|}}
          {|\vec z|}\right).
    \end{align}
    \emph{Bound \eqref{eq:bound-partialD}:}
    Using \eqref{eq:bound-partial-D} and then
     \eqref{eq:basic-exp}, \eqref{eq:basic-z}, \eqref{eq:basic-epsilon}, 
     \eqref{eq:basic-frac}, and \eqref{eq:basic-comp-eps}, we obtain
     \begin{align}
       &\left| \frac{\partial D}{\partial w^\mu}(w) \right| 
        \leq 
        e^{-m\re r(w)}
        O \left( \frac{|w^\mu|}{|mr(w)|^4} +\frac{|w^\mu|}{|mr(w)|^{5/2}}\right)
        \cr&\le 
        e^{-(2\cD|\vec z|\vee\constr{c:eps}\epsilon)}
        O_{u,\Vmax}
        \left( (|\vec z|\vee \epsilon )^{-3} +
          (|\vec z|\vee \epsilon )^{-3/2}\right)
        \le O_{u,\Vmax}\left( \frac{e^{-\cD |\vec z|}}
          {|\vec z|^3\vee\epsilon^3}\right)
        .
        \end{align}
    \emph{Bound \eqref{eq:bound-wpartialDw}:}
    The claim follows immediately from the bounds
    \eqref{eq:bound-partialD} and \eqref{eq:basic-w-upper}:
     \begin{align}
       \label{eq:proof-bound-wpartialDw}
       \left| w^\nu \frac{\partial D}{\partial w^\mu}(w) \right| 
        &\leq O_{u,\Vmax}(|\vec z|\vee \epsilon)
        O_{u,\Vmax}\left( \frac{e^{-\cD |\vec z|}}
          {|\vec z|^3\vee\epsilon^3}\right)
        = O_{u,\Vmax}\left( \frac{e^{-\cD |\vec z|}}
          {|\vec z|^2\vee\epsilon^2}\right)
        .
      \end{align}
    \emph{Bound \eqref{eq:bound-uD}:}
    From \eqref{eq:bound-partialD} we get
     \begin{align}
       \left| \epsilon u^\mu\frac{\partial D}{\partial w^\mu}(w) \right| 
        &\leq 
         O_{u,\Vmax}\left( \frac{\epsilon e^{-\cD |\vec z|}}
          {|\vec z|^3\vee\epsilon^3}\right)
        \le
        O_{u,\Vmax}\left( \frac{\sqrt{\epsilon} e^{-\cD |\vec z|}}
          {|\vec z|^{5/2}}\right)    
        .
        \end{align}
    \emph{Bound \eqref{eq:bound-dr2D}:}
    For $f(\xi):= K_1(\xi)/\xi$, we have the differential equation
    \begin{align}
      \label{eq:modified-modified-Bessel-diffeq}
      \xi f''(\xi)+3f'(\xi)-\xi f(x)=0
      \quad\text{ for } \xi\in\C\setminus\R^-_0.
    \end{align}
    This follows from the modified Bessel differential equation
    \begin{align}
      \xi^2 K_1''(\xi)+\xi K_1'(\xi)-(\xi^2+1) K_1(\xi)=0
    \end{align}
    or alternatively for $\re \xi>0$ from the integral representation
    \eqref{eq:K1} of $K_1$ via
    \begin{align}
      \xi f''(\xi)+3f'(\xi)-\xi f(x)
      &=\int_1^\infty e^{-\xi s}(\xi s^2-3s-\xi)\sqrt{s^2-1}\,ds
      \cr&=\int_1^\infty \frac{d}{ds}[-e^{-\xi s}(s^2-1)^{3/2}]\,ds =0
    \end{align}
    and then in the whole domain 
    $\xi\in\C\setminus\R^-_0$ via analytic continuation.
    Using $D(w)=-\constl{c:vorfaktor-D}f(m r(w))$ with the constant
    $\constr{c:vorfaktor-D}=m^3/(2\pi^2)$, cf.~\eqref{eq:def-D},
    and $\partial r(w)/\partial w^\mu=-w_\mu/r(w)$,
    we have
    \begin{align}
      \frac{\partial}{\partial w^\mu}D(w)=\constr{c:vorfaktor-D}m\frac{w_\mu}{r(w)}f'(m r(w))
    \end{align}
    and hence, using the differential 
    equation~\eqref{eq:modified-modified-Bessel-diffeq}
    in the second equality and $r(w)^2=-w^\kappa w_\kappa$ in the third equality:
    \begin{align}
      \frac{\partial}{\partial w^\nu}
      \left(r(w)^2\frac{\partial}{\partial w^\mu}D(w)\right)
      &=\constr{c:vorfaktor-D}m\left[\left(g_{\mu\nu}r(w)-\frac{w_\mu w_\nu}{r(w)}\right) f'(mr(w))
        -m w_\mu w_\nu f''(mr(w))
        \right]
        \nonumber\\
        &=\constr{c:vorfaktor-D}m
        \left[\left(g_{\mu\nu}r(w)+2\frac{w_\mu w_\nu}{r(w)}\right) f'(mr(w))
        -m w_\mu w_\nu f(mr(w))\right]
        \nonumber\\
        &=-g_{\mu\nu}w^\kappa\frac{\partial}{\partial w^\kappa}D(w)
          +2 w_\nu\frac{\partial}{\partial w^\mu}D(w)+m^2w_\mu w_\nu D(w)
          .
    \end{align}
Using the bound \eqref{eq:bound-wpartialDw}, the next to last expression 
in the bound \eqref{eq:proof-bound-wwD}, and then
\eqref{eq:basic-comp-eps}, this implies
the claimed bound \eqref{eq:bound-dr2D} as follows:
\begin{align}
  &\left|\frac{\partial}{\partial w^\nu}
  \left(r(w)^2\frac{\partial}{\partial w^\mu}D(w)\right)\right|\le 
  \cr&
  O_{u,\Vmax}\left( \frac{e^{-\cD |\vec z|}}
        {|\vec z|^2\vee\epsilon^2}\right)+
e^{-(2\cD  |\vec z|\vee\constr{c:eps}\epsilon)}O_{u,\Vmax}\left(1 +(|\vec z|\vee\epsilon)^{1/2}\right)
  \le        
  O_{u,\Vmax}\left(\frac{e^{-\cD |\vec z|}}{|\vec z|^2}
  \right).
\end{align}
\emph{Bound \eqref{eq:bound-pminus-R3}:}
The claimed bound follows from the definition \eqref{eq:pminus}
of $p^-$, using the just proven bound 
\eqref{eq:bound-partialD}, the last but one expression in the
bound \eqref{eq:proof-bound-D}, and
\eqref{eq:basic-comp}:
\begin{align}
    &\left\| p^-(w) \right\| 
    \leq
    O\left(m^{-1}\sum_{\mu=0}^4 \left|\frac{\partial D(w)}{\partial w^\mu}\right|+
        |D(w)|\right)
    \cr&
    \le  O_{u,\Vmax}\left( \frac{e^{-\cD |\vec z|}}
    {(|\vec z|\vee\epsilon)^3}\right)+
    e^{-2\cD |\vec z|}O_{\Vmax}\left( |\vec z|^{-2} +|\vec z|^{-3/2}\right)
    \le 
    O_{u,\Vmax}\left( \frac{e^{-\cD |\vec z|}}
    {|\vec z|^3}\right).
\end{align}
\end{proof}

\begin{proof}[Proof of Lemma~\ref{lem:lower-bounds}]
    Recall the definition of $p_-$ and $D$ in \eqref{eq:pminus} and
    \eqref{eq:def-D}, and the linear independence of the five matrices
    $\gamma^0,\gamma^1,\gamma^2,\gamma^3,1\in\C^{4\times 4}$.
    Using
    $f(\xi):= K_1(\xi)/\xi$ as in the previous proof,
    the bound in
    \eqref{eq:lower-bound-K1}
    involving $f'$,
    and the fact that $0 < r(z)\leq |\vec z|\le |z|$, cf. \eqref{eq:rz-z},
    we obtain the following for $z\in \R^4\setminus\{0\}$  
    with appropriate numerical constants:
    \begin{align}
        \| p_-(z) \|
        &\geq
        \constl{c:n-pm-1}
        \left(|(m)^{-1}\partial D(z)| + |D(z)|\right)
        \geq
        \constl{c:n-pm-2}
        \left|\frac{z}{r(z)} f'(mr(z))\right|
        \\
        &\geq
        \constr{c:n-pm-2}
        \left|f'(mr(z))\right|
        \geq 
        \constr{c:n-pm-2}
        \frac{e^{-mr(z)}}{(mr(z))^3}
        \geq 
        \constr{c:n-pm-3}
        \frac{e^{-m|\vec z|}}{|\vec z|^3}.
    \end{align}
\end{proof}
\fi
\end{appendix}

\paragraph{Acknowledgment} This work was partially funded by the Elite Network
of Bavaria through the Junior Research Group “Interaction between Light and
Matter”.


\ifx\arxiv\undefined

\else

\fi

\end{document}